\documentclass[11pt, twoside]{article}

\usepackage[english]{babel}

\usepackage{amssymb}
\usepackage{amsmath}
\usepackage{amsthm}
\usepackage{yhmath}
\usepackage{mathrsfs}
\usepackage{color}

\usepackage{txfonts}
\usepackage{indentfirst}
\usepackage{booktabs}
\usepackage{graphicx}
\usepackage{subfigure}
\usepackage{tikz}
\usepackage{enumerate}
\usepackage{anysize}
\usepackage{setspace}
\usepackage{float}
\usepackage{color}
\usepackage{chngcntr}
\usepackage{titlesec}
\usepackage[title,titletoc]{appendix}
\titleformat{\paragraph}[block]{\normalfont\bfseries}{\theparagraph}{1em}{} 
\titleformat{\subparagraph}[block]{\normalfont\bfseries}{\thesubparagraph}{1em}{} 

\usepackage[colorlinks=true,
linkcolor=blue,
citecolor=red,
urlcolor=magenta,
]{hyperref}

\allowdisplaybreaks

\def\red{\color{red}}

\def\sgn{{\mathop\mathrm{\,sgn\,}}}

\newtheorem{theorem}{Theorem}[section]
\newtheorem{lemma}[theorem]{Lemma}

\newtheorem{proposition}[theorem]{Proposition}
\theoremstyle{definition}
\newtheorem{remark}[theorem]{Remark}

\newtheorem{definition}[theorem]{Definition}
\renewcommand{\appendix}{\par
   \setcounter{section}{0}%
   \setcounter{subsection}{0}%
   \setcounter{subsubsection}{0}%
   \gdef\thesection{\@Alph\c@section}%
   \gdef\thesubsection{\@Alph\c@section.\@arabic\c@subsection}%
   \gdef\theHsection{\@Alph\c@section.}%
   \gdef\theHsubsection{\@Alph\c@section.\@arabic\c@subsection}%
   \csname appendixmore\endcsname
 }

\numberwithin{equation}{section}

\begin{document}

\arraycolsep=1pt

\title{\bf\Large
Limit Properties at Critical Indices of  Linear Canonical Riesz Potentials and Their Applications to Security of Multi-Image Encryption
\footnotetext{\hspace{-0.35cm} 2020 {\it Mathematics Subject
Classification}.
Primary 42B15; Secondary  42B10, 47G40,  35J05, 94A08, 68P25.
\endgraf {\it Key words and phrases.} 
linear canonical Riesz potential,  linear canonical Laplacian operator, linear canonical transform, (non-)convergence, asymmetric cascaded method, multi-image encryption.
\endgraf This project is partially supported by
the National Natural Science Foundation of China
(Grant Nos. 12431006 and 12371093),   the
China Postdoctoral Science Foundation (Grant No.
2024M760238),  the Beijing Natural Science Foundation (Grant No. 1262011), the Fundamental Research Funds for the Central Universities (Grant No. 2253200028),
and the National Research Foundation of Korea (Grant No. RS-2026-25474453).}}
\author{
Zunwei Fu,  Dachun Yang\footnote{Corresponding
author, E-mail: \texttt{dcyang@bnu.edu.cn}/{\red \today}/Final version.}
\  \ and Shuhui Yang}
\date{}
\maketitle

\vspace{-0.8cm}

\begin{center}
\begin{minipage}{13.5cm}
{\small{\textbf{Abstract}}\quad}
In this article we introduce the linear canonical Riesz potential (for short, LCRP)
and give its symbol in terms of linear canonical transforms.
Driven by image processing, we establish the convergence/divergence of these LCRPs for different kinds of functions. Concretely,  for grating functions, we prove that their classical Riesz potentials diverge, whereas their LCRP converge due to the key role of chirp functions. For the characteristic function ${\mathbf 1}_P$ of a convex polygon $P$, we show that the limit of its Riesz potential at any non-boundary point $\boldsymbol{x}$ equals ${\mathbf 1}_P(\boldsymbol{x})$, but its limit at the boundaries differ from ${\mathbf 1}_P$, while it is known that, for any Schwartz function $f$, the limit of its Riesz potential at any point $\boldsymbol{x}$ always equals $f(\boldsymbol{x})$.   Based on these and the inverse operator of the LCRP (namely the linear canonical Laplacian operator), we propose an asymmetric cascaded  LCRP method for the multi-image encryption and create an efficient and secure cryptosystem. Systematic security evaluations, including sensitivity, statistical, noise attack, and occlusion attack analyses, demonstrate its robustness and its security. Even for a single image, the proposed method is more efficient than the known encryption approach based on the fractional Riesz potential. The novelty of these results lies in that the convergence and the divergence of LCRTs at the critical indices, respectively, for ``good" Schwartz functions and for ``bad" discrete image functions essentially affect the security of image encryption and decryption. 
\end{minipage}
\end{center}

\vspace{0.2cm}

\tableofcontents

\vspace{0.2cm}
\section{Introduction} 
It is well known that the Fourier transform and the Riesz potential serve as powerful tools for analyzing stationary signals. However, a fundamental limitation arises when applying them to nonstationary signals, their global frequency representation fails to capture localized or time-varying structures. In numerous real-world applications, signals exhibit time-varying spectral content and hence need analytical tools capable of revealing the localized frequency behavior over time. A widely used model for such nonstationary behavior is the following chirp signal, characterized by a linear variation in its instantaneous frequency over time, which may either increase (up-chirp) or decrease (down-chirp) (see, for example, Jiang et al. \cite{cjll23,cjll211,cjll212,jlcl25,lcj22,ljcl251} and Fu et al. \cite{ fglwy23,flyy24,yfyll25}).

\begin{definition}\label{def-dancp}
Let  $\boldsymbol{\alpha}
:=(\alpha_1,\ldots,\alpha_n)\in  \mathbb{R}^n$. The \emph{chirp function}
$e_{\boldsymbol{\alpha}}$ is defined by setting, for any
$\boldsymbol{x}:=(x_1, \ldots,x_n)\in \mathbb{R}^n$,
\begin{equation}\label{1cf}
e_{\boldsymbol{\alpha}}(\boldsymbol{x}):=e^{2\pi
i\sum\limits_{i=1}^{n}\alpha_ix_i^2}.
\end{equation}  
\end{definition}

 In what follows, $\mathscr{S}(\mathbb{R}^n)$ denotes the \textit{Schwartz space} on $\mathbb{R}^n$ equipped with its standard topology (defined by a countable family of norms), and $\mathscr{S}'(\mathbb{R}^n)$ its \textit{dual} (the \textit{space of all Schwarz distributions}) equipped with the weak-$\ast$ topology. Based on the chirp function, we consider the following  \emph{chirp Poisson equation}. For any $\boldsymbol{\alpha} := (\alpha_1, \ldots, \alpha_n) \in \mathbb{R}^n$ and $u \in \mathscr{S}(\mathbb{R}^n)$,
\begin{equation}\label{1cpe}
\Delta_{\boldsymbol{\alpha}} u := e_{-\boldsymbol{\alpha}} \Delta (e_{\boldsymbol{\alpha}} u)= f,  
\end{equation}
where $f$ is a known function, $\Delta:=\sum_{i=1}^{n}\frac{\partial^2}{\partial x_i^2}$ is the classical \textit{Laplacian operator}, and $e_{\boldsymbol{\alpha}}$ is as in \eqref{1cf}. 
This equation can be solved by reducing it to a standard Poisson equation via the substitution $w := e_{\boldsymbol{\alpha}} u$ and $g := e_{\boldsymbol{\alpha}} f$ and using the fundamental solution of the Laplacian equation (see, for instance, \cite[Chapter 2, p. 22]{e10}).  When $n>2$, the solution $u$ of \eqref{1cpe} is given by setting, for any $\boldsymbol{x}\in \mathbb{R}^n$,
\begin{equation}\label{e1.2x}
u(\boldsymbol{x}) = -\frac{1}{(n-2)\omega_n} e_{-\boldsymbol{\alpha}}(\boldsymbol{x}) \int_{\mathbb{R}^n} \frac{  f(\boldsymbol{y})}{ |\boldsymbol{x} - \boldsymbol{y}|^{n-2}}e_{\boldsymbol{\alpha}}(\boldsymbol{y})\, d\boldsymbol{y},
\end{equation}
where $\omega_n$ denotes the area of the unit sphere in $\mathbb{R}^n$.

Replacing the exponent 2 in \eqref{e1.2x} by a real parameter $\beta \in (0, n)$
yields the following fractional Riesz potential associated with chirp functions.
Recall that the \textit{Gamma function  $\Gamma$} is defined by setting,
for any $\beta\in (0,\infty)$,
$$\Gamma(\beta): = \int_0^{\infty} t^{\beta-1} e^{-t}\,dt.$$

\begin{definition}\label{frp1}
Let $\beta\in(0,n)$ and $\boldsymbol{\alpha}\in  \mathbb{R}^n$.
The \emph{fractional Riesz potential}
$I^{\boldsymbol{\alpha}}_\beta $
is defined by setting, for any $f\in \mathscr{S}(\mathbb{R}^n)$  and
$\boldsymbol{x}\in \mathbb{R}^n$,
\begin{equation}\label{eq-rp}
I^{\boldsymbol{\alpha}}_\beta(f)(\boldsymbol{x}):=
\frac{1}{\gamma(\beta)}e_{-\boldsymbol{\alpha}}(\boldsymbol{x})
\int_{\mathbb{R}^n}\frac{f(\boldsymbol{y})}{|\boldsymbol{x}
-\boldsymbol{y}|^{n-\beta}}e_{\boldsymbol{\alpha}}
(\boldsymbol{y})\,d\boldsymbol{y},
\end{equation}
where $e_{\boldsymbol{\alpha}}$ is as in \eqref{1cf} and $\gamma(\beta):=\pi ^{\frac{n}{2}}2^\beta
\frac{\Gamma(\frac{\beta}{2})}{\Gamma(\frac{n-\beta}{2})}.$
\end{definition}

\begin{remark}
\begin{enumerate}[(i)]
\item If $\boldsymbol{\alpha}:=(\frac{\cot\zeta_1}{2},\ldots,\frac{\cot\zeta_n}{2})$ with each
$\zeta_k\in (0,2\pi)$ and $\zeta_k \neq \pi$, then $I^{\boldsymbol{\alpha}}_\beta$ in this case coincides with the \textit{fractional Riesz potential} in  \cite[Definition 1.4]{flyy24}. 

\item If $\boldsymbol{\alpha}:=(0,\ldots,0)$,
then $I^{\boldsymbol{\alpha}}_\beta$ in this case coincides with the classical Riesz potential in  \cite{r1949,s1970}. 
\end{enumerate}
\end{remark}

Fu et al. \cite{flyy24} applied the fractional Riesz potential to image encryption. The experiments in \cite{flyy24} demonstrate that the encryption method based on the fractional Riesz potential significantly enhances the system security with its effectiveness validated in encrypting both grayscale and color images. Next, we recall the two-parameter counterpart of the single-parameter chirp function given in Definition \ref{def-dancp}.
\begin{definition}
Let $\boldsymbol{a}:=(a_1,\ldots,a_n),
\boldsymbol{b}:=(b_1,\ldots,b_n)\in\mathbb R^n$ with each $b_j\neq0$.
The \emph{chirp function}
$e_{\boldsymbol{a},\boldsymbol{b}}$  is defined by setting,  for any $\boldsymbol{x}:=(x_1,\ldots,x_n)\in \mathbb{R}^n$,
\begin{equation}\label{1eq}
{e}_{\boldsymbol{a},\boldsymbol{b}}( \boldsymbol x
) :=e^{\pi i\sum\limits_{j=1}^n\frac{a_j}{b_j}x_j^2}.
\end{equation}
\end{definition}

A \textit{natural question} is that, if we introduce the following linear canonical Riesz potential by replacing the single-parameter chirp function in \eqref{eq-rp} with the two-parameter chirp function in \eqref{1eq}, whether this new potential can characterize and analyze non-stationary signals more effectively. The main target of this article is to give a positive answer to this question.
 
\begin{definition}\label{def-lcrp}
Let  $\beta\in(0,n)$ and
$\boldsymbol{a}:=(a_1,\ldots,a_n),
\boldsymbol{b}:=(b_1,\ldots,b_n)\in\mathbb R^n$  with each $b_j\neq0$.
The  \emph{linear canonical Riesz potential} (for short, 
LCRP) $I^{\boldsymbol{a},\boldsymbol{b}}_\beta$ is defined 
by setting, for any $f\in \mathscr{S}(\mathbb{R}^n)$ 
and $\boldsymbol{x}\in \mathbb{R}^n$,
\begin{equation*} 
I^{\boldsymbol{a},\boldsymbol{b}}_\beta(f)(\boldsymbol{x}):=
\frac{1}{\gamma(\beta)}{e}_{-\boldsymbol a,\boldsymbol b} 
(\boldsymbol{x})\int_{\mathbb{R}^n}\frac{f(\boldsymbol{y})}
{|\boldsymbol{x}-\boldsymbol{y}|^{n-\beta}}
{e}_{\boldsymbol a,\boldsymbol b}
(\boldsymbol{y})\,d\boldsymbol{y}
\end{equation*}
with $\gamma(\beta)$ as in \eqref{eq-rp},
${e}_{-\boldsymbol a,\boldsymbol b}$  as  in \eqref{1eq} with $\boldsymbol{a}$ 
replaced by $-\boldsymbol a$, and ${e}_{\boldsymbol a,\boldsymbol b}$
as in \eqref{1eq}.
\end{definition}

\begin{remark}
\begin{enumerate}[(i)]
\item Let the notation be the same as in Definition \ref{def-lcrp}.
If let $\boldsymbol{\alpha}:=(\alpha_1,\ldots,\alpha_n)$ with each $\alpha_j:=\frac{a_j}{2b_j} $, then the LCRP 
coincides with the fractional Riesz potential $I^{\boldsymbol{\alpha}}_{\beta}$ in Definition \ref{frp1}. It is known that the fractional Fourier transform completely matches the fractional Riesz potential and  both together play a crucial role in image processing. Observe that two-parameter chirp functions simultaneously appear in the definitions of LCRPs and linear canonical transforms. Thus, it is reasonable to expect that the linear canonical transform may perfectly match the LCRP, which is proved to be true in Section \ref{sec2}.

\item If $n=3$, $\beta=2$, $\boldsymbol{a}:=(0, 0, 0)$, and 
$\boldsymbol{b}:=(1,1,1)$, then the LCRP in this case coincides with  the classical \textit{Newton potential} in \cite[(1.1)]{ks11}.
\end{enumerate}
\end{remark}

In what follows we give the symbol of LCRPs in terms of linear canonical transforms. Driven by image processing, we establish the convergence/divergence of these LCRPs for different kinds of functions. Concretely,  for grating functions, we prove that their classical Riesz potentials diverge, whereas their LCRPs converge due to the key role of chirp functions. For the characteristic function ${\mathbf 1}_P$ of a convex polygon $P$, we show that the limit of its Riesz potential at any non-boundary point $\boldsymbol{x}$ equals ${\mathbf 1}_P(\boldsymbol{x})$, but its limit at the boundaries differ from ${\mathbf 1}_P$, while it is known that, for any Schwartz function $f$, the limit of its Riesz potential at any point $\boldsymbol{x}$ always equals $f(\boldsymbol{x})$.   Based on these and the inverse operator of the LCRP (namely the linear canonical Laplacian operator), we propose an asymmetric cascaded  LCRP method for the multi-image encryption and create an efficient and secure cryptosystem. Systematic security evaluations, including sensitivity, statistical, noise attack, and occlusion attack analyses, demonstrate its robustness and its security. Even for a single image, the proposed method is more efficient than the known encryption approach based on the fractional Riesz potential. The novelty of these results lies in that the convergence and the divergence of LCRTs at the critical indices, respectively, for ``good" Schwartz functions and for ``bad" discrete image functions essentially affect the security of image encryption and decryption. 

The remainder of this article is organized as follows.

Section \ref{sec2} consists of two subsections.
In Subsection \ref{sec2.1} we  give the symbol of LCRPs in terms of linear canonical transforms (Theorem \ref{thm-frp}).
Building upon this, we further introduce the linear canonical Laplacian operator (for short, LCLO) in  Definition \ref{def-lclo}, which serves as the inverse of the LCRP (Proposition \ref{pro-lclo}).
Subsequently, we conduct a rigorous analysis on the limit behavior at the critical cases of the parameters $\beta$, $\boldsymbol{a}$, and $\boldsymbol{b}$ appearing in the LCRP $I^{\boldsymbol{a},\boldsymbol{b}}_\beta$. To be precise, in Theorem \ref{p-sharp} we  prove that, for any Schwartz function, the limits at the critical points of the parameters $\beta$, $\boldsymbol{a}$, and $\boldsymbol{b}$ in $I^{\boldsymbol{a},\boldsymbol{b}}_\beta$ pointwise converge, while
in Theorem \ref{rem-imp-1}, we consider the parameters $\boldsymbol{a}:=(a_1, a_2), \boldsymbol{b}:=(b_1, b_2)\in\mathbb{R}^2$ with each $b_j\neq0$, and, if letting $k_j:= \frac{a_j}{b_j}$ for $j\in\{1,2\}$, we show that, for a grating function $f$ on $\mathbb{R}^2$, its classical Riesz potential diverges and, assuming   $k_1\neq 0\neq k_2$, then its LCRP $I_1^{\boldsymbol{a},\boldsymbol{b}} f $ converges.
Furthermore, in Theorem \ref{thm-polygon}, we demonstrate that, for $f := \mathbf{1}_{P}$ with $P$ being an arbitrary open convex polygon on $\mathbb{R}^2$, 
$$\lim_{\beta\to 0^+}I_\beta f(\boldsymbol{x})= f(\boldsymbol{x})$$ 
for any interior or exterior point  $\boldsymbol{x}$ of $P$, while the limits at the boundary heavily depend on the local geometric structure: having limit $\frac{1}{2}$ for points on straight edges and having limit $\frac{\alpha}{2\pi}$ for corner vertices with an internal angle $\alpha$.

These results demonstrate two distinct analytic mechanisms (Remarks \ref{zyr-1} and \ref{zyr-2}). First, the classical Riesz potentials strictly diverge for bounded functions that do not decay at infinity (such as grating functions). In contrast, their LCRPs  converge because their chirp functions introduce highly oscillatory factors that regularize the behavior at infinity. This highlights a significant transition in image processing from the ``incapability'' of the classical Riesz potential to the ``capability'' of the LCRP when handling global image backgrounds. Second, for discontinuous functions (such as the characteristic functions of open convex polygons), the pointwise limiting behavior of the classical Riesz potential $I_\beta$ exhibits jump discontinuities and severe geometric dependencies at the boundaries (for example,   straight edges and sharp corners) as the parameter $\beta \to 0^+$. This indicates that, in image encryption, if there is a slight deviation in the parameter $\beta$ of $I_\beta$, the original information at the image edges and corners becomes extremely difficult to recover, thereby significantly enhancing the security and the key sensitivity of the ciphertext.

To prove Theorem \ref{rem-imp-1}, we transform the divergent integral into a convergent Fresnel-type integral by completing the square in the phase function and applying the Lebesgue dominated convergence theorem. Meanwhile, to prove Theorem \ref{thm-polygon}, we employ the asymptotic expansion of the Gamma function through the normalization constant and the integral computation of the Riesz potential, and precisely compute the limit values across different spatial regions by leveraging the radial symmetry of the potential kernel alongside topological localization and sectorial bounding techniques via the squeeze theorem.
Furthermore, in Remark \ref{the-spaces}, we generalize these phenomena in Theorems \ref{rem-imp-1} and \ref{thm-polygon} to two fundamental spaces in image processing: bounded non-decaying function spaces and piecewise smooth function spaces. This generalization uncovers the intrinsic analytic mechanisms, specifically oscillatory regularization and limit discontinuity, that explain why the LCRP succeeds in encrypting images with global backgrounds where classical methods fail and how it achieves high sensitivity at image edges.

In Subsection \ref{sec2.3}, applying Theorem \ref{thm-frp} we present numerical simulations of the LCRP and the LCLO. Through systematic simulations on images with different parameter sets (the parameters  $\boldsymbol{A}$ and $\beta$ of the LCRP $I^{\boldsymbol{A}}_\beta$  and  the parameters  $\boldsymbol{A}$ and $\gamma$ of the LCLO $\Delta_{\gamma}^{\boldsymbol{A}}$),  we present the resulting  2D grayscale images, 3D color images, and their LCT domain amplitude distributions (Figures \ref{FIG3.1}--\ref{FIG3.5}).  These experiments demonstrate that the LCRP and the LCLO are mutually inverse under specific conditions and that the parameter $\beta$ of the LCRP $I^{\boldsymbol{A}}_\beta$ effectively governs the image sparsity in the spatial domain and the amplitude modulation in the LCT domain.

In Section \ref{sec4}, based on Theorems \ref{thm-frp}, \ref{p-sharp}, \ref{rem-imp-1},  and \ref{thm-polygon}, Proposition \ref{pro-lclo}, and Remarks \ref{zyr-1}, \ref{zyr-2}, and \ref{the-spaces}, we propose a novel asymmetric cascaded LCRP multi-image encryption method.
We elaborate on the encryption procedure, which includes a dual-phase mask generation algorithm based on image separation, nonlinear phase modulation, and a cascaded LCRP process. The decryption side utilizes the LCLO for recovery.
Numerical results (Figures \ref{fig.421}--\ref{fig.424}) confirm that the method exhibits high security with respect to keys (LCRT matrices $A_j$ and orders $\beta_j$). A systematic four-pronged security evaluation, encompassing quantitative key sensitivity (Figures  \ref{fig. 11}--\ref{fig. bjmse}), statistical analysis (Tables  \ref{tab.421}--\ref{tab.424} and Figures  \ref{fig.3ys}--\ref{fig. zhifangtu}), noise attack resistance (Figure  \ref{fig. jxzs}), and occlusion attack robustness (Figure  \ref{fig. oa}), validates the exceptional security and the robustness of the method proposed above.

Finally, Section \ref{sec5} concludes this article with a detailed summary.

We end this section by making some notational conventions. Let $\mathbb{N}:=\{1,2,\ldots\}$ and $\mathbb{Z}_+$
denote the set of all non-negative integers. We use $C$ to denote a positive constant,
which is independent of the main parameters involved,
whose value may vary from line to line. The \emph{notation}
$g\lesssim h$ means $g\leq Ch$. If $g\lesssim h$ and $h\lesssim g$, we then write $g\sim h$. For any set $\Omega \subset \mathbb{R}^n$, we use $\mathring{\Omega}$, $\overline{\Omega}$, and $\partial \Omega$ to denote its interior, closure, and boundary, respectively.  Furthermore, throughout this article, unless otherwise specified, we consistently assume that $\boldsymbol{A} := (A_{1}, \ldots, A_{n})$ with each $A_{k} := \begin{bmatrix} a_{k} & b_{k} \\ c_{k} & d_{k} \end{bmatrix} \in M_{2\times2}(\mathbb{R})$ (where $M_{2\times2}(\mathbb{R})$ denotes the set of all $2 \times 2$ real matrices with determinant 1). For convenience, we also denote the vectors formed by these entries as $\boldsymbol{a} := (a_1, \ldots, a_n)$, $\boldsymbol{b} := (b_1, \ldots, b_n)$, $\boldsymbol{c} := (c_1, \ldots, c_n)$, and $\boldsymbol{d} := (d_1, \ldots, d_n)$. We denote the \textit{origin} of $\mathbb{R}^n$ by $\boldsymbol{0}$. The limit ${\beta\to 0^+}$ means that $\beta>0$ and $\beta\to 0$. Finally, in all proofs we consistently retain the notation introduced in the original theorem (or related statement). 

\section[Riesz Potentials and Laplacian Operators Meet Linear Canonical Transforms]{Riesz Potentials and  Laplacian Operators Meet  \\ Linear Canonical Transforms}\label{sec2}

The Fourier transform, refined over more than two centuries of development, has become a cornerstone in scientific and engineering disciplines. It underpins the theoretical foundations of signal processing and analysis, and remains widely used for analyzing stationary signals (for recent applications in graph signal processing, see Plonka et al. \cite{ppst23} and Sun et al. \cite{ccs23,ccls23}). As a generalization of the Fourier transform, the fractional Fourier transform introduces additional degrees of freedom, enhancing its flexibility and effectiveness in handling non-stationary signals (see, for example, Abbas et al.  \cite{ asf2017} and Zayed et al. \cite{krz21,krz22,krz23, n1980,z1998, z2018,z2019,z24}).

Stein, in his monograph \cite{s1993}, succinctly summarized the scope of harmonic analysis into three central categories: maximal averages, singular integrals, and oscillatory integrals. The Fourier transform serves as a canonical example of an oscillatory integral, while the Hilbert transform stands as a pivotal example of a singular integral. As a quintessential manifestation of singular integrals, the Hilbert transform plays a significant role in signal processing and has been widely applied in areas such as modulation theory, edge detection, and filter design (see, for example, \cite{b19921, g1946, ppst23, xwx09}). In recent years, the integration of the fractional Fourier transform with the Hilbert transform has led to the introduction of the fractional Hilbert transform, with related research yielding a series of advancements (see, for example, \cite{cfgw21, z1998}). On the other hand, the Riesz transform, as a natural extension of the Hilbert transform to the $n$-dimensional Euclidean space, has undergone corresponding developments. In \cite{fglwy23}, Fu et al. combined the fractional Fourier transform with the Riesz transform to propose a fractional Riesz transform, which has subsequently been applied to edge detection. Compared with the classical Riesz transform, the fractional Riesz transform not only retains the capability to capture global information but also enables the extraction of local features along arbitrary directions.

Both being integral operators, the Riesz transform and the Riesz potential offer distinct advantages that stem from their specific types, singular integrals versus fractional integrals. For instance, the Riesz transform, being inherently directional, is particularly suitable for applications like edge detection, where feature orientation is crucial.  In contrast, the Riesz potential incorporates an additional parameter, offering greater degrees of freedom that can be exploited in more complex processing tasks like image encryption. 
The profound connection between the classical Riesz potential and the Fourier transform is well established. Specifically, for any $\beta \in \mathbb{R}$, the \textit{Riesz potential $I_\beta$} is defined in the Fourier domain for functions $f$ in the Schwartz class $\mathscr{S}(\mathbb{R}^n)$ by the multiplicative relation: for any $\boldsymbol{\xi}\in\mathbb{R}^n$,
\begin{equation}\label{jdsz}
 \left( I_\beta f \right)^{\wedge} (\boldsymbol{\xi}) = (2\pi|\boldsymbol{\xi}|)^{-\beta} \widehat{f}(\boldsymbol{\xi}),\end{equation}
 where the \textit{Fourier transform $\widehat{f}$} of $f$ is defined as
 $$\widehat{f}(\boldsymbol{\xi}) := \mathcal{F}(f)(\boldsymbol{\xi}) := \int_{\mathbb{R}^{n}} f(\boldsymbol{y}) e^{-2\pi i \boldsymbol{\xi} \cdot \boldsymbol{y}}  d\boldsymbol{y}.$$
This elegant relation in the Fourier domain unifies the fractional integral operator (when $\beta > 0$), the identity operator ($\beta = 0$), and the differential operator ($\beta< 0$, namely the fractional Laplacian operator). For a comprehensive treatment of this operator and its properties, we refer to the original work of Riesz \cite{r1949} and the comprehensive monograph by Stein \cite[Chapter V]{s1970}.
A similar relationship holds between the fractional Riesz potential and the fractional Fourier transform, as established in \cite[Theorem 2.6]{flyy24}. This foundational principle naturally leads to the question which transform shares an analogous relationship with the newly introduced LCRP.

To answer this, we turn to the linear canonical transform (LCT), providing a more comprehensive mathematical tool for analyzing and processing non-stationary signals as a further extension of the fractional Fourier transform. The LCT is also referred to as the ABCD matrix transform \cite{b1996}, the generalized Fresnel transform \cite{ja1996}, the quadratic phase system \cite{b1978}, and the extended fractional Fourier transform \cite{hll1997}. The multidimensional LCT was introduced in \cite{dp1999} as follows.

\begin{definition}\label{nDLCT}
 For any $f \in \mathscr{S}(\mathbb{R}^{n})$, its \emph{linear canonical transform} (for short, LCT) $\mathscr{L}_{\boldsymbol{A}}(f)$ is defined by setting, for any $\boldsymbol{u} \in \mathbb{R}^{n}$,
$$
\mathscr{L}_{\boldsymbol{A}}(f)(\boldsymbol{u}) := \int_{\mathbb{R}^{n}} f(\boldsymbol{x}) K_{\boldsymbol{A}}(\boldsymbol{x}, \boldsymbol{u})  d\boldsymbol{x},
$$
where, for any $\boldsymbol{x} := (x_{1}, \ldots, x_{n}), \boldsymbol{u} := (u_{1}, \ldots, u_{n}) \in \mathbb{R}^{n}$,
$$
K_{\boldsymbol{A}}(\boldsymbol{x}, \boldsymbol{u}) := \prod_{k=1}^{n} K_{A_{k}}(x_{k}, u_{k})
$$
and, for any $k \in \{1, \ldots, n\}$,
$$
K_{A_{k}}(x_{k}, u_{k}) :=
\begin{cases}
C_{A_{k}} e_{a_{k},b_{k}}(x_{k}) e_{-b_{k}}(x_{k}, u_{k}) e_{d_{k},b_{k}}(u_{k}) & \text{if } b_{k} \neq 0, \\
\sqrt{d_{k}} e^{{\pi i}c_{k}d_{k}u_{k}^{2}} \delta(x_{k} - d_{k}u_{k}) & \text{if } b_{k} = 0
\end{cases}
$$
with $C_{A_{k}} := \sqrt{\dfrac{1}{i  b_{k}}}$, $e_{a_{k},b_{k}}(x_{k}) := e^{\pi i\frac{a_{k}}{b_{k}}x_{k}^{2}}$, $e_{-b_{k}}(u_{k}, x_{k}) := e^{-2\pi i\frac{u_{k}x_{k}}{b_{k}}}$, $e_{d_{k},b_{k}}(u_{k}) := e^{\pi i\frac{d_{k}}{b_{k}}u_{k}^{2}}$, and 
$\delta$ being the \emph{Dirac measure} at the origin.
\end{definition}

Next, we recall the concept of the LCT of Schwartz distributions, which is exactly  \cite[Definition 2.7]{yfyll25}.

\begin{definition} Assume $b_k \neq 0$ for any $k \in \{1, \ldots, n\}$. 
For any Schwartz distribution $u\in \mathscr{S}'(\mathbb{R}^n)$, its LCT $\mathscr{L}_{\boldsymbol{A}}u$  is defined by setting, for any $f \in \mathscr{S}(\mathbb{R}^n)$,
$$\langle\mathscr{L}_{\boldsymbol{A}}u,f\rangle := \langle u, \mathscr{L}_{\boldsymbol{\widetilde{A}}} f\rangle, 
$$ where $\boldsymbol{\widetilde{A}}:=(\widetilde{A}_1,\ldots,\widetilde{A}_n)$ with each $\widetilde{A}_k:=\begin{bmatrix} d_k & b_k \\ c_k & a_k \end{bmatrix}$.
\end{definition}

\begin{remark}(see \cite{yfyll25})\label{rem-sd-lct}
Assume $b_k \neq 0$ for any $k \in \{1, \ldots, n\}$.
Define $C_{\boldsymbol{A}}:=\prod^n_{k=1}C_{A_k}$,
where $C_{A_k}$ is the same as in Definition \ref{nDLCT}. Let $f\in{\mathscr{S}'(\mathbb{R}^n)}$.
Its LCT $\mathscr{L}_{\boldsymbol{A}}(f)$ is
precisely the Fourier transform of the signal
$g:={e}_{\boldsymbol{a},
	\boldsymbol{b}}f$
multiplied by a chirp function, where ${e}_{\boldsymbol{a},
	\boldsymbol{b}}$ is the same as in \eqref{1eq}.
In other words, the LCT $\mathscr{L}_{\boldsymbol{A}}(f)$ 
can be decomposed into that, for any
$\boldsymbol x\in \mathbb{R}^n$,
\begin{equation*} 
	\mathscr{L}_{\boldsymbol A}(f)(\boldsymbol x)=
	C_{\boldsymbol A}
	{e}_{\boldsymbol{d},\boldsymbol{b}}
	(\boldsymbol x)\mathscr{F}\left({e}_{\boldsymbol{a},
		\boldsymbol{b}}f \right)
	(\widetilde{\boldsymbol{x}}),
\end{equation*}
holds in  $\mathscr{S}'(\mathbb{R}^n)$,
where ${e}_{\boldsymbol{d},
	\boldsymbol{b}}$ is the same as in \eqref{1eq}
with $\boldsymbol{a}$  
replaced  by   $\boldsymbol{d}$
and where $\boldsymbol{\widetilde x}:=
\frac{\boldsymbol x}{\boldsymbol b}:= (\frac{x_1}{b_1},
\ldots, \frac{x_n}{b_n})$.
\end{remark}

Now, we recall the definition of the inverse LCT, which is exactly \cite[Definition 39]{z24}.

\begin{definition}\label{def-ia}
Define the inverse matrix vector $\boldsymbol{A}^{-1} := (A_1^{-1}, \ldots, A_n^{-1})$ with each $A_k^{-1} := \begin{bmatrix} d_k & -b_k \\ -c_k & a_k \end{bmatrix}$.
For any $f \in \mathscr{S}(\mathbb{R}^n)$, its \emph{multidimensional inverse} LCT $(\mathscr{L}_{\boldsymbol{A}})^{-1}(f)$ is defined by setting, for any $\boldsymbol{x} \in \mathbb{R}^n$, 
$$
(\mathscr{L}_{\boldsymbol{A}})^{-1}(f)(\boldsymbol{x}) := \mathscr{L}_{\boldsymbol{A}^{-1}}(f)(\boldsymbol{x}).
$$
\end{definition}

The following lemma gives the inverse LCT theorem (see \cite[p. 9]{yfyll25}).

\begin{lemma}
For any $f\in\mathscr{S}'(\mathbb{R}^n)$ and
$\boldsymbol{x}\in \mathbb{R}^n$,
$({\mathscr{L}_{\boldsymbol{A}}})^{-1}
{\mathscr{L}_{\boldsymbol{A}}}(f)
(\boldsymbol x)=f
(\boldsymbol x).$
\end{lemma}

The LCT has been widely applied in various areas including filter design, sampling, image processing, and pattern recognition (see, for example, \cite{ddp24,   hs22,hs222,hkos15, kock08,lcjl22,  pd2000, pd2002,sko98,zl18,yz24}).  Recently, Yang et al. \cite{yfyll25} combined the LCT with the Riesz transform to propose the linear canonical Riesz transform, demonstrating its superiority over the fractional Riesz transform for refined image edge detection.

The remainder of this section is structured as follows. In Subsection \ref{sec2.1}, we first give the symbol of LCRPs in terms of linear canonical transforms (Theorem \ref{thm-frp}). Based on this, we further introduce the inverse operator of the LCRP, namely the LCLO (Definition \ref{def-lclo}). Subsequently, we analyze the limit behaviors at critical indices (Theorems \ref{p-sharp}, \ref{rem-imp-1}, and \ref{thm-polygon}) and generalize these results in Remark \ref{the-spaces} to reveal the intrinsic analytic mechanisms on image backgrounds, edges, and corners. Finally, Subsection \ref{sec2.3} presents numerical simulations of the LCRP and the LCLO applied to images.

\subsection{Linear Canonical Riesz Potentials and Linear Canonical Laplacian Operators}\label{sec2.1}

We next give  LCRPs symbol in terms of linear canonical transforms.  

\begin{theorem}\label{thm-frp}
Let $\beta \in (0, n)$ and assume $b_k \neq 0$ for any $k \in \{1, \dots, n\}$.
Then, for any $f \in \mathscr{S}(\mathbb{R}^n)$ and $\boldsymbol{u} := (u_1, \dots, u_n) \in \mathbb{R}^n$,
$$
\mathscr{L}_{\boldsymbol{A}} \left(I_\beta^{\boldsymbol{a}, \boldsymbol{b}} f \right)(\boldsymbol{u}) = (2\pi)^{-\beta} |\widetilde{\boldsymbol{u}}|^{-\beta} \mathscr{L}_{\boldsymbol{A}}(f)(\boldsymbol{u})
$$
holds in $\mathscr{S}'(\mathbb{R}^n)$, where $\widetilde{\boldsymbol{u}} := (\widetilde{u}_1, \dots, \widetilde{u}_n) = (\frac{u_1}{b_1}, \dots, \frac{u_n}{b_n})$.
The function $(2\pi)^{-\beta} |\widetilde{\boldsymbol{u}}|^{-\beta}$ is referred to as the \emph{symbol} of the \emph{LCRP}.
\end{theorem}

\begin{proof}
Let ${e}_{\boldsymbol d,\boldsymbol b}$  be the same as 
in \eqref{1eq} with $\boldsymbol{a}$  replaced  by   
$\boldsymbol d$ and   ${e}_{\boldsymbol a,\boldsymbol b}$
be the same as in \eqref{1eq}. 
From  Remark \ref{rem-sd-lct} and
Definition \ref{def-lcrp},  we infer that, for any  $f\in \mathscr{S}(\mathbb{R}^n)$ and $\boldsymbol{u} \in
\mathbb{R}^n$,
\begin{align*}
\mathscr{L}_{\boldsymbol{A}}
\left(I_\beta^{\boldsymbol{a},\boldsymbol{b}}
f\right)(\boldsymbol{u})=&C_{\boldsymbol A}
{e}_{\boldsymbol{d},\boldsymbol{b}}
(\boldsymbol u)\mathscr{F}\left({e}_{\boldsymbol{a},
\boldsymbol{b}} \, I_\beta^{\boldsymbol{a},\boldsymbol{b}}f\right)
(\widetilde{\boldsymbol{u}}) \\
=&\  \frac{1}{\gamma(\beta)}C_{\boldsymbol A}
{e}_{\boldsymbol{d},\boldsymbol{b}}
(\boldsymbol u)\mathscr{F}
\left({e}_{\boldsymbol{a},
\boldsymbol{b}}f* \frac{1}{|\boldsymbol{\cdot}|
^{n-\beta}} \right )(\widetilde{\boldsymbol{u}})
\end{align*}
in  $\mathscr{S}'(\mathbb{R}^n)$, where $\widetilde{\boldsymbol{u}} :=(\frac{u_1}{b_1},\ldots,\frac{u_n}{b_n}) $.
By this, together with \cite[Proposition 2.3.22(11) and Theorem 2.4.6]{g20141} and Remark \ref{rem-sd-lct},  we further conclude that,  for any $f\in \mathscr{S}(\mathbb{R}^n)$
and $\boldsymbol{u}\in\mathbb{R}^n$,
\begin{align*}
\mathscr{L}_{\boldsymbol{A}}
\left(I_\beta^{\boldsymbol{a},\boldsymbol{b}}
f\right)(\boldsymbol{u})=&\
\frac{1}{\gamma(\beta)}C_{\boldsymbol A}
{e}_{\boldsymbol{d},\boldsymbol{b}}
(\boldsymbol u)\mathcal{F}({e}_{\boldsymbol{a},
\boldsymbol{b}}f)
(\widetilde{\boldsymbol{u}})
\mathscr{F} \left(\frac{1}{|\boldsymbol{\cdot}|^{n-\beta}}\right)
(\widetilde{\boldsymbol{u}})\\
=&\ \frac{1}{\gamma(\beta)}C_{\boldsymbol A}
{e}_{\boldsymbol{d},\boldsymbol{b}}
(\boldsymbol u)\mathscr{F}({e}_{\boldsymbol{a},
\boldsymbol{b}}f)
(\widetilde{\boldsymbol{u}})
\gamma(\beta)(2\pi)^{-\beta}
|\widetilde{\boldsymbol{u}}|^{-\beta} \\
=&\ (2\pi)^{-\beta}|\widetilde{\boldsymbol{u}}|^{-\beta}\mathscr{L}_{\boldsymbol{A}}\left( f\right)(\boldsymbol{u})
\end{align*}
in  $\mathscr{S}'(\mathbb{R}^n)$,
which completes the proof of   Theorem \ref{thm-frp}.
\end{proof}
 
Motivated by this symbolic correspondence, similarly to the classical case \eqref{jdsz}, we next introduce the linear canonical Laplacian operator.

\begin{definition} \label{def-lclo}
Assume $b_k \neq 0$ for any $k \in \{1, \dots, n\}$. Let $\gamma \in (0, \infty)$.
The \emph{linear canonical Laplacian operator} (for short, LCLO) $\Delta_{\gamma}^{\boldsymbol A}$ is defined by setting, for any $f \in \mathscr{S}(\mathbb{R}^n)$ and $\boldsymbol{u} := (u_1, \dots, u_n) \in \mathbb{R}^n$,
$$\mathscr{L}_{\boldsymbol A}\left(-\Delta_{\gamma}^{\boldsymbol A}f
\right)(\boldsymbol{u}) := \mathscr{L}_{\boldsymbol A}\left(\left[- e_{-\boldsymbol a,
	\boldsymbol b}\Delta\left(e_{
	\boldsymbol a,\boldsymbol b}
f\right)\right]^{\frac{\gamma}{2}}\right)
(\boldsymbol{u})
:= (2\pi)^{\gamma}|\widetilde{\boldsymbol{u}}|^{\gamma}\mathscr{L}_{\boldsymbol A}(f)
(\boldsymbol{u}),$$
where $\widetilde{\boldsymbol{u}} := (\widetilde{u}_1, \dots, \widetilde{u}_n) = (\frac{u_1}{b_1}, \dots, \frac{u_n}{b_n})$ and where $e_{\boldsymbol a, \boldsymbol b}$ and $e_{-\boldsymbol a, \boldsymbol b}$ are as in Definition \ref{def-lcrp}.
The function $(2\pi)^{\gamma}|\widetilde{\boldsymbol{u}}|^{\gamma}$ is referred to as the \emph{symbol} of the LCLO.
\end{definition}

To perfectly match the matrix notation of the LCLO and to provide a more convenient representation for the numerical simulations in the subsequent sections, we now introduce the LCRP associated with the parameter matrix $\boldsymbol{A}$.

\begin{remark}\label{rm-LCRP}
Let $\beta \in (0, n)$ and assume $b_k \neq 0$ for any $k \in \{1, \ldots, n\}$.
The \emph{linear canonical Riesz potential} $I^{\boldsymbol{A}}_\beta$ associated with $\boldsymbol{A}$ is defined by setting, for any $f \in \mathscr{S}(\mathbb{R}^n)$,
$$I^{\boldsymbol{A}}_\beta(f) := I^{\boldsymbol{a}, \boldsymbol{b}}_\beta(f).$$
\end{remark}

Based on Theorem \ref{thm-frp}, Definition \ref{def-lclo}, and Remark \ref{rm-LCRP}, we immediately obtain the following property; we omit the details.

\begin{proposition} \label{pro-lclo}
Let $\beta \in (0, n)$ and assume $b_k \neq 0$ for any $k \in \{1, \dots, n\}$.
For any $f \in \mathscr{S}(\mathbb{R}^n)$ and $\boldsymbol{u} \in \mathbb{R}^n$,
$$\mathscr{L}_{\boldsymbol A}\left(-\Delta_{\beta}^{\boldsymbol A}I^{\boldsymbol{A}}_\beta f
\right)(\boldsymbol{u}) = \mathscr{L}_{\boldsymbol A}(f)
(\boldsymbol{u}) \text{.}$$ 
\end{proposition}

Clearly, Proposition \ref{pro-lclo} indiates that the LCLO $\Delta_{\gamma}^{\boldsymbol{A}}$ can be regarded as the inverse operators of the LCRP $I^{\boldsymbol{A}}_\beta$. These operators are proved to play a key role in the encryption and the decryption processes of the image encryption scheme discussed in Section \ref{sec4}.
 
Next,  we  show that, for any Schwartz function $f$, $I^{\boldsymbol{a},\boldsymbol{b}}_\beta f$ pointwise converges at the critical points of the parameters  $\boldsymbol{a}$,   $\boldsymbol{b}$, and $\beta$. It is worth noting that, while \cite{lx25} provides a complete summary of the classical results found in \cite{r1949} and \cite[Chapter V]{s1970}, our theorem generalizes these classical results to the case of LCRPs.

\begin{theorem}\label{p-sharp} 
Let  $\beta\in(0,n)$,  and let  
$\boldsymbol{a}:=(a_1,\ldots,a_n),
\boldsymbol{b}:=(b_1,\ldots,b_n)\in\mathbb R^n$  with each $b_j\neq0$.
Then the following statements hold.
\begin{enumerate}[\rm (i)]
\item
For any $f\in\mathscr{S}(\mathbb{R}^n)$ and $\boldsymbol x\in\mathbb R^n$,  
$$\lim_{\boldsymbol{a}\to\boldsymbol{0}}I^{\boldsymbol{a},\boldsymbol{b}}_\beta
f(\boldsymbol x)=I_\beta f(\boldsymbol x),$$
where $\boldsymbol{a}\to\boldsymbol{0}$ means that $a_j \to 0$ for any $j \in \{1, \dots, n\}$.
\item 
For any $f\in\mathscr{S}(\mathbb{R}^n)$ and $\boldsymbol x\in\mathbb R^n$,  
$$\lim_{\beta\to 0^+}I^{\boldsymbol{a},\boldsymbol{b}}_\beta f(\boldsymbol x)
= f(\boldsymbol x).$$

\item 
For any $f \in \mathscr{S}(\mathbb{R}^n)$ with $\int_{\mathbb{R}^n} {e}_{\boldsymbol{b},\boldsymbol{a}}(\boldsymbol{y}) f(\boldsymbol{y}) \, d\boldsymbol{y} = 0$ and for any $\boldsymbol{x} \in \mathbb{R}^n$,
$$\lim_{\beta \to n^-} I^{\boldsymbol{a},\boldsymbol{b}}_\beta f(\boldsymbol{x}) = I^{\boldsymbol{a},\boldsymbol{b}}_n f(\boldsymbol{x}),$$
where ${\beta\to n^-}$ means $\beta\in (0,n)$ and $\beta\to n$ and where $I^{\boldsymbol{a},\boldsymbol{b}}_n$ denotes the normalized logarithmic potential defined by setting
$$I^{\boldsymbol{a},\boldsymbol{b}}_n f(\boldsymbol{x}) := \frac{e_{-\boldsymbol{a},\boldsymbol{b}}(\boldsymbol{x})}{\pi^{\frac{n}{2}} 2^{n-1} \Gamma\left(\frac{n}{2}\right)} 
\int_{\mathbb{R}^n} e_{\boldsymbol{a},\boldsymbol{b}}(\boldsymbol{y}) \ln \frac{1}{|\boldsymbol{x} - \boldsymbol{y}|} f(\boldsymbol{y}) \, d\boldsymbol{y}$$
with ${e}_{-\boldsymbol a,\boldsymbol b}$  as 
in \eqref{1eq} in which $\boldsymbol{a}$ is replaced  by   
$-\boldsymbol a$ and   with ${e}_{\boldsymbol a,\boldsymbol b}$
as in \eqref{1eq}.
\end{enumerate}
\end{theorem}
\begin{proof}
In order to prove (i), fix $f\in\mathscr{S}(\mathbb{R}^n)$ and $\boldsymbol x\in\mathbb R^n$.
Using the elementary inequality $|e^{i\theta}-1|\leq |\theta|$ for any $\theta\in\mathbb{R}$, we obtain, for any $\boldsymbol y\in\mathbb{R}^n$, 
\begin{equation*} 
\left|{e}_{\boldsymbol a,\boldsymbol b}(\boldsymbol{y})-1\right| = \left|e^{\pi i\sum_{j=1}^{n}\frac{a_j}{b_j}y_j^2}-1\right| \leq \pi \sum^{n}_{j=1}\frac{|a_j|}{|b_j|}y_j^2.
\end{equation*}
This yields 
\begin{align*}
&\left| \int_{\mathbb{R}^n}\frac{f(\boldsymbol{y}){e}_{\boldsymbol a,\boldsymbol b}(\boldsymbol{y})}{|\boldsymbol{x}-\boldsymbol{y}|^{n-\beta}}\,d\boldsymbol{y} - \int_{\mathbb{R}^n}\frac{f(\boldsymbol{y})}{|\boldsymbol{x}-\boldsymbol{y}|^{n-\beta}}\,d\boldsymbol{y}\right| \\
&\quad \leq \int_{\mathbb{R}^n}\frac{|f(\boldsymbol{y})|}{|\boldsymbol{x}-\boldsymbol{y}|^{n-\beta}}\left|{e}_{\boldsymbol a,\boldsymbol b}(\boldsymbol{y})-1\right|\,d\boldsymbol{y}
\leq \sum^{n}_{j=1}\frac{\pi|a_j|}{|b_j|}\int_{\mathbb R^n}\frac{|f(\boldsymbol y)| y_j^2}{|\boldsymbol x-\boldsymbol y|^{n-\beta}}\,d\boldsymbol{y}\\
&\quad = \sum^{n}_{j=1}\frac{\pi|a_j|}{|b_j|}\int_{|\boldsymbol x-\boldsymbol y|\leq1}\frac{|f(\boldsymbol y)| y_j^2}{|\boldsymbol x-\boldsymbol y|^{n-\beta}}\,d\boldsymbol{y} + \sum^{n}_{j=1}\frac{\pi|a_j|}{|b_j|}\int_{|\boldsymbol x-\boldsymbol y|>1}\cdots\,d\boldsymbol{y}\\
&\quad =: {\rm H}_1+{\rm H}_2.
\end{align*}

We first estimate ${\rm H}_2$. Since $f\in\mathscr{S}(\mathbb{R}^n)$, we deduce that the function $|f(\boldsymbol y)|y_j^2$ is integrable over $\mathbb R^n$, and hence we obtain
\begin{align*}
0\leq {\rm H}_2\leq\sum^{n}_{j=1}\frac{\pi |a_j|}{|b_j|}\int_{\mathbb R^n}\left|f(\boldsymbol y)\right|y_j^2\,d\boldsymbol{y}\to 0
\end{align*}
as $\boldsymbol{a}\to\boldsymbol{0}$; consequently, ${\rm H}_2 \to 0$ as $\boldsymbol{a}\to\boldsymbol{0}$.

We now turn to the estimation of ${\rm H}_1$.  
For any $\boldsymbol y\in\mathbb{R}^n$ satisfying $|\boldsymbol x-\boldsymbol y|\leq1$, the triangle inequality gives $|\boldsymbol y|\leq|\boldsymbol y-\boldsymbol x|+|\boldsymbol x|\leq 1+|\boldsymbol x|$, which implies $y_j^2 \le (1+|\boldsymbol x|)^2$. Since $f\in\mathscr{S}(\mathbb{R}^n)$, from \cite[Definition 1.8.1]{g24}, it follows that, for any $N\in\mathbb{Z}_+$, there exists a positive constant $C_N$ such that $|f(\boldsymbol y)| \leq C_N (1+|\boldsymbol y|)^{-N}$ for any $y\in\mathbb{R}^n$. Choosing $N=0$, we find that $|f(\boldsymbol y)| \leq C_0$ uniformly on $\mathbb{R}^n$. Therefore, we bound ${\rm H}_1$ by
\begin{align*}
0\leq {\rm H}_1 &\leq \sum^{n}_{j=1}\frac{\pi|a_j|}{|b_j|}\int_{|\boldsymbol x-\boldsymbol y|\leq1}\frac{C_0 y_j^2}{| \boldsymbol x-\boldsymbol y|^{n-\beta}}\,d\boldsymbol{y}\\
&\leq \sum^{n}_{j=1}\frac{\pi|a_j|C_0(1+|\boldsymbol x|)^2}{|b_j|}\int_{|\boldsymbol x-\boldsymbol y|\leq1}\frac{1}{| \boldsymbol x-\boldsymbol y|^{n-\beta}}\,d\boldsymbol{y}\to 0
\end{align*}
as $\boldsymbol{a}\to\boldsymbol{0}$ (note that the integral $\int_{|\boldsymbol x-\boldsymbol y|\leq1} |\boldsymbol x-\boldsymbol y|^{\beta-n}\,d\boldsymbol{y}$ is finite since $\beta>0$).
Hence, ${\rm H}_1 \to 0$ as $\boldsymbol{a}\to\boldsymbol{0}$.

Combining the estimates for ${\rm H}_1$ and ${\rm H}_2$, we finally conclude that
$$\lim_{\boldsymbol{a}\to\boldsymbol{0}}\left| \int_{\mathbb{R}^n}\frac{f(\boldsymbol{y})}{|\boldsymbol{x}-\boldsymbol{y}|^{n-\beta}}{e}_{\boldsymbol a,\boldsymbol b}(\boldsymbol{y})\,d\boldsymbol{y}-\int_{\mathbb{R}^n}\frac{f(\boldsymbol{y})}{|\boldsymbol{x}-\boldsymbol{y}|^{n-\beta}}\,d\boldsymbol{y}\right|=0.$$
Equivalently,
$$\lim_{\boldsymbol{a}\to\boldsymbol{0}} \int_{\mathbb{R}^n}\frac{f(\boldsymbol{y})}{|\boldsymbol{x}-\boldsymbol{y}|^{n-\beta}}{e}_{\boldsymbol a,\boldsymbol b}(\boldsymbol{y})\,d\boldsymbol{y}=\int_{\mathbb{R}^n}\frac{f(\boldsymbol{y})}{|\boldsymbol{x}-\boldsymbol{y}|^{n-\beta}}\,d\boldsymbol{y}.$$
This and the continuity of exponential functions further imply that  
$$\lim_{\boldsymbol{a}\to\boldsymbol{0}}I^{\boldsymbol{a},\boldsymbol{b}}_\beta f(\boldsymbol x)=I_\beta f(\boldsymbol x).$$
This finishes the proof of (i).

To show (ii), according to Theorem \ref{thm-frp}, for any $f \in \mathscr{S}(\mathbb{R}^n)$, the LCT $\mathscr{L}_{\mathbf{A}}(I_{\beta}^{\boldsymbol{a},\boldsymbol{b}} f)$ of $I_{\beta}^{\boldsymbol{a},\boldsymbol{b}} f$ satisfies that, for any $\boldsymbol{u} \in \mathbb{R}^n$, 
\begin{equation}\label{eqc1}
\mathscr{L}_{\mathbf{A}}\left(I_{\beta}^{\boldsymbol{a},\boldsymbol{b}} f\right)(\boldsymbol{u}) = (2\pi)^{-\beta} |\widetilde{\boldsymbol{u}}|^{-\beta} \mathscr{L}_{\boldsymbol{A}}(f)(\boldsymbol{u}),
\end{equation}
where $\widetilde{\boldsymbol{u}} = (\frac{u_1}{b_1}, \ldots, \frac{u_n}{b_n})$. 
Since both $(2\pi)^{-\beta}$ and $|\widetilde{\boldsymbol{u}}|^{-\beta}$ are continuous on $\beta$ and equal 1 when $\beta = 0$,  we infer that
$\lim_{\beta\to 0^+} (2\pi)^{-\beta} |\widetilde{\boldsymbol{u}}|^{-\beta} = 1 $
for all $\boldsymbol{u} \in \mathbb{R}^n \setminus \{\boldsymbol{0}\}$.
To recover the spatial domain behavior, we apply the inverse LCT to both sides of \eqref{eqc1} to obtain
$$I_{\beta}^{\boldsymbol{a},\boldsymbol{b}} f = \mathscr{L}_{\mathbf{A}}^{-1} \left[ (2\pi)^{-\beta} |\widetilde{\boldsymbol{u}}|^{-\beta} \mathscr{L}_{\mathbf{A}}(f) \right].$$
Fix $\beta$ in a bounded neighborhood of $0$, say $\beta \in [0, \beta_0]$ with $\beta_0 \in (0,n)$. The multiplier $|\widetilde{\boldsymbol{u}}|^{-\beta}$ is bounded by
$\max\{1,|\widetilde{\boldsymbol{u}}|^{-\beta_0}\}$. 
Note $\mathscr{L}_{\mathbf{A}}(f) \in \mathscr{S}(\mathbb{R}^n)$, which, together with $\beta_0 < n$ and  the Lebesgue dominated convergence theorem, implies that 
$$\lim_{\beta\to 0^+} I_{\beta}^{\boldsymbol{a},\boldsymbol{b}} f = \mathscr{L}_{\mathbf{A}}^{-1} \left[ \lim_{\beta\to 0^+} (2\pi)^{-\beta} |\widetilde{\boldsymbol{u}}|^{-\beta} \mathscr{L}_{\mathbf{A}}(f) \right] = \mathscr{L}_{\mathbf{A}}^{-1} \left[ \mathscr{L}_{\mathbf{A}}(f) \right] = f.$$
This  finishes the  proof of (ii).

To prove (iii), let $f \in \mathscr{S}(\mathbb{R}^n)$ and $g := e_{\boldsymbol{a},\boldsymbol{b}} f \in \mathscr{S}(\mathbb{R}^n)$. 
The assumption $\int_{\mathbb{R}^n} e_{\boldsymbol{a},\boldsymbol{b}}(\boldsymbol{y})f(\boldsymbol{y}) \, d\boldsymbol{y}=0$ implies that 
$ \int_{\mathbb{R}^n} g(\boldsymbol{y}) \, d\boldsymbol{y} = 0. $
By this and \cite[Lemma 1.1(ii)]{lx25}, we conclude that, for any $\boldsymbol{x}\in {\mathbb{R}^n}$,
\begin{equation*} 
\lim_{\beta \to n} I_\beta g(\boldsymbol{x}) = I_n g(\boldsymbol{x}),
\end{equation*}
where $I_n g$ is the normalized logarithmic potential defined by
$$
I_n g(\boldsymbol{x}) := \frac{1}{\pi^{\frac{n}{2}} 2^{n-1} \Gamma\left(\frac{n}{2}\right)} \int_{\mathbb{R}^n} \ln \frac{1}{|\boldsymbol{x} - \boldsymbol{y}|} \, g(\boldsymbol{y}) \, d\boldsymbol{y}.
$$
From this and Definition \ref{def-lcrp}, we infer that 
\begin{align*}
\lim_{\beta \to n} I_\beta^{\boldsymbol{a},\boldsymbol{b}} f(\boldsymbol{x}) 
&=\lim_{\beta \to n} e_{-\boldsymbol{a},\boldsymbol{b}}(\boldsymbol{x}) I_\beta g(\boldsymbol{x})\\
&=e_{-\boldsymbol{a},\boldsymbol{b}}(\boldsymbol{x}) \frac{1}{\pi^{\frac{n}{2}} 2^{n-1} \Gamma\left(\frac{n}{2}\right)} \int_{\mathbb{R}^n} \ln \frac{1}{|\boldsymbol{x} - \boldsymbol{y}|} e_{\boldsymbol{a},\boldsymbol{b}}(\boldsymbol{y}) f(\boldsymbol{y}) \, d\boldsymbol{y}= I_n^{\boldsymbol{a},\boldsymbol{b}} f(\boldsymbol{x}),\end{align*}
which completes the proof of  (iii) and hence Theorem  \ref{p-sharp}.
\end{proof}

Natural images typically consist of backgrounds along with discontinuous structures such as edges and contours, and therefore lack the infinite smoothness and the rapid decay properties that are characteristic of Schwartz functions. Based on these fundamental differences, we reveal two intrinsic analytic mechanisms through Theorems \ref{rem-imp-1}, and \ref{thm-polygon}. These mechanisms directly explain the core advantages of the LCRP in image encryption. 
On the one hand, a global image background is analytically characterized by its infinite Lebesgue measure and non-decaying behavior at spatial infinity. The grating function serves as the canonical mathematical model for backgrounds. By evaluating this worst-case non-decaying scenario, this mechanism elucidates why the LCRP can successfully perform the encryption function where classical operators fail.
 On the other hand, concerning the edge structures of images (exemplified by the characteristic functions of the open convex polygon), it explains why the LCRP achieves exceptional encryption performance and high key sensitivity at structural singularities.

An important tool to obtain these results is the following incomplete Fresnel integral on $\mathbb{R}$ from \cite{j2015}.

\begin{definition} 
Let $\kappa \in \mathbb{R}\setminus\{0\}$. The \emph{incomplete Fresnel integral} $F_\kappa$ is defined by setting, for any $s \in \mathbb{R}$,
\begin{equation*}
	F_\kappa(s) := \int_0^s e^{\pi i \kappa t^2} \, dt.
\end{equation*}
When $s=\infty$,  $F_\kappa(s)$ is called the \textit{standard Fresnel integral}.
\end{definition}

To state the main results, we also need the concept of grating functions (see \cite[p. 3]{mzlzf96}). The \textit{grating function} $f$ on $\mathbb{R}^2$ is defined by setting, for any $\boldsymbol{x} := (x_1, x_2) \in \mathbb{R}^2$,
\begin{align}\label{gf}
f(\boldsymbol{x}) := \sum_{n=1}^{m} c_n\mathbf{1}_{[2n, 2n+1]}(x_1),\end{align}
where $m \in \mathbb{N}$ and $c_n \in \mathbb{R}$, and we let $\|f\|_{\infty}:=\max\{|c_n|,n\in\{1,\ldots,m\}\}$. We next present the divergence/convergence at critical indices of $I_1 f$ and $I_1^{\boldsymbol{a},\boldsymbol{b}}f$ when $f$ is a grating function on $\mathbb{R}^2$.

\begin{theorem}\label{rem-imp-1}
Let $\boldsymbol{a}:=(a_1, a_2), \boldsymbol{b}:=(b_1, b_2)\in\mathbb R^2$ with each $b_j\neq0$. Let  $k_j:= \frac{a_j}{b_j}$ for any $j\in\{1,2\}$. Then, for any grating function $f$  on $\mathbb{R}^2$ as in \eqref{gf} and for any $\boldsymbol{x} := (x_1, x_2) \in \mathbb{R}^2$, the following statements hold.
\begin{enumerate}
\item[\rm (a)] The classical Riesz potential diverges, namely
$I_1 f(\boldsymbol{x}) = \infty$.
\item[\rm (b)] If $k_1\neq0\neq k_2$, then
$$\left|\lim_{k_1,k_2 \to \kappa} I_1^{\boldsymbol{a},\boldsymbol{b}} f(\boldsymbol{x})\right| \leq C\|f\|_{\infty}\frac{1}{\sqrt{\min\{|k_1|, |k_2|\}}},$$
where $C$ is a positive constant independent of $f$,  $k_1$, and  $k_2$.
\end{enumerate}
\end{theorem}

\begin{remark}\label{zyr-1}
Theorem \ref{rem-imp-1} demonstrates that the classical Riesz potential may diverge for non-smooth functions (such as discontinuous characteristic functions), while the LCRP may converge even for such functions  because chirp functions have oscillatory factors. This further highlights the big difference in image encryption from the ``\emph{incapability}" of the classical Riesz potential to the ``\emph{capability}" of the LCRP.
\end{remark}

\begin{proof}[Proof of Theorem \ref{rem-imp-1}]
Note that $f \notin \mathscr{S}(\mathbb{R}^2)$ due to its non-decaying behavior along the $x_2$-direction and its discontinuities. By the linearity of the operators $I_1$ and $I_1^{\boldsymbol{a},\boldsymbol{b}}$, without loss of generality, we may assume $f(\boldsymbol{x}) := \sum_{n=1}^{2} \mathbf{1}_{[2n, 2n+1]}(x_1)$. 

To show (a), we observe that $f$ is non-negative and consists of two vertical stripes. We can establish a lower bound by considering only the first stripe. Let $f_1(\boldsymbol{y}) := \mathbf{1}_{[2, 3]}(y_1)$.
Clearly, $f(\boldsymbol{y}) \geq f_1(\boldsymbol{y})$ for all $\boldsymbol{y} \in \mathbb{R}^2$. 
Recall that, for any  $\boldsymbol{x} \in \mathbb{R}^2$, the classical Riesz potential satisfies
\begin{equation*}
I_1 f(\boldsymbol{x}) \geq I_1 f_1(\boldsymbol{x}) = \frac{1}{2\pi} \int_{-\infty}^{\infty} \int_{2}^{3} \frac{1}{|\boldsymbol{x} - \boldsymbol{y}|} \, dy_1 \, dy_2.
\end{equation*}
For any $\boldsymbol{y} = (y_1, y_2)$ in the stripe $[2, 3] \times \mathbb{R}$, we clearly have $|\boldsymbol{y}| \leq |y_1| + |y_2| \leq 3 + |y_2|$. The triangle inequality 
implies $|\boldsymbol{x} - \boldsymbol{y}| \leq |\boldsymbol{x}| + 3 + |y_2|$. 
Letting $C_{\boldsymbol{x}} := |\boldsymbol{x}| + 3$, which is a constant depending only on $\boldsymbol{x}$, we obtain
$$
I_1 f(\boldsymbol{x}) \geq \frac{1}{2\pi} \int_{-\infty}^{\infty} \int_{2}^{3} \frac{1}{C_{\boldsymbol{x}} + |y_2|} \, dy_1 \, dy_2 = \frac{1}{2\pi} \int_{-\infty}^{\infty} \frac{1}{C_{\boldsymbol{x}} + |y_2|} \, dy_2.
$$
It is easy to verify that the last integral diverges to $\infty$. Hence, $I_1 f(\boldsymbol{x}) = \infty$ for all $\boldsymbol{x} \in \mathbb{R}^2$, which completes the proof of (a).

To prove (b), we now investigate the limit of $I_1^{\boldsymbol{a},\boldsymbol{b}} f(\boldsymbol{x})$ as $k_1, k_2 \to \kappa \in [-\infty,0)\cup(0,\infty]$ for any $\boldsymbol{x} \in \mathbb{R}^2$.
Let $f_n(\boldsymbol{x}) := \mathbf{1}_{[2n, 2n+1]}(x_1)$ denote the $n$-th single stripe for $n \in \{1, 2\}$. By linearity, $I_1^{\boldsymbol{a},\boldsymbol{b}} f(\boldsymbol{x}) = \sum_{n=1}^{2} I_1^{\boldsymbol{a},\boldsymbol{b}} f_n(\boldsymbol{x})$. It suffices to verify the convergence for any single infinite stripe $f_n$.
For this purpose, we make the change of variables $\boldsymbol{z} := \boldsymbol{y} - \boldsymbol{x}$ and hence the condition  $2n \leq y_1 \leq 2n+1$ transforms to $\rho_{n,1} \leq z_1 \leq \rho_{n,2}$, where $\rho_{n,1} := 2n - x_1$ and $\rho_{n,2} := 2n+1 - x_1$, which, together with Definition  \ref{def-lcrp}, yields 
\begin{equation}\label{eq:shifted-n}
I_1^{\boldsymbol{a},\boldsymbol{b}} f_n(\boldsymbol{x}) = \frac{1}{2\pi} \int_{\rho_{n,1} \leq z_1 \leq \rho_{n,2}} \frac{1}{|\boldsymbol{z}|} e^{\pi i \Phi(\boldsymbol{z})} \, d\boldsymbol{z},
\end{equation}
where $\Phi(\boldsymbol{z}) := k_1 z_1^2 + k_2 z_2^2 + 2k_1 x_1 z_1 + 2k_2 x_2 z_2$.
Introducing polar coordinates $\boldsymbol{z} = (r\cos\theta, r\sin\theta)$, equation \eqref{eq:shifted-n} becomes
\begin{equation}\label{eq:polar-form-n}
I_1^{\boldsymbol{a},\boldsymbol{b}} f_n(\boldsymbol{x}) = \frac{1}{2\pi} \int_{-\pi}^{\pi} \left\{\int_{R_n(\theta)} e^{\pi i [ K(\theta) r^2 + L(\theta) r ]} \, dr \right\}\, d\theta,
\end{equation}
where $K(\theta) := k_1\cos^2\theta + k_2\sin^2\theta$, $L(\theta) := 2 ( k_1 x_1 \cos\theta + k_2 x_2 \sin\theta)$, and $$R_n(\theta) := \{ r \geq 0 : \rho_{n,1} \leq r\cos\theta \leq \rho_{n,2} \}.$$
Without loss of generality, we may assume $\cos\theta\neq 0$. Denote the inner $r$-integral of \eqref{eq:polar-form-n} by $J_n(\theta; k_1,k_2)$. 
Noting $$\pi i \left[K(\theta) r^2 + L(\theta) r\right ] = \pi i K(\theta)\left[r + \frac{L(\theta)}{2K(\theta)}\right]^2 - \pi i \frac{L(\theta)^2}{4K(\theta)},$$ by the change of variables $t = r + \frac{L(\theta)}{2K(\theta)}$, we obtain
\begin{equation*} 
J_n(\theta; k_1,k_2) = e^{-\pi i \frac{L(\theta)^2}{4K(\theta)}} \int_{t_1}^{t_2} e^{\pi i K(\theta) t^2} \, dt = e^{-\pi i \frac{L(\theta)^2}{4K(\theta)}} [ F_{K(\theta)}(t_2) - F_{K(\theta)}(t_1) ],
\end{equation*}
where $t_1$ and $t_2$ correspond to the transformed boundary points of $R_n(\theta)$, namely $t_j = r_j + \frac{L(\theta)}{2K(\theta)}$ for $j \in \{1, 2\}$ with $r_1 := \max\{0, \min(\frac{\rho_{n,1}}{\cos\theta}, \frac{\rho_{n,2}}{\cos\theta})\}$ and $r_2 := \max\{0, \max(\frac{\rho_{n,1}}{\cos\theta}, \frac{\rho_{n,2}}{\cos\theta})\}$.

Using a change of variables again, we obtain, for any $s \in \mathbb{R}$, 
$$F_{K(\theta)}(s) = \frac{\sqrt{2}}{\sqrt{|K(\theta)|}} \int_0^{\frac{s\sqrt{|K(\theta)|}}{\sqrt{2}}} e^{2\pi i u^2 \sgn[K(\theta)]}\, du.$$
If ${\sgn}[K(\theta)]=1$, the convergence of the standard Fresnel integral (see \cite[Theorem 1]{j2015}) indicates that the function $P(T):=\int_{0}^{T}e^{2\pi i u^{2}}du$ has a finite limit as $T\to\infty$, which, together with the continuity of $P(T)$ on $[0,\infty)$, further implies that it is bounded. If ${\sgn}[K(\theta)]=-1$, the integral is precisely the complex conjugate of $P(T)$ and hence it is also bounded. Consequently, 
$$|F_{K(\theta)}(s)| \lesssim \frac{1}{\sqrt{|K(\theta)|}}.$$
Because $\kappa \neq 0$, for $k_1, k_2$ sufficiently close to $\kappa$, they must share the same sign. This guarantees $$|K(\theta)| = |k_1\cos^2\theta + k_2\sin^2\theta| \geq \min\{|k_1|, |k_2|\},$$ from which we deduce that
$$
|J_n(\theta; k_1, k_2)| \lesssim \frac{1}{\sqrt{|K(\theta)|}} \lesssim \frac{1}{\sqrt{\min\{|k_1|, |k_2|\}}}.
$$
By this and \eqref{eq:polar-form-n}, we conclude that 
$$
\left|I_1^{\boldsymbol{a},\boldsymbol{b}} f_n(\boldsymbol{x})\right| \lesssim \frac{1}{\sqrt{\min\{|k_1|, |k_2|\}}}.
$$
Combining  this  and the triangle inequality yields  
$$ \left|I_1^{\boldsymbol{a},\boldsymbol{b}} f(\boldsymbol{x})\right| \le \sum_{n=1}^{2} \left|I_1^{\boldsymbol{a},\boldsymbol{b}} f_n(\boldsymbol{x})\right| \lesssim \frac{1}{\sqrt{\min\{|k_1|, |k_2|\}}}. $$
This finishes the proof of (b) and hence Theorem \ref{rem-imp-1}.
\end{proof}

Having established the stark contrast between the incapability of the classical Riesz potential and the capability of the LCRP in handling global image backgrounds, we now turn our attention to its cryptographic performance at local boundaries. In natural images, structural features such as edges and contours mathematically manifest  as bounded functions with jump discontinuities,
for example,  $f := \mathbf{1}_{P}$ is the characteristic function of an arbitrary open convex polygon $P \subset \mathbb{R}^2$. On such characteristic functions, we have the following results.

\begin{theorem} \label{thm-polygon}
Let $\beta \in (0,2)$ and ${P}$ be an arbitrary open convex polygon in $ \mathbb{R}^2$.   Then the following assertions hold.
\begin{enumerate}
\item[\rm (a)] \textbf{Interior and exterior points:} For any $\boldsymbol{x} \in  $ or $\boldsymbol{x} \in \mathbb{R}^2 \setminus \overline{P}$, 
$$\lim_{\beta\to 0^+}I_\beta \mathbf{1}_{P}(\boldsymbol{x})= \mathbf{1}_{P}(\boldsymbol{x}) \text{.}$$ 
\item[\rm (b)] \textbf{Boundary edge points:} For any point $\boldsymbol{x} \in \partial P$ that lies on a straight edge (excluding the vertices), 
$$\lim_{\beta\to 0^+}I_\beta \mathbf{1}_{P}(\boldsymbol{x})=\frac{1}{2}\neq \mathbf{1}_{P}(\boldsymbol{x}) \text{.}$$ 
\item[\rm (c)] \textbf{Boundary corner points:} For any vertex $\boldsymbol{x} \in \partial P$ with an internal angle $\alpha \in (0, \pi)$ (measured in radians), 
$$\lim_{\beta\to 0^+}I_\beta \mathbf{1}_{P}(\boldsymbol{x})=\frac{\alpha}{2\pi}\neq \mathbf{1}_{P}(\boldsymbol{x}) \text{.}$$ 
\end{enumerate}
\end{theorem}
 
\begin{remark}\label{zyr-2}
Theorem \ref{thm-polygon} indicates that, even for non-Schwartz functions, the limiting behavior of the classical Riesz potential when parameters approach critical values may be discontinuous. This characteristic further indicates that, in image encryption, if there is a slight deviation in the parameter $\beta$ of $I_\beta$, the original information at the image edges and corners becomes extremely difficult to recover, thereby significantly enhancing the security and the key sensitivity of the ciphertext. 
\end{remark}

To prove Theorem \ref{thm-polygon}, we need the following two lemmas.
In what follows, for any $\beta \in (0,2)$, $o(\beta^k)$ denotes the \textit{higher-order infinitesimal}, which means $\lim_{\beta\to 0^+}\frac {o(\beta^k)} {\beta^k}=0.$

\begin{lemma}\label{lem-gamma-limit}
Let $\beta \in (0,2)$ and $\frac{1}{\gamma(\beta)}:= \frac{1}{\pi} 2^{-\beta} \frac{\Gamma(1 - \frac{\beta}{2})}{\Gamma(\frac{\beta}{2})}$. Then
\begin{equation}\label{eqbz1}
\frac{1}{ \gamma(\beta)} = \frac{1}{\pi  2^\beta} \left[ \frac{\beta}{2} + \frac{\gamma \beta^2}{2} + o(\beta^3) \right],
\end{equation}
where $\gamma \approx 0.577$ is the Euler-Mascheroni constant.  
\end{lemma}

\begin{proof}
The function $\Gamma(1 - \frac{\beta}{2})$ is analytic near $\beta=0$ and has the following Taylor expansion 
\begin{equation}\label{o1}
\Gamma\left(1 - \frac{\beta}{2}\right) = 1 + \frac{\gamma\beta}{2} + \left(\frac{\gamma^2}{8} + \frac{\pi^2}{48}\right)\beta^2 + o(\beta^3)
\end{equation}
 (see \cite[Section 8.32]{gr14}). The function $\Gamma(\frac{\beta}{2})$ has a simple pole at $\beta = 0$ and its Laurent expansion is  
\begin{equation}\label{o2}
\Gamma\left(\frac{\beta}{2}\right) = \frac{2}{\beta} - \gamma + \left(\frac{\gamma^2}{4} + \frac{\pi^2}{24}\right)\beta + o(\beta^2)
\end{equation}
(see \cite[Section 8.32]{gr14}). Combining \eqref{o1} and \eqref{o2} and applying the geometric series expansion that $\frac{1}{2-u} = \frac{1}{2}\sum_{k\in\mathbb{Z}_+}(\frac{u}{2})^k$, we directly obtain
\begin{align*}
\frac{\Gamma(1 -\frac{\beta}{2})}{\Gamma(\frac{\beta}{2})} =
\frac {\beta + \frac{\gamma\beta^2}{2} + (\frac{\gamma^2}{8} + \frac{\pi^2}{48})\beta^3 + o(\beta^4)}{2 - \gamma \beta+ (\frac{\gamma^2}{4} + \frac{\pi^2}{24})\beta^2 + o(\beta^3)} 
= \frac{\beta}{2} + \frac{\gamma\beta^2}{2} + o(\beta^3) \text{.}
\end{align*}
Substituting this into the definition of $\frac{1}{ \gamma(\beta)}$ yields
$$\frac{1}{ \gamma(\beta)} = \frac{1}{\pi  2^\beta} \left[ \frac{\beta}{2} + \frac{\gamma \beta^2}{2} + o(\beta^3) \right],$$
which completes the proof of Lemma \ref{lem-gamma-limit}.
\end{proof}

\begin{lemma}\label{lem-sector-limit}
Let $\beta \in (0,2)$ and ${\rm Sec}(\boldsymbol{x}, R, \theta)$ denote the {\rm sectorial region} in $\mathbb{R}^2$ with vertex $\boldsymbol{x}$, radius $R$, and central angle $\theta$. Then 
\begin{equation}\label{sa}
\lim_{\beta\to 0^+} I_\beta \mathbf{1}_{{\rm Sec}(\boldsymbol{x}, R, \theta)}(\boldsymbol{x}) = \frac{\theta}{2\pi}.
\end{equation}
\end{lemma}

\begin{proof}
By the translation invariance of the Riesz potential, we may assume $\boldsymbol{x} = \boldsymbol{0}$. For simplicity of presentation, let $S := {\rm Sec}(\boldsymbol{0}, R, \theta)$. For $\beta \in (0,2)$, using polar coordinates we obtain
\begin{align*}
I_\beta \mathbf{1}_S(\boldsymbol{0})
= \frac{1}{\gamma(\beta)} \int_{S} \frac{1}{|\boldsymbol{y}|^{2-\beta}}\,d\boldsymbol{y}
= \frac{1}{\gamma(\beta)} \int_{0}^{\theta} \int_{0}^{R} \frac{1}{r^{2-\beta}}\, r\,dr\,d\phi = \frac{1}{\gamma(\beta)} \, \theta \, \frac{R^{\beta}}{\beta}.
\end{align*}
Combining this and \eqref{eqbz1} yields
\begin{align*}
\lim_{\beta\to 0^+} I_\beta \mathbf{1}_{{S}}(\boldsymbol{x})
=\lim_{\beta\to 0^+} \frac{\theta}{\pi}\, \frac{R^{\beta}}{2^{\beta}} \left[\frac{1}{2} + \frac{\gamma}{2}\beta + o(\beta^{2})\right]
= \frac{\theta}{2\pi}.
\end{align*}
This finishes the proof of Lemma \ref{lem-sector-limit}.
\end{proof}

\begin{proof}[Proof of Theorem \ref{thm-polygon}]
We first give a detailed proof for the \textit{open unit square $S: = (-1,1) \times (-1,1)$} and then indicate how this topological localization and sectorial bounding techniques utilized for $S$ can also be applied to any convex polygon by merely adjusting the sector angles.

\textbf{Proof of (a) for $\mathbf{1}_S$.} We divide the proof into interior points and exterior points of the square $S$.

\textbf{Case $\boldsymbol{x} \in S$ (Interior points).}  Choose $\delta\in(0,1)$ sufficiently small such that $B(\boldsymbol{x},\delta)\subset S$. We decompose the integral domain defining $I_\beta  \mathbf{1}_S(\boldsymbol{x})$ into local and far components at $\boldsymbol{x}$  as follows
\begin{align*}
I_\beta  \mathbf{1}_S(\boldsymbol{x})&=\frac1{\gamma(\beta)}\int_{S}
\frac1{|\boldsymbol{x}-\boldsymbol{y}|^{2-\beta}}\,d\boldsymbol{y}\\
&=\frac1{\gamma(\beta)}\left[\int_{B(\boldsymbol{x},\delta)}\frac1{|\boldsymbol{x}-\boldsymbol{y}|^{2-\beta}}\,d\boldsymbol{y}
+\int_{S\setminus B(\boldsymbol{x},\delta)}\cdots\,d\boldsymbol{y}\right]
=:\frac1{\gamma(\beta)}({\rm J}_1+{\rm J}_2).
\end{align*}
For ${\rm J}_1$, we apply the change of variables $\boldsymbol{z}:=\boldsymbol{y}-\boldsymbol{x}$ followed by a polar coordinate transformation to obtain
\begin{equation}\label{e2.6}
{\rm J}_1=\int_{B(\boldsymbol{0},\delta)}
\frac1{|\boldsymbol{z}|^{2-\beta}}d\boldsymbol{z}
=2\pi\int_0^\delta r^{\beta-1}dr=\frac{2\pi}{\beta}\delta^\beta.
\end{equation}
For ${\rm J}_2$, noting that $|\boldsymbol{x}-\boldsymbol{y}|\ge\delta$ for all $\boldsymbol{y} \in S\setminus B(\boldsymbol{x},\delta)$ and the domain of integration is bounded, having measure at most 4, we conclude that  
\begin{equation*}
0\le {\rm J}_2\le\int_{S\setminus B(\boldsymbol{x},\delta)}\frac1{\delta^{2-\beta}}d\boldsymbol{y}
\le\frac{4}{\delta^{2-\beta}}\le\frac {4}{\delta^2}.
\end{equation*}
From this and \eqref{eqbz1}, it follows that
$$0\le\lim_{\beta\to 0^+}\frac 1{\gamma(\beta)}{\rm J}_2\le\lim_{\beta\to 0^+}\frac 4{\delta^22^{\beta}}\left[\frac{\beta}{2} + \frac{\gamma\beta^2}{2} + o(\beta^3)\right]=0.$$
By \eqref{e2.6} and \eqref{eqbz1}, we find that 
$$\lim_{\beta\to 0^+}\frac 1{\gamma(\beta)}{\rm J}_1=\lim_{\beta\to 0^+}\frac{2}{\beta}\frac{\delta^\beta}{2^\beta}\left[\frac{\beta}{2} + \frac{\gamma\beta^2}{2} + o(\beta^3)\right]=1.$$
From the above both limits, we infer that, for any $\boldsymbol{x}\in S$,
\begin{equation}\label{eq321}
\lim_{\beta\to 0^+} I_\beta  \mathbf{1}_S(\boldsymbol{x}) = 1= \mathbf{1}_S(\boldsymbol{x}),
\end{equation} 
which completes the proof of \textbf{Case $\boldsymbol{x} \in S$ (Interior points).} 

\textbf{Case $\boldsymbol{x} \notin \overline{S}$ (Exterior points).} Since $\overline{S}$ is closed, it follows that $d := \inf_{\boldsymbol{y} \in S} |\boldsymbol{x} - \boldsymbol{y}| > 0$. For any $\boldsymbol{y} \in S$, we obvious have $|\boldsymbol{x} - \boldsymbol{y}| \ge d$. Thus, the integrand is bounded above by $d^{\beta-2}$, and hence
\begin{equation*}
0 \le I_\beta  \mathbf{1}_S(\boldsymbol{x}) =\frac1{\gamma(\beta)}\int_{S}\frac1{|\boldsymbol{x}-\boldsymbol{y}|^{2-\beta}}\,d\boldsymbol{y}
\le \frac{1}{\gamma(\beta)} \int_{S} \frac{1}{d^{2-\beta}} \, d\boldsymbol{y} = \frac{4}{\gamma(\beta)} d^{\beta-2},
\end{equation*}
which, combined with \eqref{eqbz1}, further implies that
$$0 \le \lim_{\beta\to 0^+}I_\beta  \mathbf{1}_S(\boldsymbol{x}) \le \lim_{\beta\to 0^+}\frac{4}{\pi 2^\beta} \left[ \frac{\beta}{2} + \frac{\gamma\beta^2}{2} + o(\beta^3) \right] d^{\beta-2}=0.$$
Consequently, for any $\boldsymbol{x} \notin \overline{S}$,
\begin{equation}\label{eq421}
\lim_{\beta\to 0^+} I_\beta  \mathbf{1}_S(\boldsymbol{x}) = 0 =  \mathbf{1}_S(\boldsymbol{x}),
\end{equation}
which completes the proof of \textbf{Case $\boldsymbol{x} \notin \overline{S}$ (Exterior points)}. 

Combining \eqref{eq321} and \eqref{eq421} thus completes the proof of (a).

\textbf{Proof of (b) for $\mathbf{1}_S$.} Without loss of generality, we may only examine a point on the right edge of $S$, $\boldsymbol{x} = (1, x_2)$, where $x_2 \in (-1, 1)$.
The internal angle of $S$ at $\boldsymbol{x}$ is $\pi$.
First, we establish the lower bound.
Choose the radius $R := \min\{1+x_2, 1-x_2, 2\}$ to construct a semi-circle $H_{\rm in}$ centered at $\boldsymbol{x}$ that is completely inscribed within $S$.
Since $H_{\rm in} \subset S$, we have $\mathbf{1}_{H_{\rm in}} \le \mathbf{1}_{S}$. By this, the positivity of $I_\beta$, and \eqref{sa}, we obtain 
\begin{equation}\label{eqq1}
\liminf_{\beta\to 0^+} I_\beta \mathbf{1}_S(\boldsymbol{x}) \ge \lim_{\beta\to 0^+} I_\beta \mathbf{1}_{H_{\rm in}}(\boldsymbol{x}) = \frac{\pi}{2\pi} = \frac{1}{2}.
\end{equation}
Next, we establish the upper bound. To this end, construct a semi-circle $H_{\rm out}$ centered at $\boldsymbol{x}$ facing the left half-plane with the radius $R = 2\sqrt{2}$. The square $S$ is entirely encompassed by this semi-circle, namely $S \subset H_{\rm out}$. Then, from this, the positivity of $I_\beta$, and \eqref{sa}, we deduce that 
\begin{equation}\label{eqq2}
\limsup_{\beta\to 0^+} I_\beta \mathbf{1}_S(\boldsymbol{x}) \le \lim_{\beta\to 0^+} I_\beta \mathbf{1}_{H_{\rm out}}(\boldsymbol{x}) = \frac{\pi}{2\pi} = \frac{1}{2}.
\end{equation}
Combining  \eqref{eqq1} and \eqref{eqq2} yields $\lim_{\beta\to 0^+} I_\beta \mathbf{1}_S(\boldsymbol{x})=\frac{1}{2} \neq \mathbf{1}_S(\boldsymbol{x})$,
which completes the proof of (b).

\textbf{Proof of (c) for $\mathbf{1}_S$.} Without loss of generality, we may only evaluate the top-right corner point of $S$, $\boldsymbol{x} = (1, 1)$.
The internal angle at this corner is $\frac\pi 2$.
For the lower bound, we construct a quarter-circle $Q_{\rm in}$ centered at $(1,1)$ with radius $R= 2$.
Clearly, $Q_{\rm in} \subset S$, which, combined with the positivity of $I_\beta$ and \eqref{sa}, implies that
\begin{equation}\label{eqq3}
\liminf_{\beta\to 0^+} I_\beta \mathbf{1}_S(\boldsymbol{x}) \ge \lim_{\beta\to 0^+} I_\beta \mathbf{1}_{Q_{\rm in}}(\boldsymbol{x}) = \frac{\pi/2}{2\pi} = \frac{1}{4}.
\end{equation}
For the upper bound, we construct an enlarged quarter-circle $Q_{\rm out}$ centered at $(1,1)$ with the radius $R = 2\sqrt{2}$, yielding $S \subset Q_{\rm out}$. Consequently, from this, the positivity of $I_\beta$, and \eqref{sa}, it follows that
\begin{equation}\label{eqq4}
\limsup_{\beta\to 0^+} I_\beta \mathbf{1}_S(\boldsymbol{x}) \le \lim_{\beta\to 0^+} I_\beta \mathbf{1}_{Q_{\rm out}}(\boldsymbol{x}) = \frac{\pi/2}{2\pi} = \frac{1}{4}.
\end{equation}
Combining  \eqref{eqq3} and \eqref{eqq4}, we conclude that $\lim_{\beta\to 0^+} I_\beta \mathbf{1}_S(\boldsymbol{x})=\frac{1}{4} \neq \mathbf{1}_S(\boldsymbol{x})$, which completes the proof of (c). 

\textbf{\textit{Generalization to Arbitrary Convex Polygons:}} For an arbitrary convex polygon $P$, the argument above applies perfectly. Because $P$ is strictly convex, it is entirely contained within the finite angular sector defined by the internal angle $\alpha$ extended from any vertex.
By replacing the $\frac\pi 2$ quarter-circles with bounded sectors of angle $\alpha$ (with radius sufficiently large to encompass $P$), the squeeze theorem yields $\frac{\alpha}{2\pi}$.
This universally confirms part (c) and finishes the proof for the general case and hence Theorem \ref{thm-polygon}.
\end{proof} 

\begin{remark} \label{the-spaces}
The phenomena observed in Theorems \ref{rem-imp-1} and \ref{thm-polygon} are not isolated cases but reflect the intrinsic analytic properties of the operators acting on specific function spaces. We generalize these observations to two fundamental spaces in image processing.

\begin{enumerate}[\rm (i)]
\item \textbf{Bounded Non-Decaying Function Space (Image Backgrounds).} 		This space represents the global background of all images, which consists of all functions having bounded pixel values that do not decay to zero at the image boundaries.

\item \textbf{Piecewise Smooth Function Space (Image Edges).}
This space models the edges and the contours within an image, which consists of all functions characterized by smooth regions separated by jump discontinuities, such as straight boundaries and sharp corners.

\end{enumerate}

To intuitively illustrate the distinct behaviors discussed above and their critical roles in our cryptosystem, we summarize the convergence analysis for image backgrounds and the sensitivity mechanism at image edges in Tables \ref{tab:convergence-comparison} and \ref{tab:sensitivity-mechanism}, respectively.

\begin{table}[H]
\caption{Comparison of Convergence Behaviors on Bounded Non-Decaying Backgrounds.}
\label{tab:convergence-comparison}
\centering
\setlength{\tabcolsep}{0.16cm}
\begin{tabular}{p{2.6cm} p{5.5cm} p{5.5cm}}
\toprule[1.5pt]
\textbf{Aspect} & \textbf{Classical Riesz Potential $I_\beta$} & \textbf{LCRP $I^{\boldsymbol{a},\boldsymbol{b}}_\beta$} \\
\midrule[1pt]
\textbf{Integral State} & \textbf{Divergence}  & \textbf{Convergence}  \\ [0.2cm]
\textbf{Math Reason} & Lack of decay factors causes infinite energy accumulation at infinity. & Oscillatory factors induce cancellation, regularizing the integral. \\ [0.2cm]
\textbf{Processing\newline Outcome} & Cannot compute potentials for images with non-zero backgrounds. & Effectively processes full-frame images with global backgrounds. \\
\bottomrule[1.5pt]
\end{tabular}
\end{table}

\begin{table}[H]
\caption{Mechanism of Key Sensitivity at Edges Based on Limit Discontinuity.}
\label{tab:sensitivity-mechanism}
\centering
\setlength{\tabcolsep}{0.183cm}
\begin{tabular}{p{2.708cm} p{5.45cm} p{5.45cm}}
\toprule[1.5pt]
\textbf{Aspect} & \textbf{Mathematical Phenomenon ($\beta \to 0$)} & \textbf{Cryptographic Consequence} \\
\midrule[1pt]
\textbf{Limit Outcome} & \textbf{Value Deviation}: Converges to geometric ratios (for example, $\frac{1}{2}$ for straight edges or $\frac{\alpha}{2\pi}$ for corners) instead of the true signal value $1$. & \textbf{Edge Destruction}: Sharp corners and sharp edges having no precise phase alignment blur into noise. \\ [0.7cm]
\textbf{Stability} & \textbf{Analytic Instability}: Reconstruction at singularities is discontinuous and highly sensitive to perturbations. & \textbf{High Key Sensitivity}: Guarantees theoretical robustness against key estimation attacks. \\
\bottomrule[1.5pt]
\end{tabular}
\end{table}	
\end{remark}
\subsection{Numerical Simulations\label{sec2.3}}
This subsection presents numerical simulations of the LCRP and the LCLO applied to images. Based on Theorem \ref{thm-frp}, Remark \ref{rm-LCRP}, Definition \ref{def-lclo}, and the fast LCT algorithm, together with the specially designed symbols for the LCRP and the LCLO (related numerical approaches can be found in the monograph \cite{ppst23}), we facilitate efficient computation of these operators. Through systematic simulation experiments using a controlled variable approach, we examines the effects of parameters $\boldsymbol{A}$ and $\beta$ in the LCRP $I^{\boldsymbol{A}}_\beta$ and the LCLO $\Delta_{\beta}^{\boldsymbol A}$ on images. The experimental results confirm that the LCRP and the LCLO operators achieve amplitude modulation in the LCT domain. Moreover, we identify quantitative patterns governing the influence of $\beta$ on imaging outcomes. This comparative analysis enhances the understanding of the LCRP and the LCLO operators, provides a basis for their applications in signal processing, and suggests new avenues for future research.

The numerical simulations in this section are performed using the test images displayed in Figure \ref{FIG3.1}(a) and (b). Subfigure (a) shows a $400 \times 400$ pixel grayscale image in two dimensions (for short, 2D), while Subfigure (b) presents the corresponding three-dimensional (for short, 3D) color representation. For continuous domain modeling, Subfigure (a) is represented by the \emph{Gaussian function} $G$ defined over the spatial domain $(x_1, x_2) \in [0, 400]^2$ as
$$G(x_1, x_2) := e^{ -\frac{(x_1 - 200)^2 + (x_2 - 200)^2}{2\sigma^2} },$$ 
where $\sigma = 50$ denotes the standard deviation  and $(x_1, x_2)$ denote spatial coordinates. The peak of the function $G$ is located at the center $(200, 200)$. The pixel intensities in Subfigure (a) are normalized to the range $[0, 1]$, a standard preprocessing step to ensure uniformity in visualization and analysis. Subfigures (c) and (d) in Figure \ref{FIG3.1} illustrate the amplitude distributions of the Gaussian image from Subfigure (a) in two different LCT domains. Specifically, the LCT parameter matrix used in Subfigure (c) is $\boldsymbol{A} := (A_1, A_2)$ with 
$A_1=\begin{bmatrix}{6}&{50}\\{0.7}&{6}
\end{bmatrix}$ and $A_2=\begin{bmatrix}
{3}&{400}\\{0.02}&{3}\end{bmatrix}$,
while for Subfigure (d) the parameter matrix is $\boldsymbol{B}:= (B_1, B_2)$ with
$B_1=\begin{bmatrix}{10}&{495}\\{0.2}&{10}
\end{bmatrix}$ and $B_2=\begin{bmatrix}
{1}&{20}\\{0}&{1}\end{bmatrix}$. 
A comparison of  Subfigures (c) and (d) clearly shows that the amplitude distribution of the same Gaussian image varies significantly under different LCT domains. This highlights the importance of selecting an appropriate LCT parameter matrix according to different image characteristics in practical applications.
 
\begin{figure}[H]
\centering
\subfigcapskip=-18pt
\subfigure[]{\includegraphics[width=0.37\linewidth]
{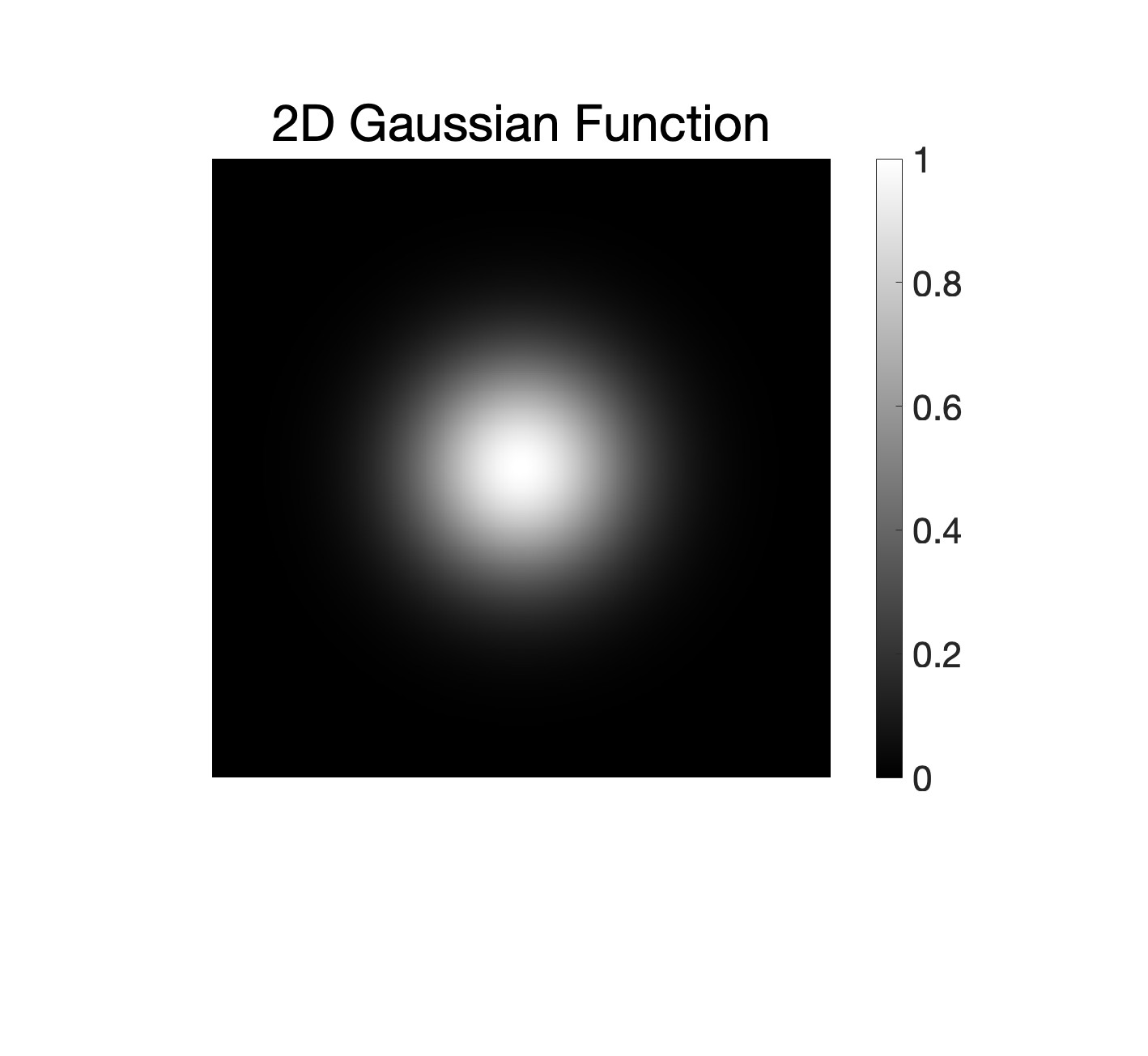}}\hspace{0.2cm}
\subfigure[]{\includegraphics[width=0.37\linewidth]
{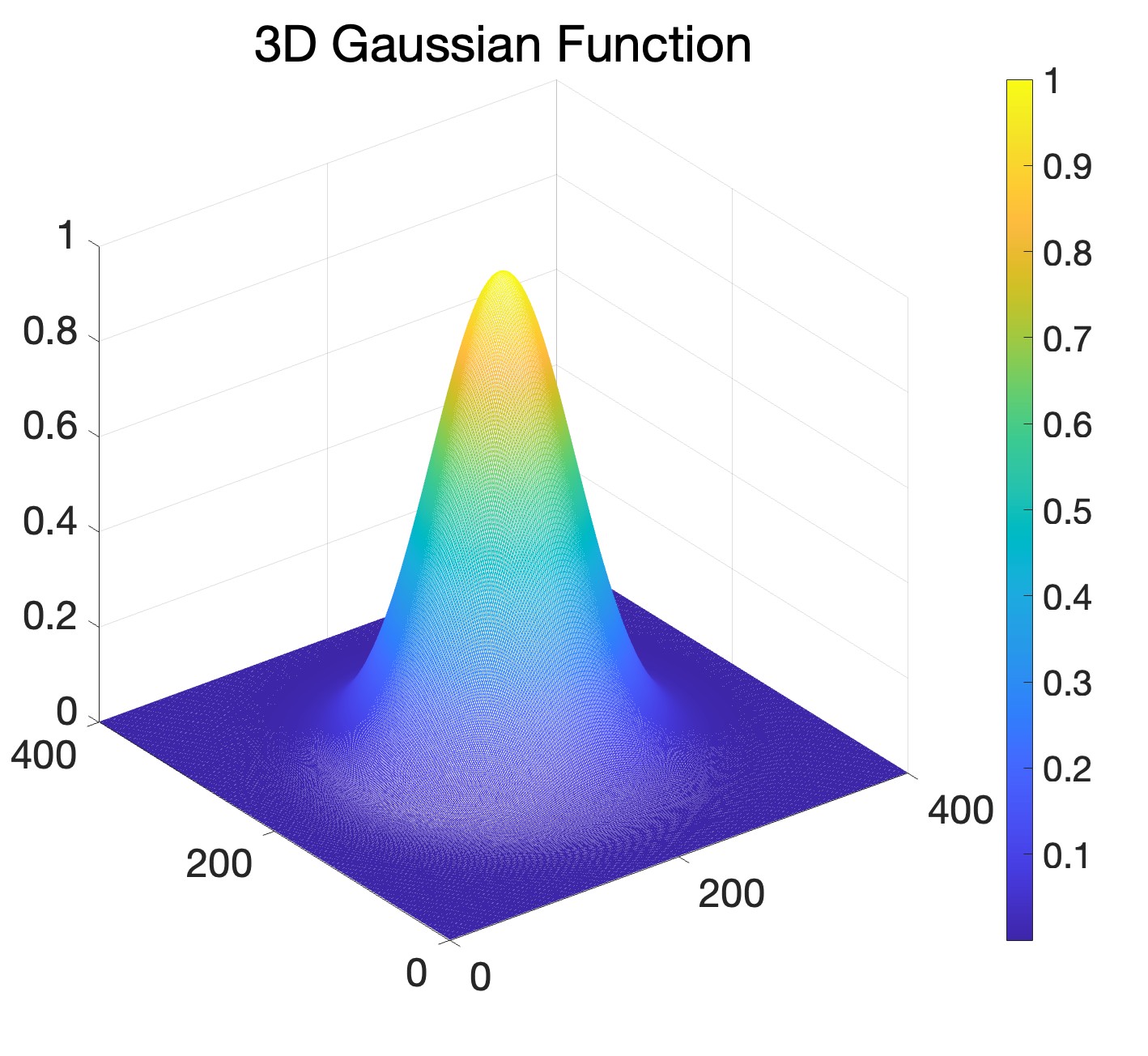}}\\ \vspace{0.5cm}
\subfigure[]{\includegraphics[width=0.37\linewidth]
{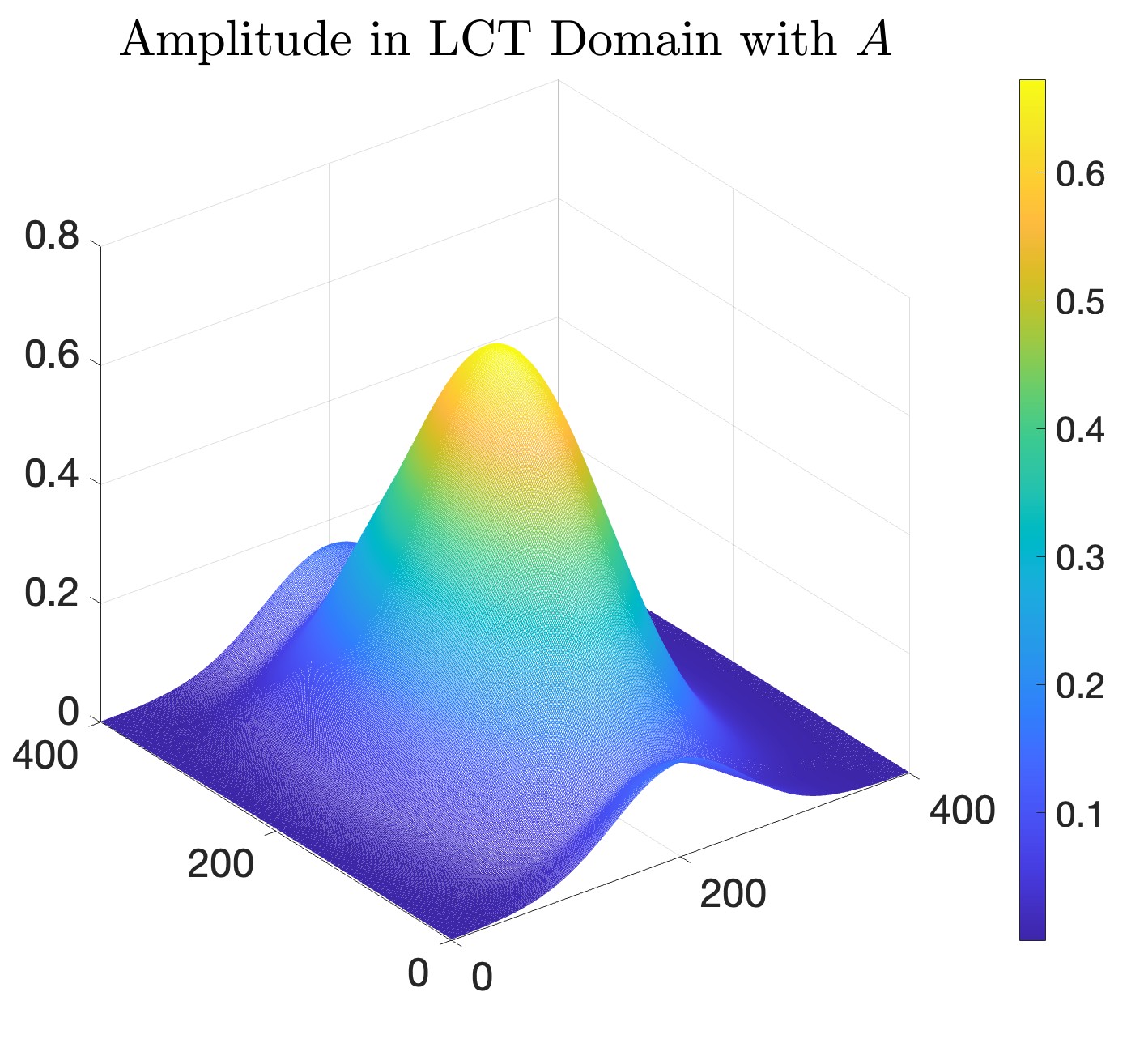}}\hspace{0.2cm}
\subfigure[]{\includegraphics[width=0.37\linewidth]
{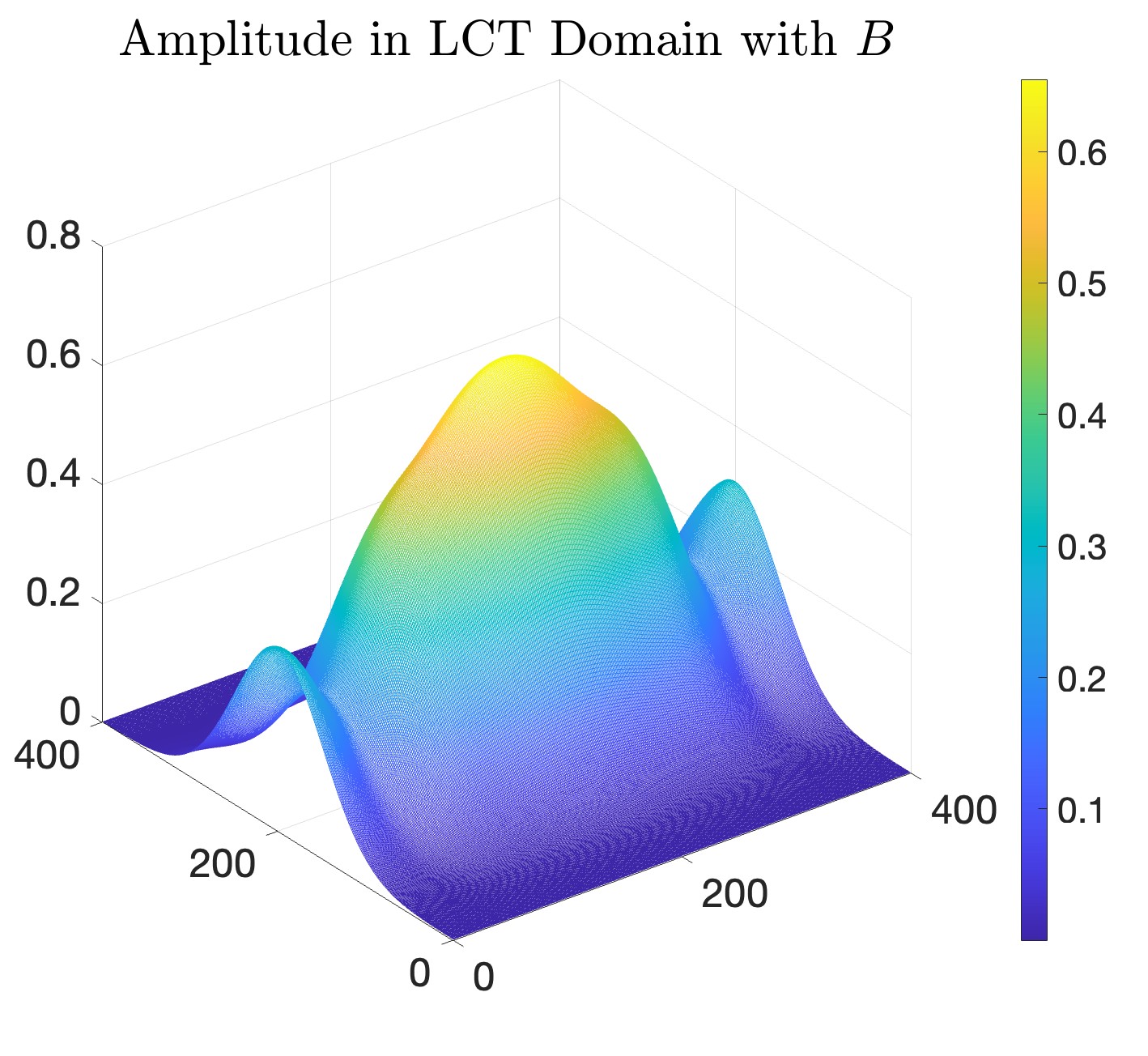}}
\vspace{-0.3cm}
\caption{Gaussian Function in the Spatial Domain 
and Different LCT Domains.}
\label{FIG3.1}
\end{figure}

 \begin{figure}[H]
\centering
\subfigcapskip=-10pt
\subfigure[]{\includegraphics[width=0.32\linewidth]
{ 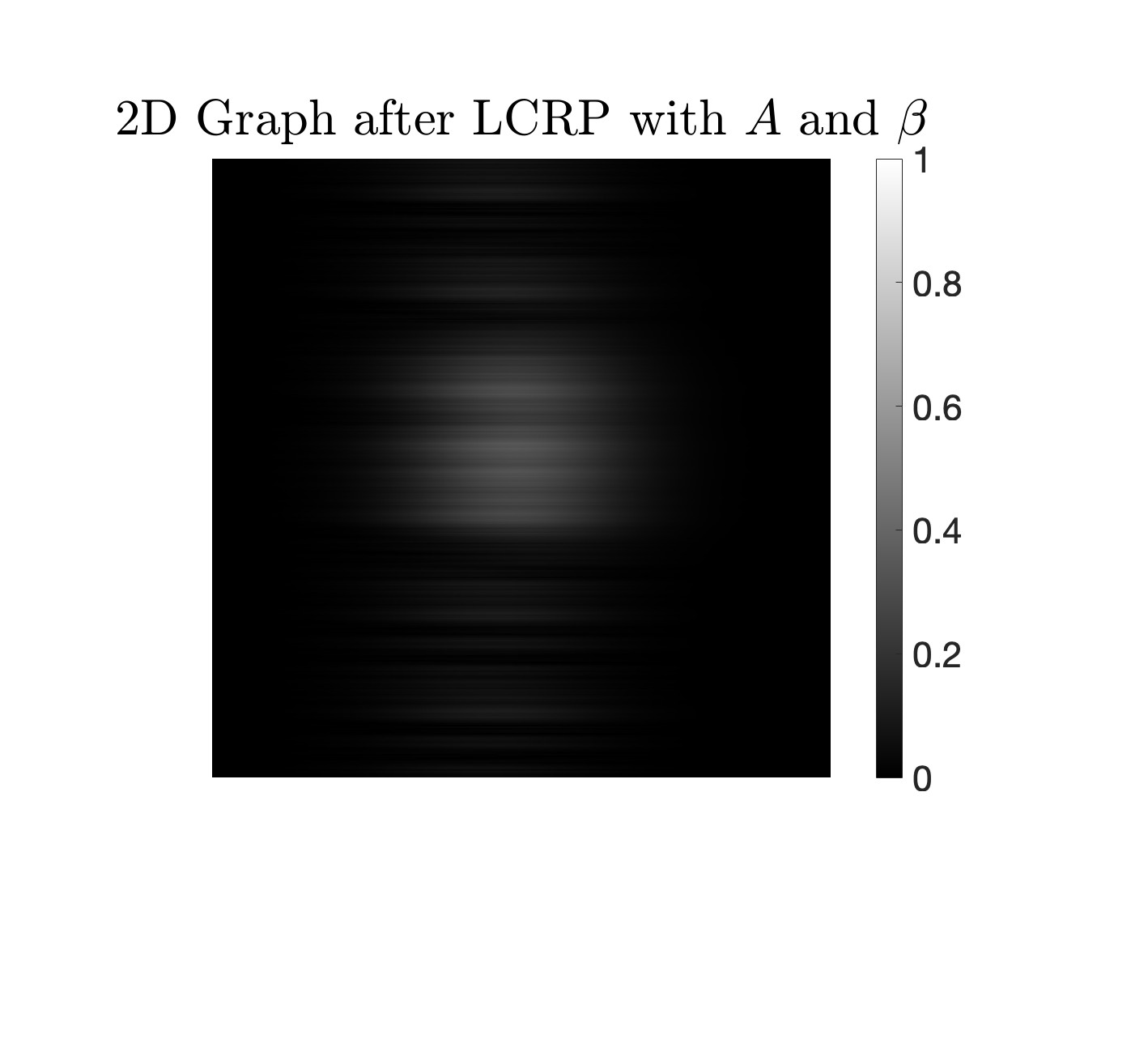}}\ \ \
\subfigure[]{\includegraphics[width=0.32\linewidth]
{ 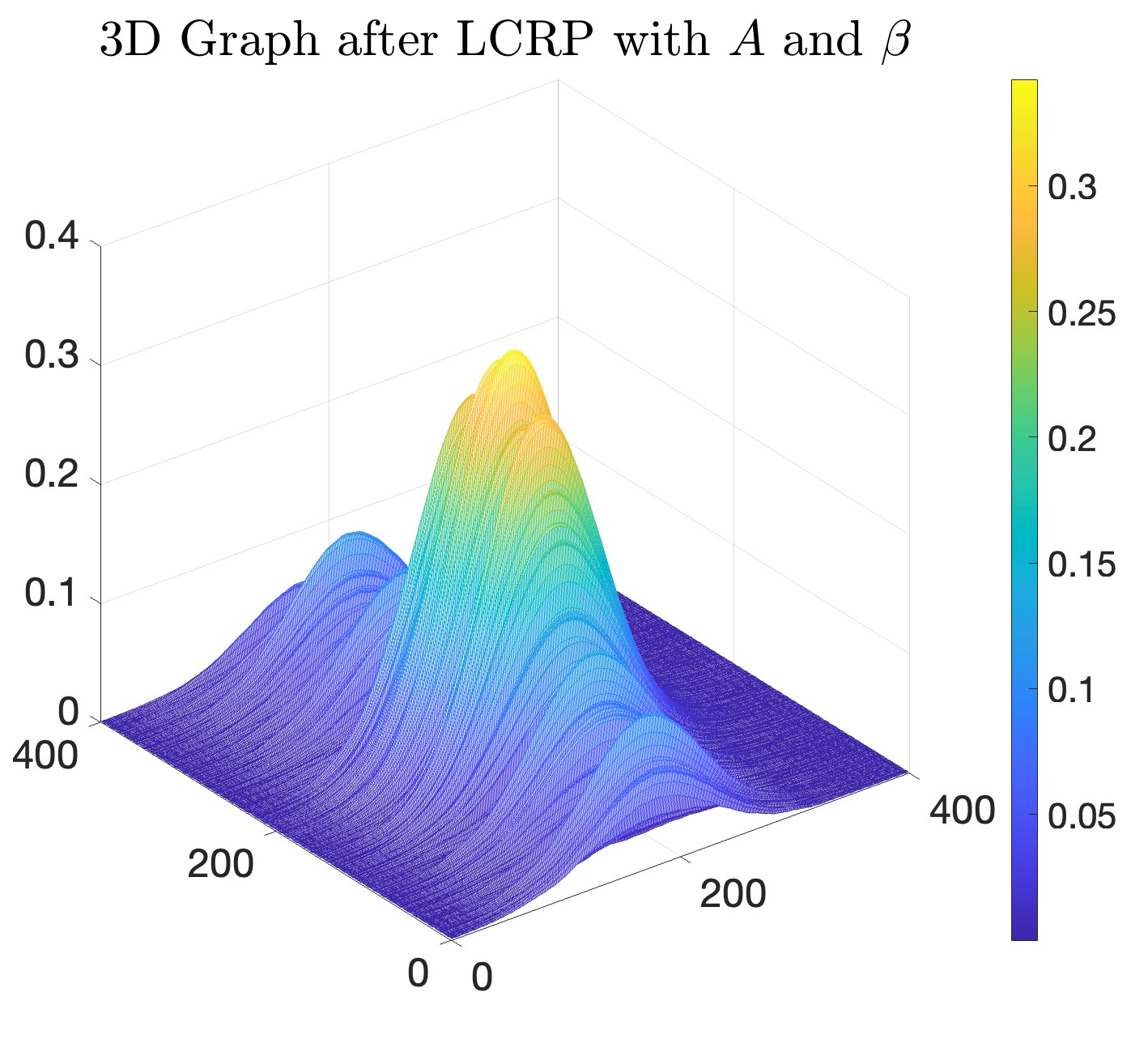}}\ \ \
\subfigure[]{\includegraphics[width=0.32\linewidth]
{ 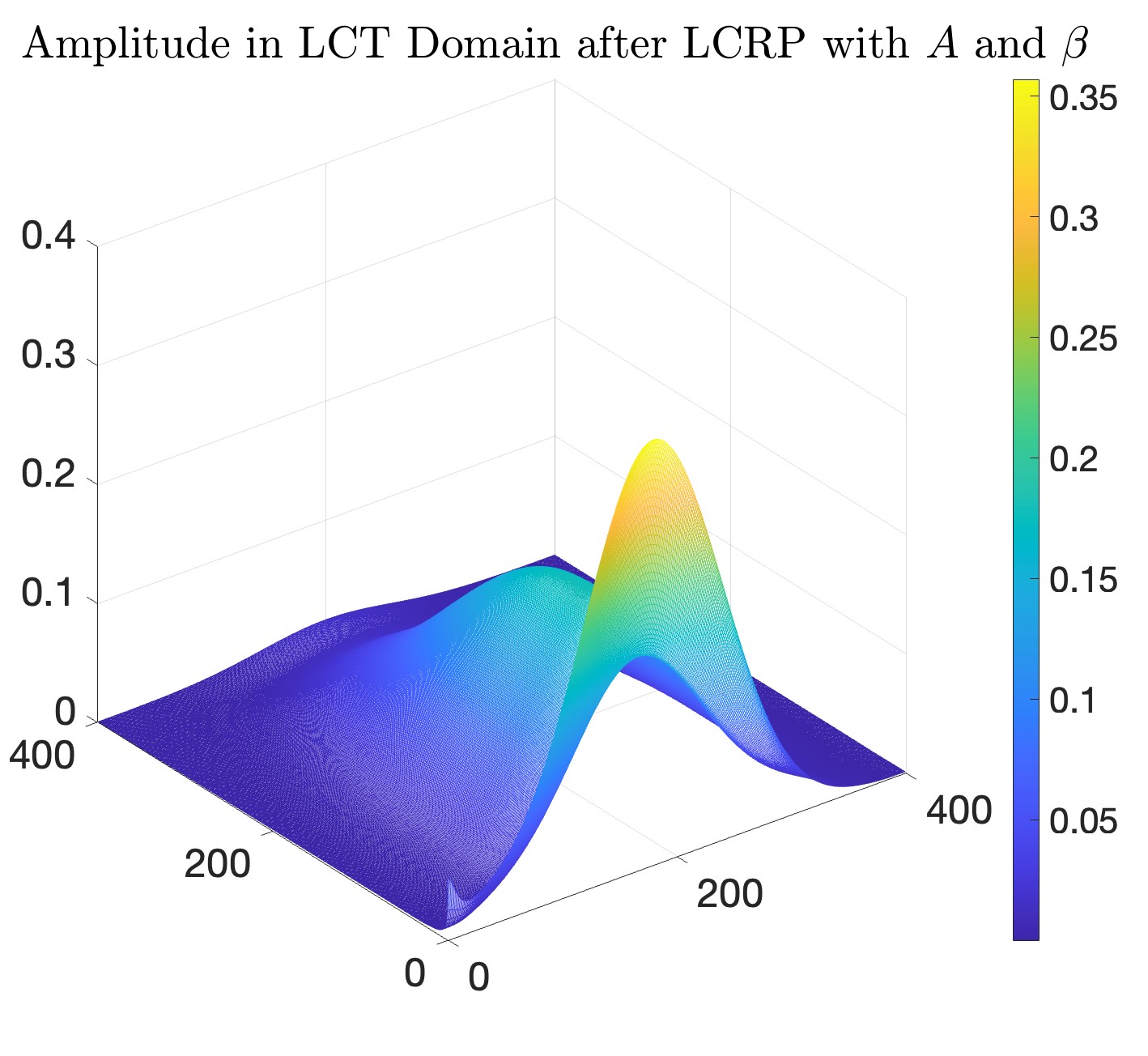}}\\\vspace{0.5cm}
\subfigure[]{\includegraphics[width=0.32\linewidth]
{ 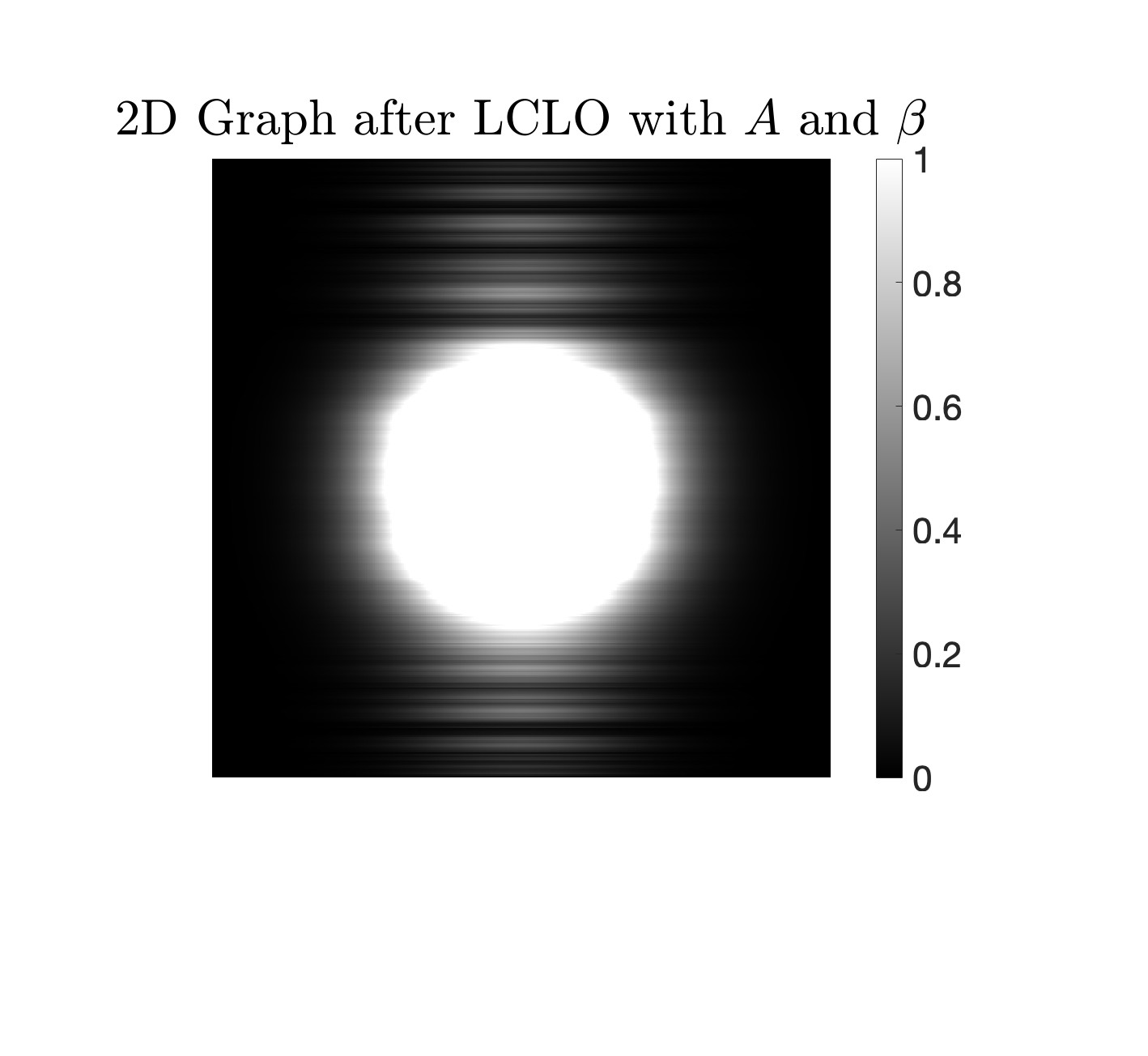}}\ \ \
\subfigure[]{\includegraphics[width=0.32\linewidth]
{ 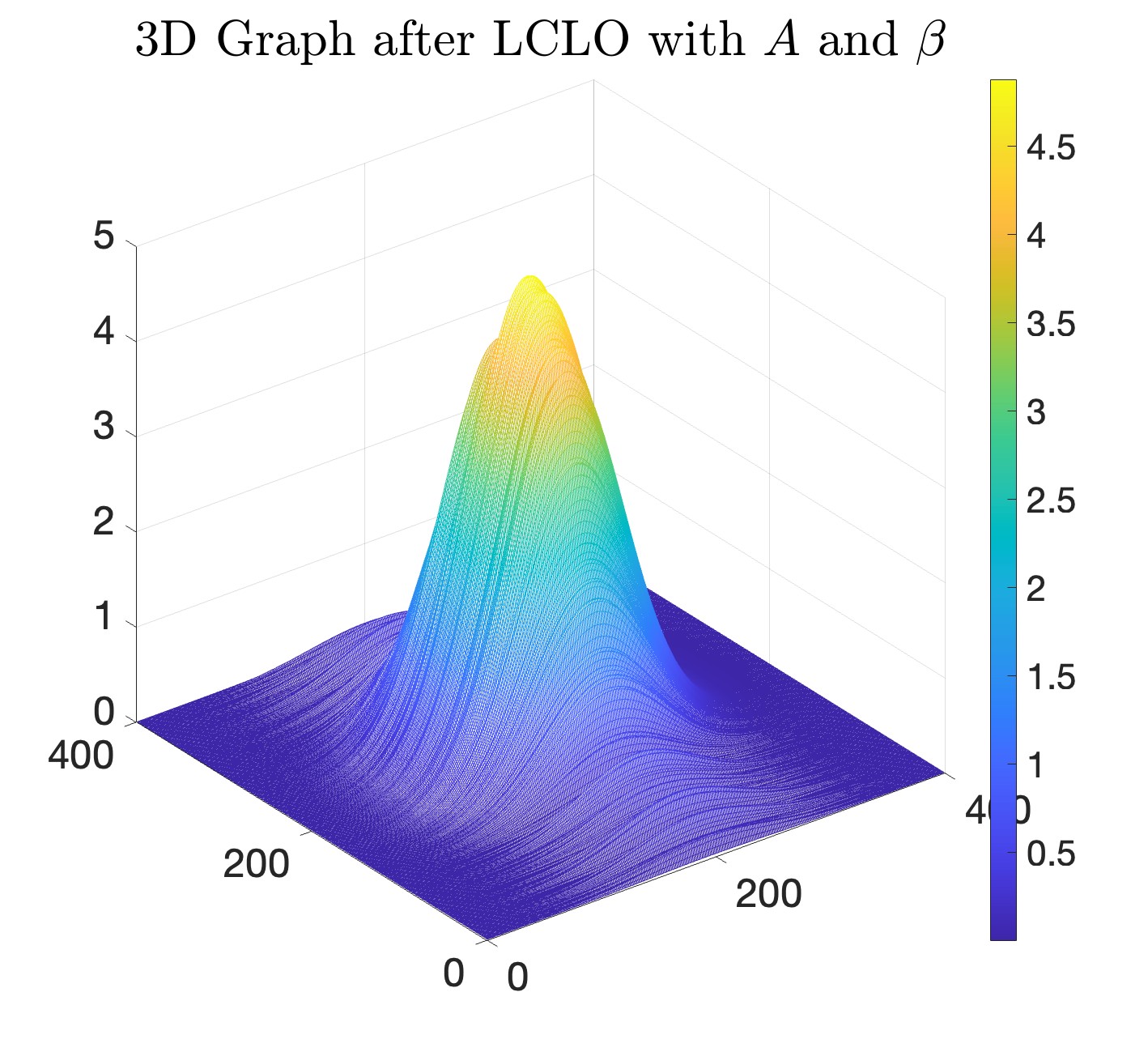}}\ \ \
\subfigure[]{\includegraphics[width=0.32\linewidth]
{ 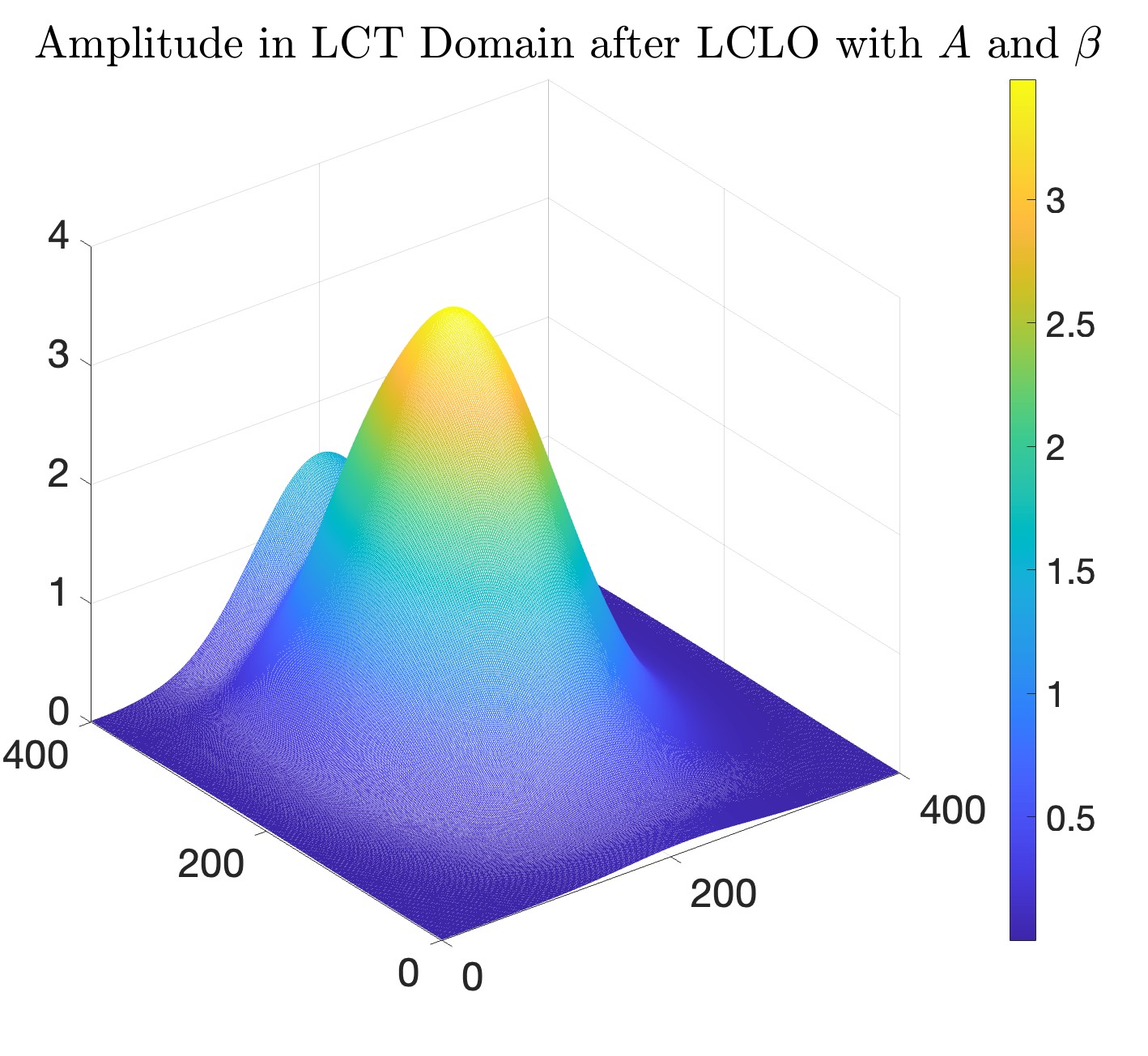}}
\vspace{-0.3cm}
\caption{LCRP and LCLO with $\boldsymbol A$ and $\beta=1.1$.}
\label{FIG3.2}
\end{figure}
 
\begin{figure}[H]
\centering
\subfigcapskip=-10pt
\subfigure[]{\includegraphics[width=0.321\linewidth]
{ 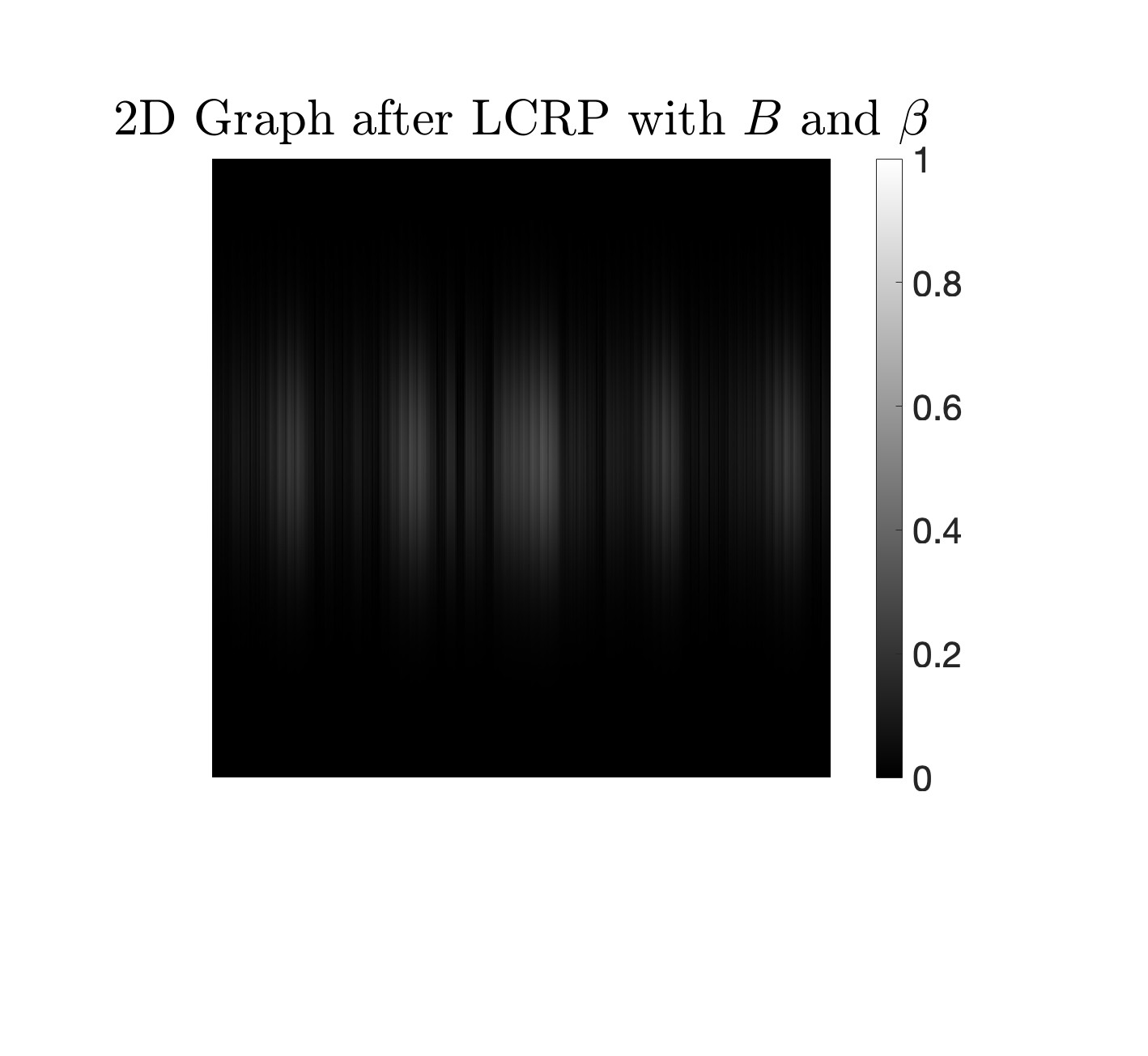}}\ \ \
\subfigure[]{\includegraphics[width=0.321\linewidth]
{ 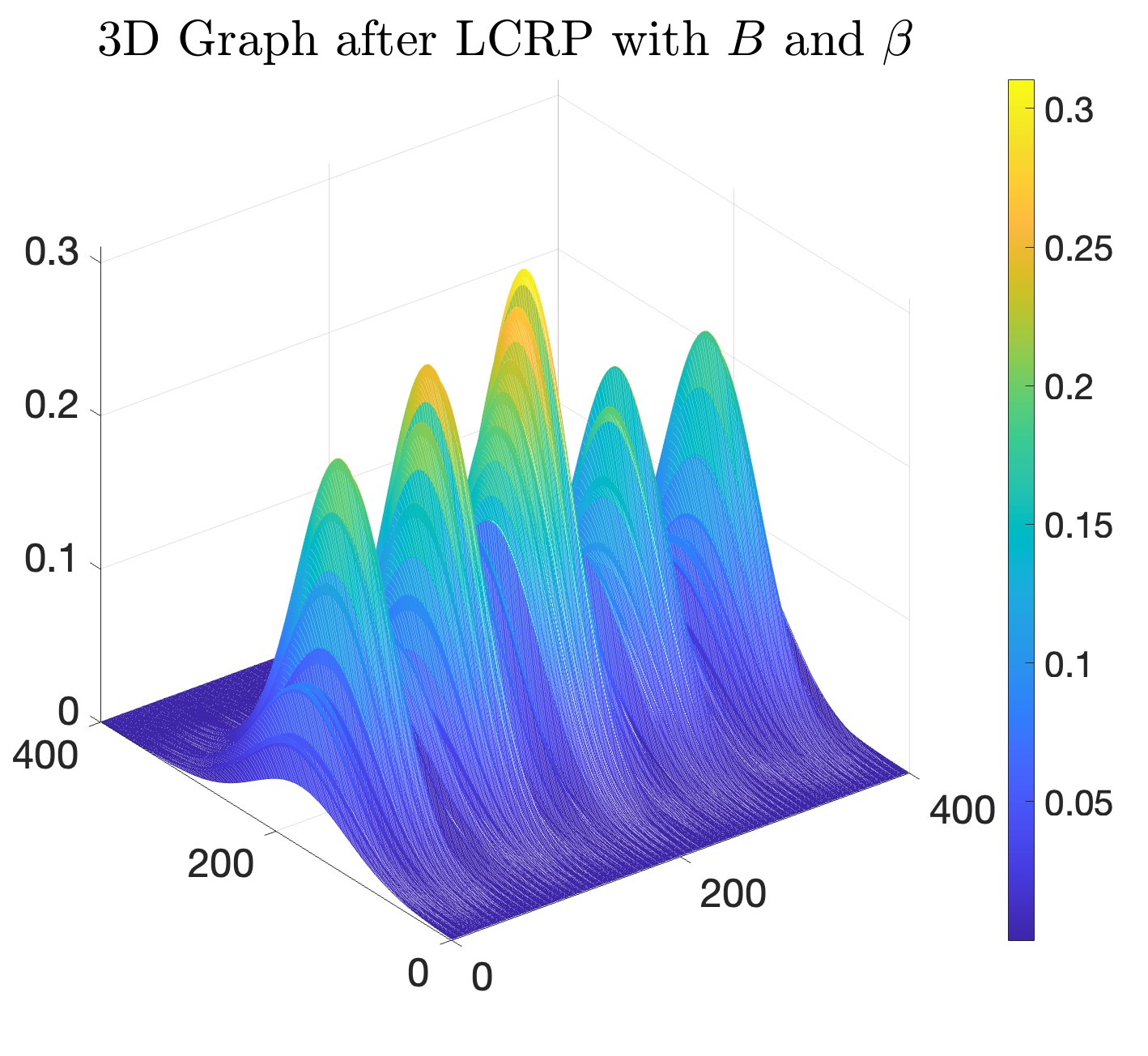}}\ \ \
\subfigure[]{\includegraphics[width=0.321\linewidth]
{ 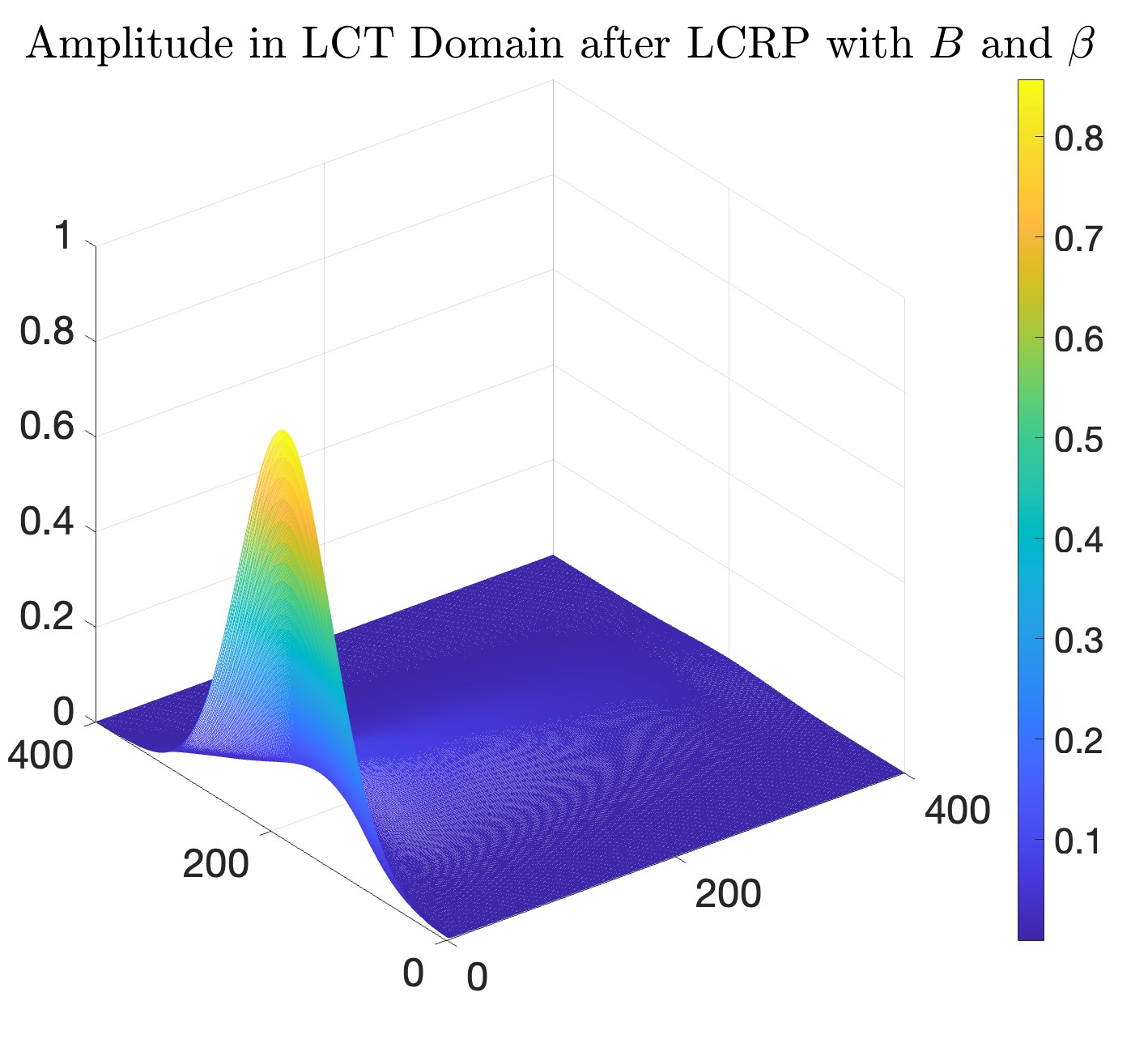}}\\
\vspace{0.5cm}
\subfigure[]{\includegraphics[width=0.321\linewidth]
{ 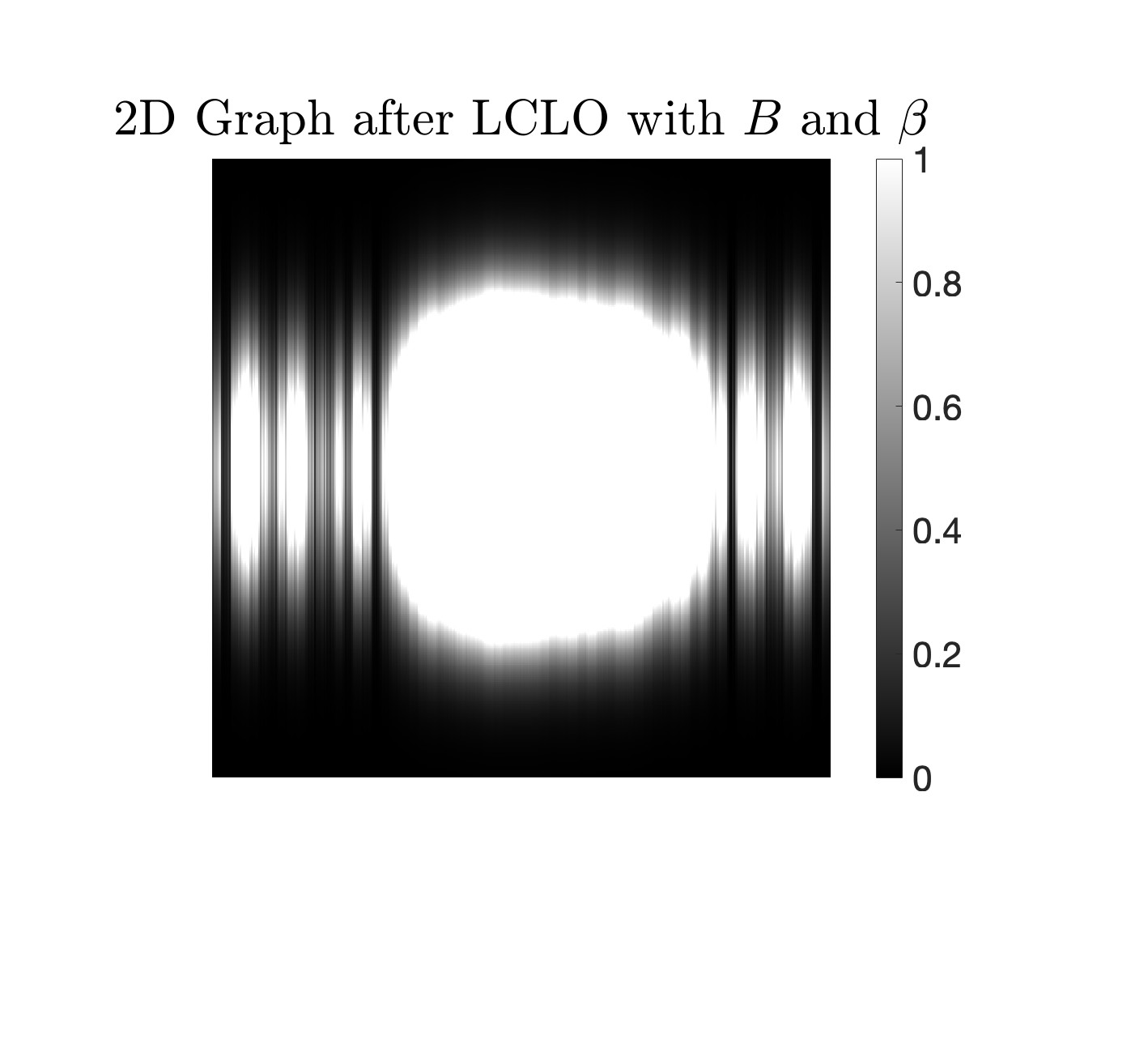}}\ \ \
\subfigure[]{\includegraphics[width=0.321\linewidth]
{ 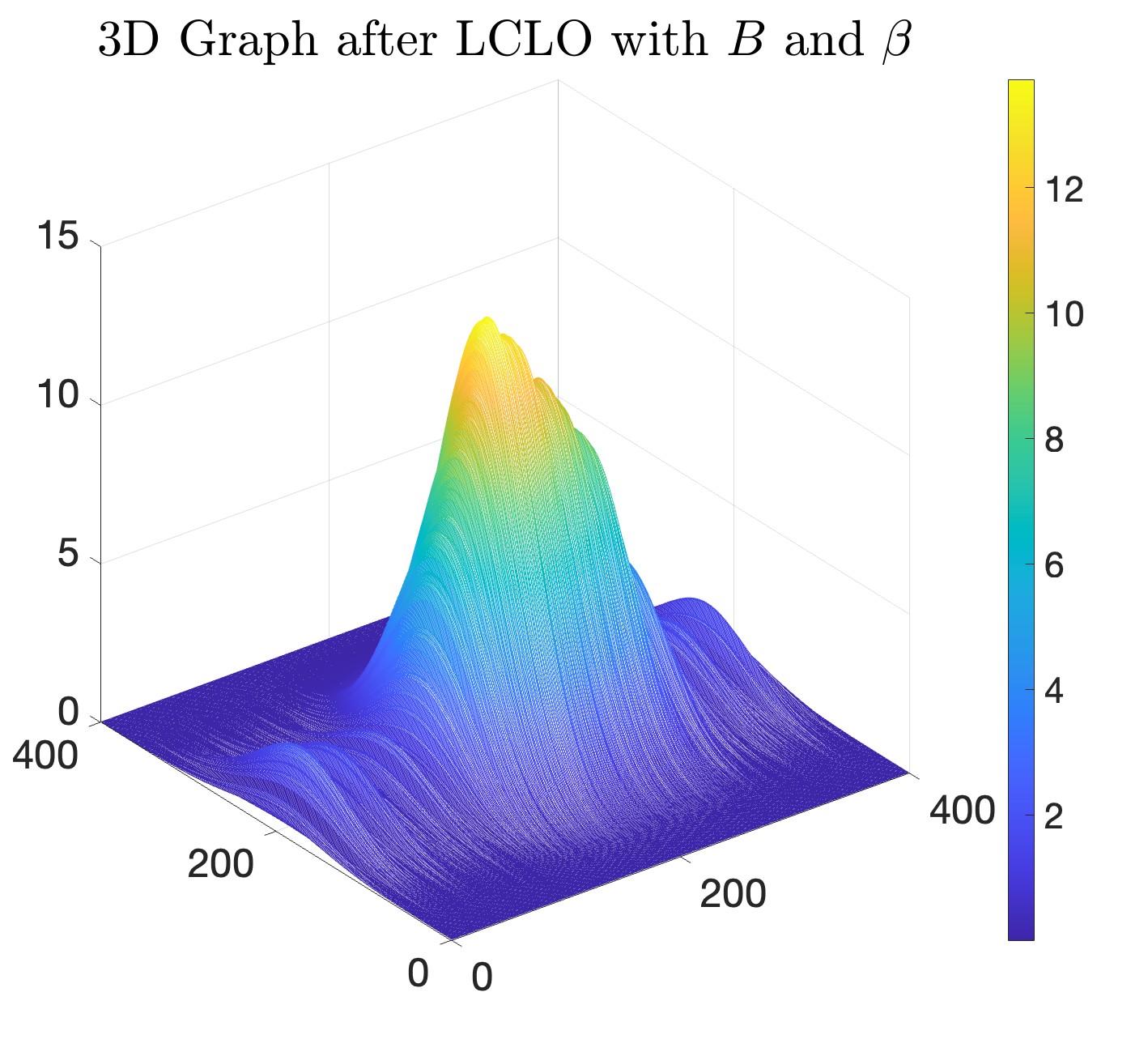}}\ \ \
\subfigure[]{\includegraphics[width=0.321\linewidth]
{ 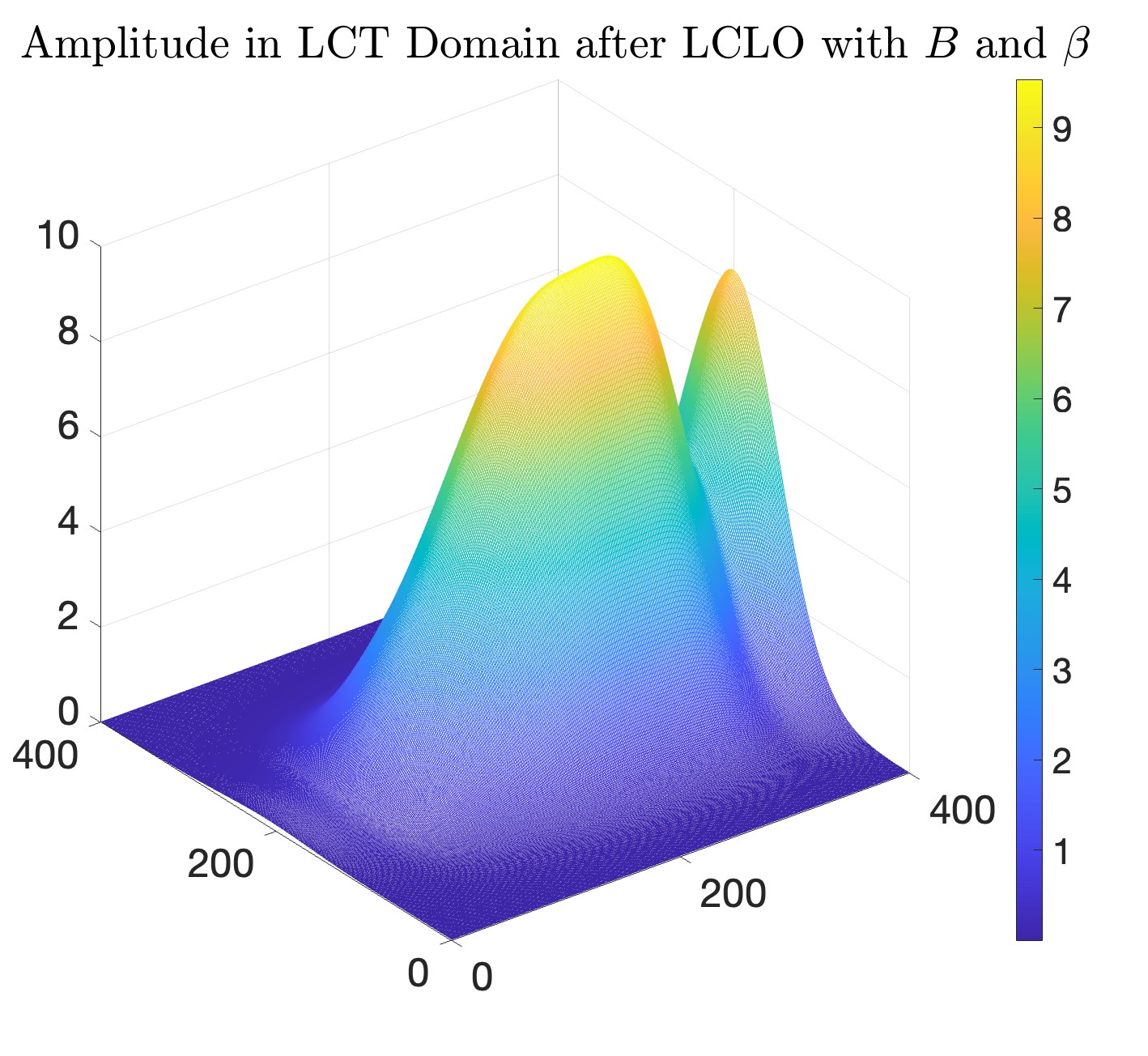}}
\vspace{-0.3cm}
\caption{LCRP and LCLO with $\boldsymbol B$ and $\beta=1.1$.}
\label{FIG3.3}
\end{figure}

Figure \ref{FIG3.2} shows the numerical simulation  of the Gaussian function processed by the LCRP $ I^{\boldsymbol{A}}_\beta$ and the LCLO $ \Delta_{\beta}^{\boldsymbol{A}}$, where  $\boldsymbol{A}$ is the same as in  Figure \ref{FIG3.1} and $\beta = 1.1$. Subfigure (a) displays the 2D grayscale image of the Gaussian function after LCRP processing, Subfigure (b) presents the corresponding 3D color representation of Subfigure (a), and Subfigure (c) illustrates the amplitude distribution of Subfigure (a) in the corresponding LCT domain. Subfigure (d) shows the 2D grayscale image of the Gaussian function after LCLO processing, Subfigure (e) provides the 3D color representation of Subfigure (d), and Subfigure (f) illustrates the amplitude distribution of Subfigure (d) in the corresponding LCT domain.
  
\begin{figure}[H]
\centering
\subfigcapskip=-10pt
\subfigure[]{\includegraphics[width=0.32\linewidth]
{ 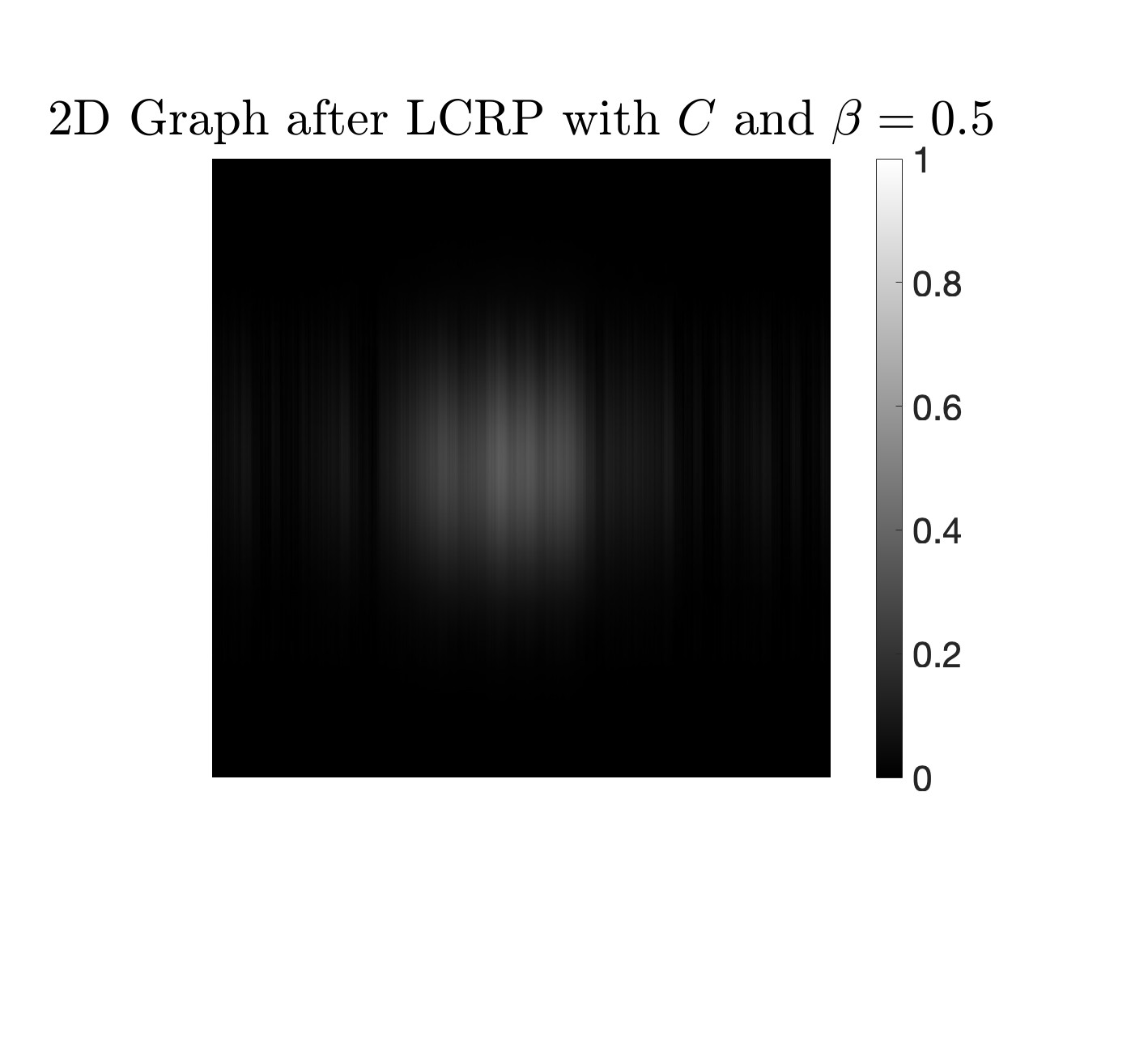}}\ \ \
\subfigure[]{\includegraphics[width=0.32\linewidth]
{ 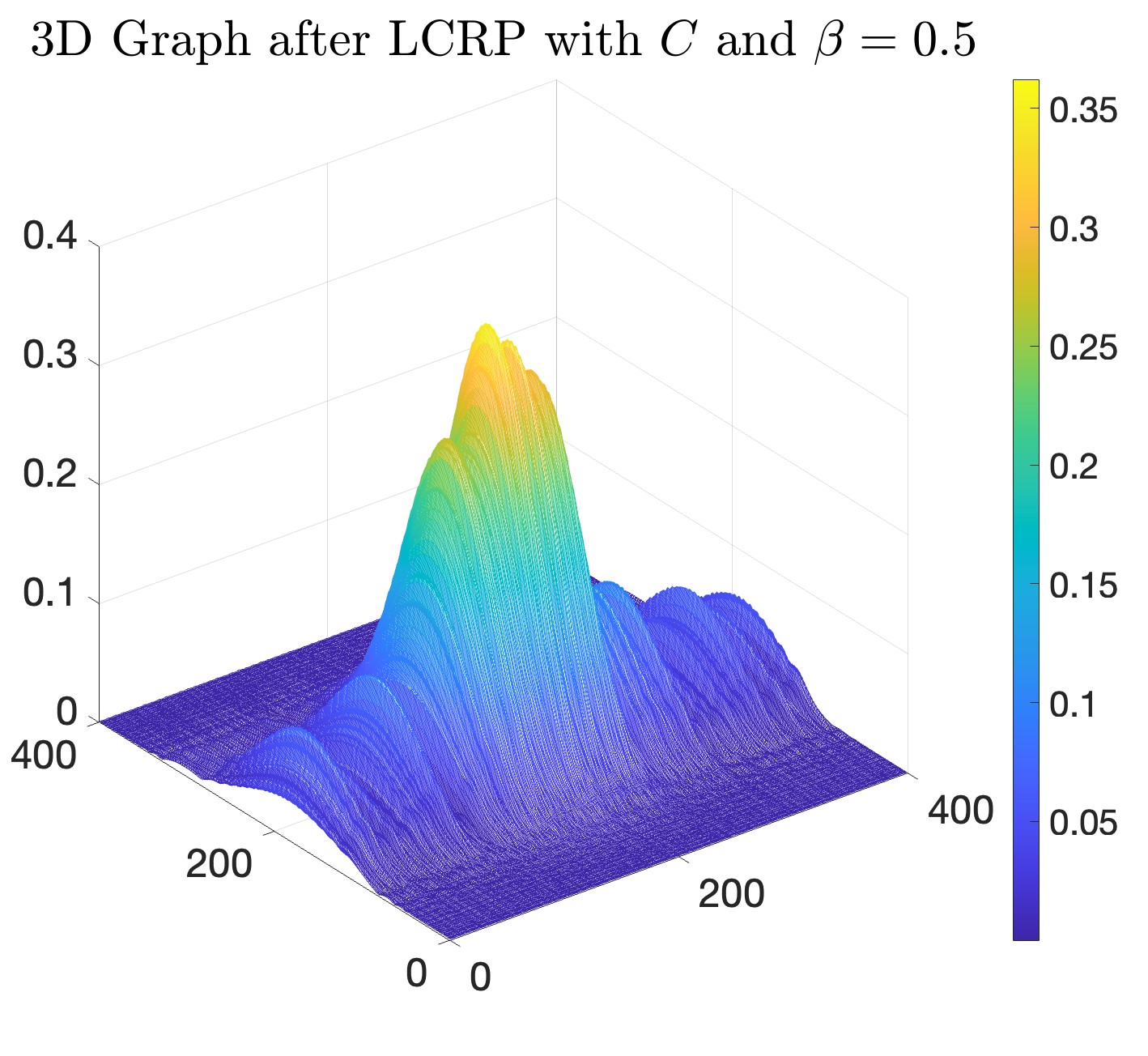}}\ \ \
\subfigure[]{\includegraphics[width=0.32\linewidth]
{ 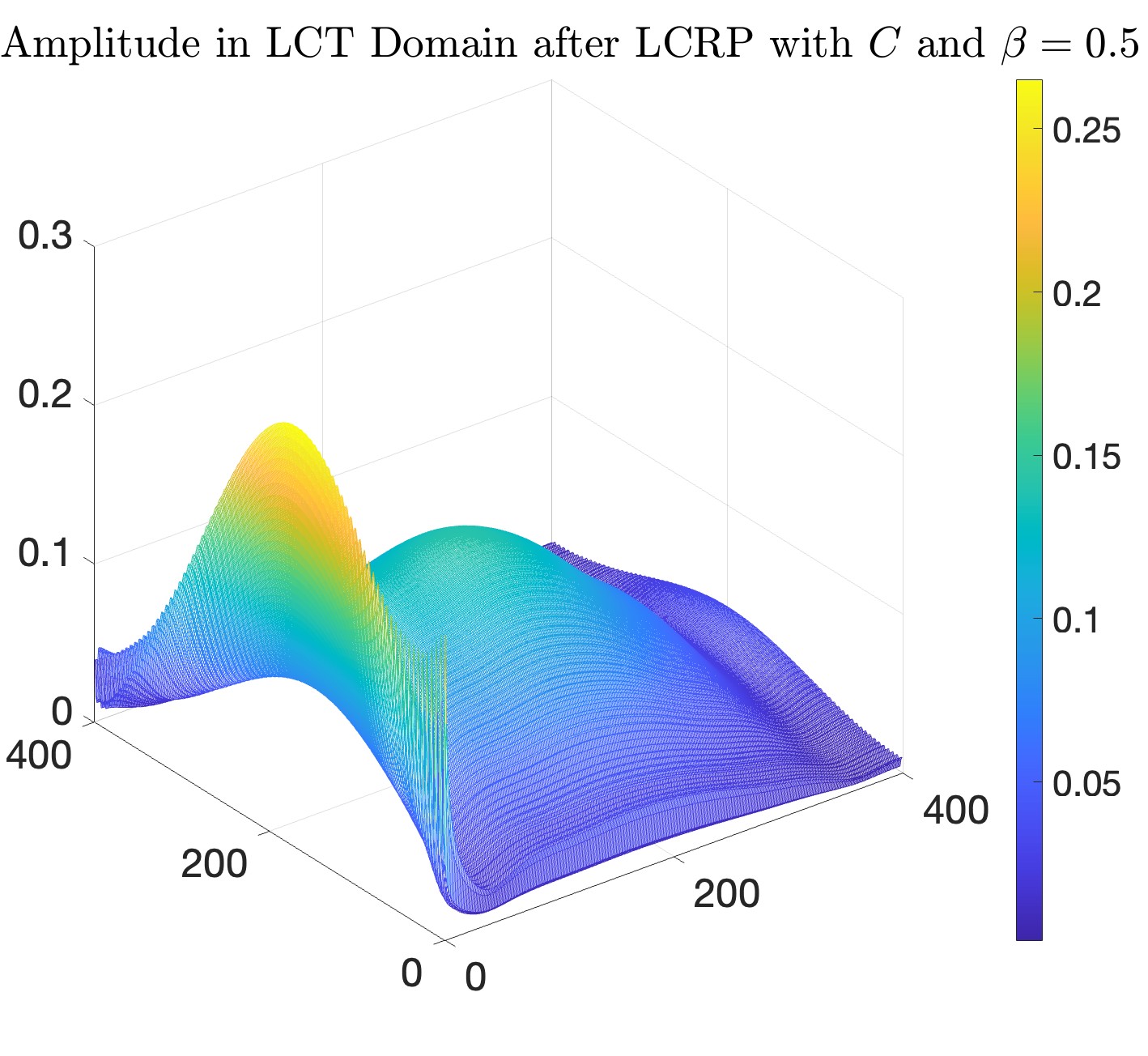}}\\
\vspace{0.5cm}
\subfigure[]{\includegraphics[width=0.32\linewidth]
{ 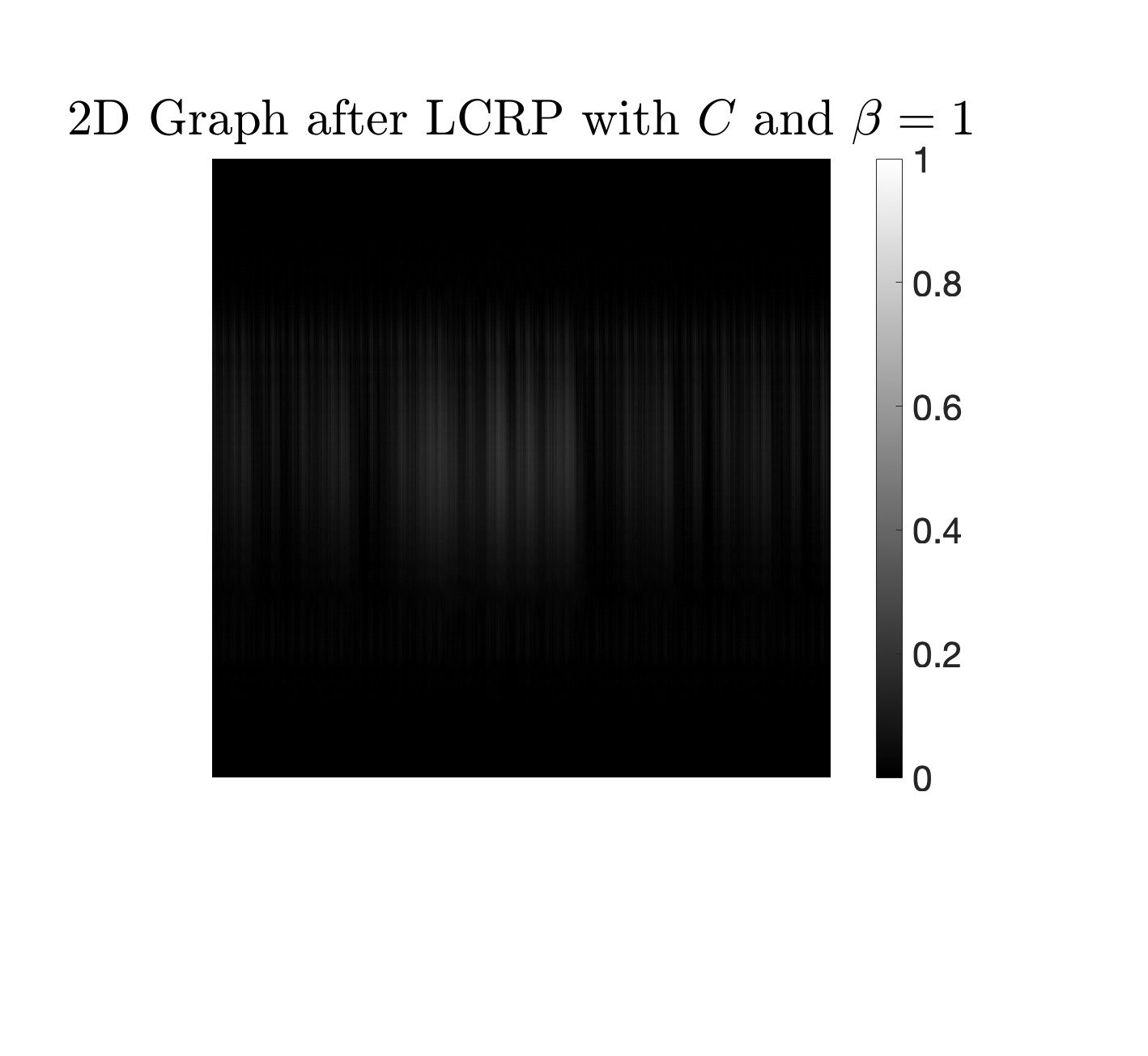}}\ \ \
\subfigure[]{\includegraphics[width=0.32\linewidth]
{ 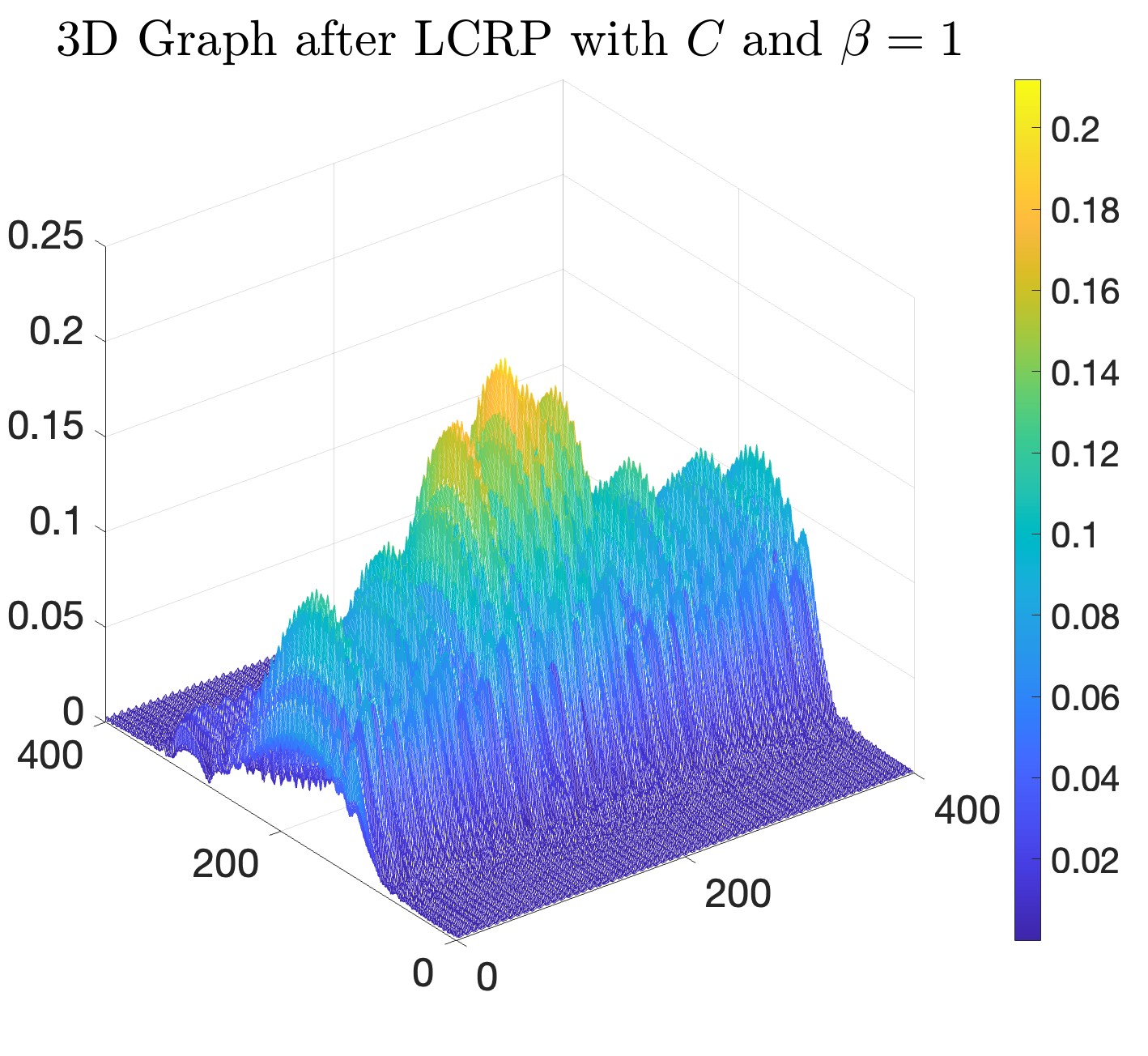}}\ \ \
\subfigure[]{\includegraphics[width=0.32\linewidth]
{ 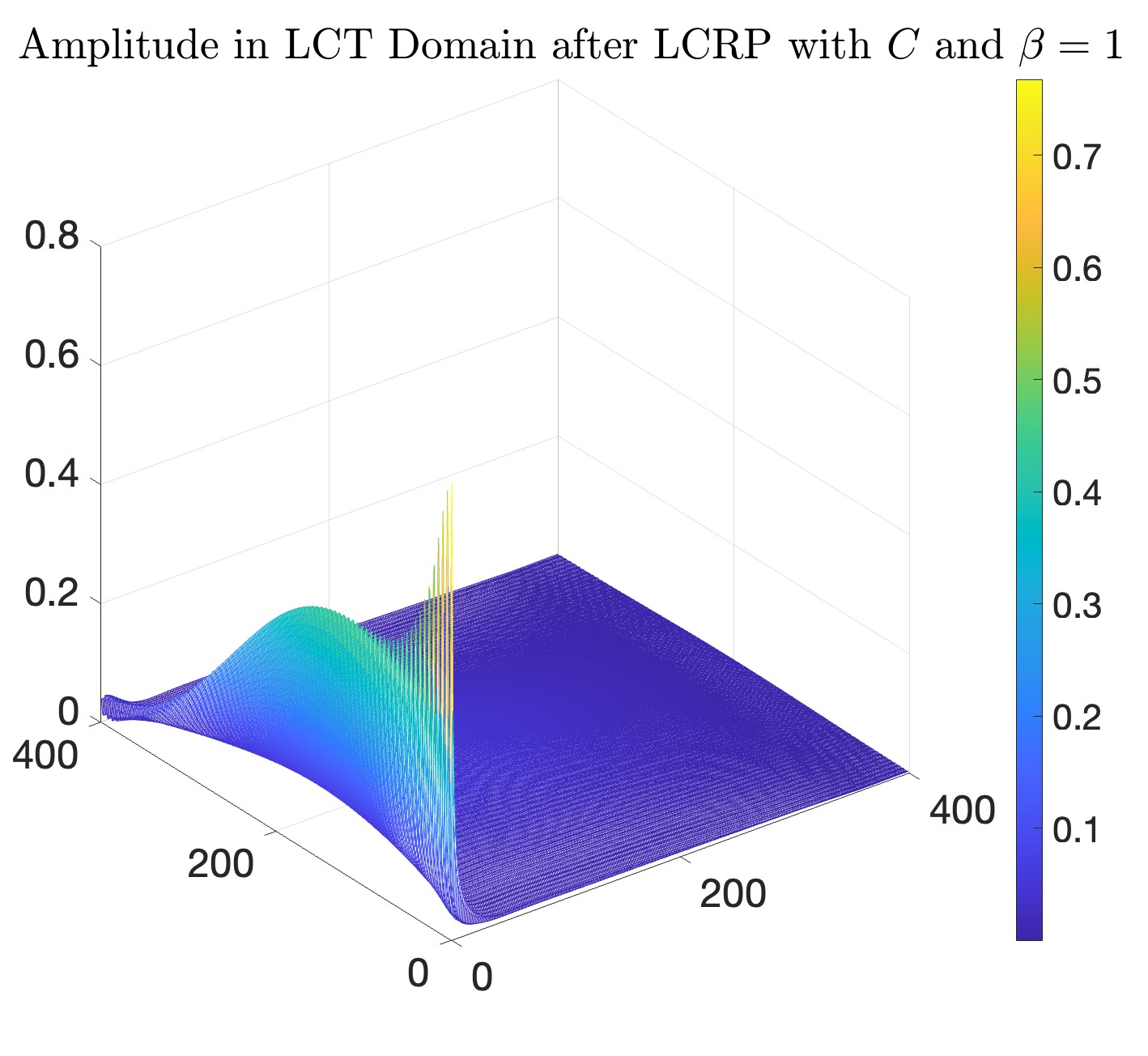}}\\
\vspace{0.5cm}
\subfigure[]{\includegraphics[width=0.32\linewidth]
{ 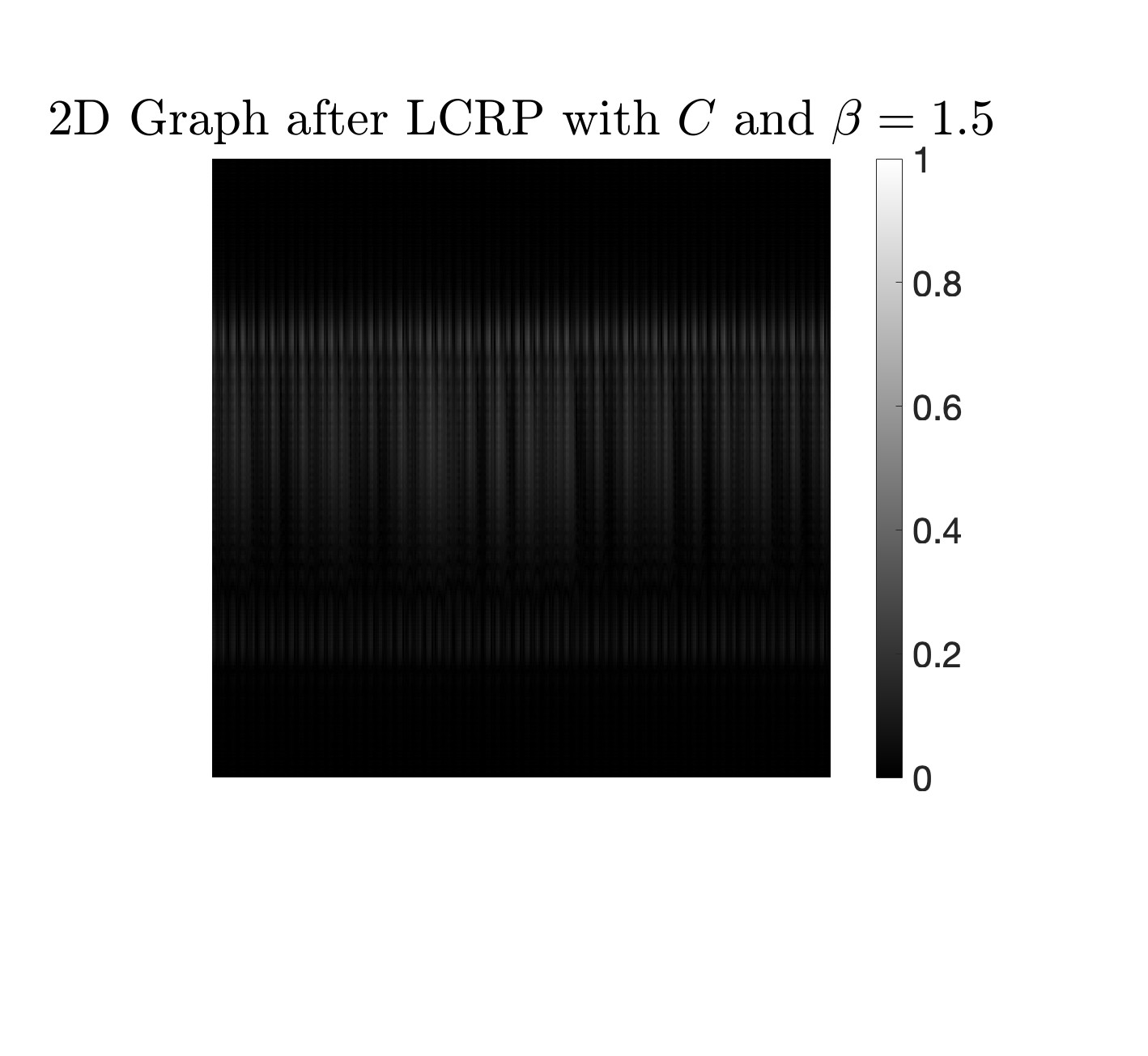}}\ \ \
\subfigure[]{\includegraphics[width=0.32\linewidth]
{ 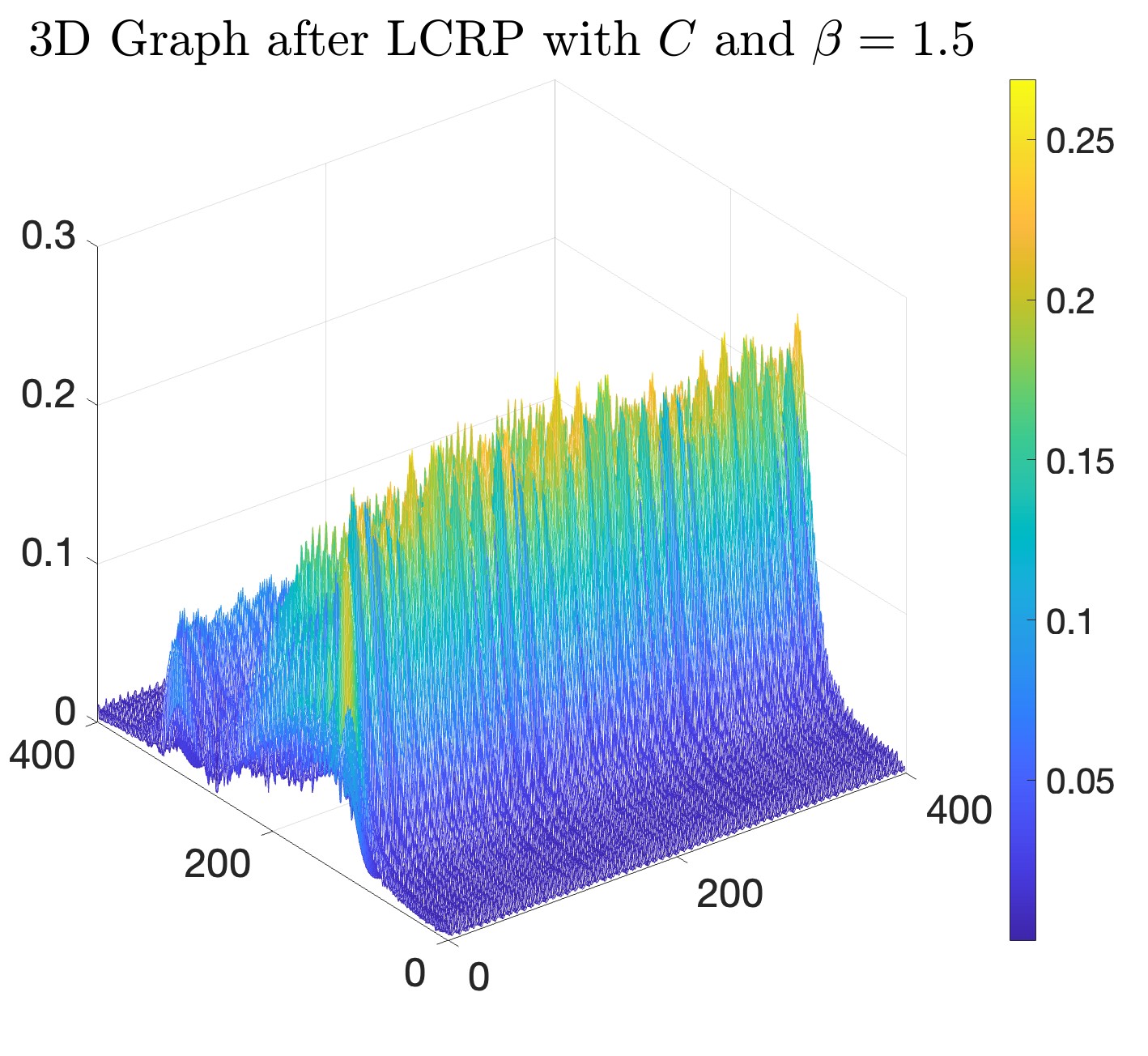}}\ \ \
\subfigure[]{\includegraphics[width=0.32\linewidth]
{ 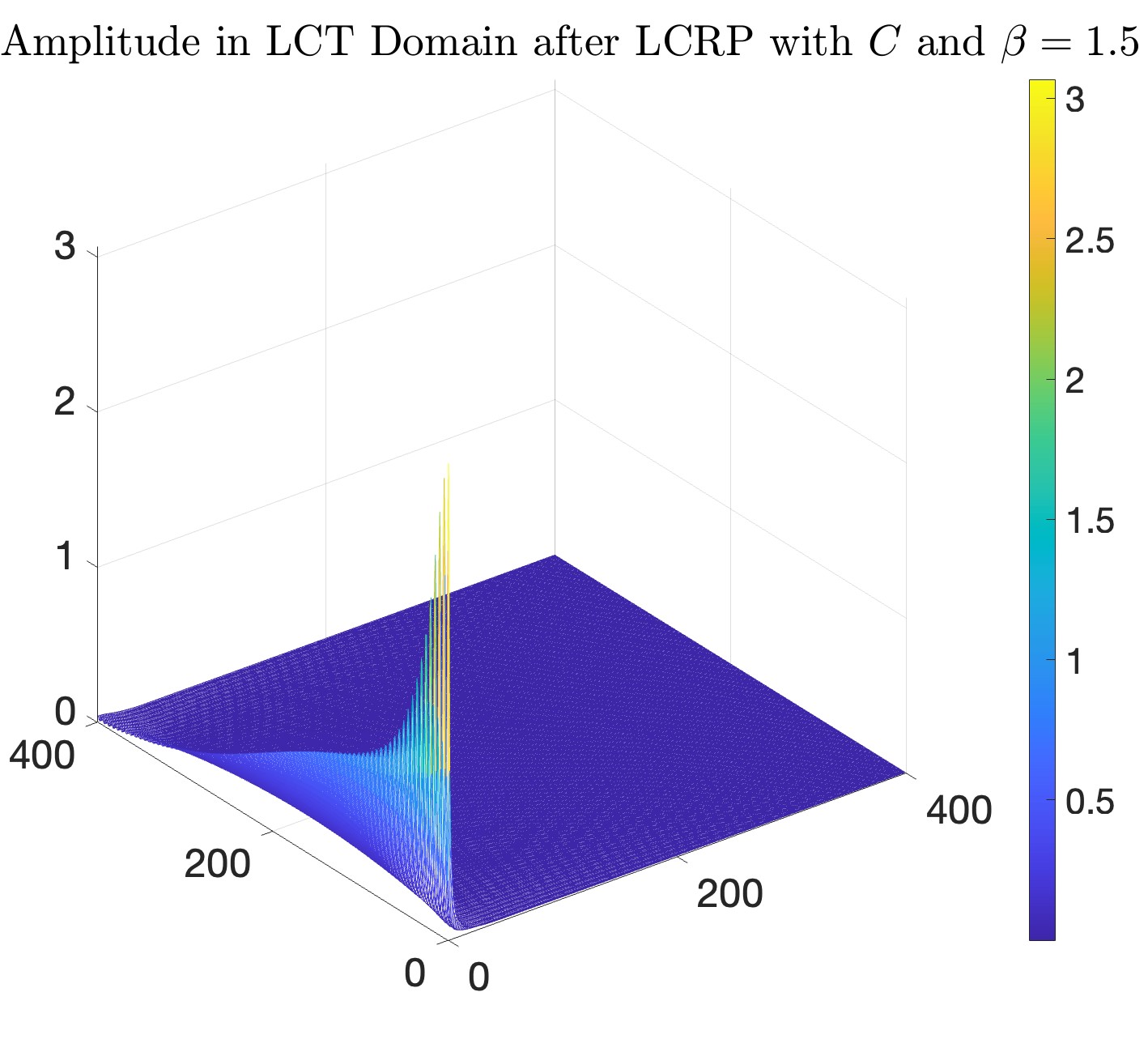}}
\vspace{-0.3cm}
\caption{LCRP with $\boldsymbol C$ and, respectively,  
$\beta=0.5$, $\beta=1$, and $\beta=1.5$.}
\label{FIG3.4}
\end{figure}
Figure \ref{FIG3.3} presents the numerical simulation of the Gaussian function processed by the LCRP $ I^{\boldsymbol{B}}_\beta$ and the LCLO $ \Delta_{\beta}^{\boldsymbol{B}}$, where  $\boldsymbol{B}$ is the same as in Figure \ref{FIG3.1} and $\beta=1.1$. Subfigure (a) displays the 2D grayscale image of the Gaussian function after LCRP processing, Subfigure (b) presents the corresponding 3D color representation of Subfigure (a), and Subfigure (c) illustrates the amplitude distribution of Subfigure (a) in the corresponding LCT domain. Subfigure (d) shows the 2D grayscale image of the Gaussian function after LCLO processing, Subfigure (e) provides the 3D color representation of Subfigure (d), and Subfigure (f) illustrates the amplitude distribution of Subfigure (d) in the corresponding LCT domain. 

Figures \ref{FIG3.2} and \ref{FIG3.3} demonstrate that, when the parameter $\beta$ either in the LCRP $I^{\boldsymbol{A}}_\beta$ or in the LCLO $\Delta_{\beta}^{\boldsymbol A}$ is fixed, while the parameter matrix $\boldsymbol A$ varies, significant differences emerge when applied  both operators to images. These distinctions are observed primarily in two aspects: the spatial domain grayscale representations and the corresponding amplitude distributions in the LCT domain. Furthermore, comparisons between Figure \ref{FIG3.2}(c), (f) and Figure \ref{FIG3.1}(c), as well as between Figure \ref{FIG3.3}(c), (f) and Figure \ref{FIG3.1}(d) reveal that the Gaussian functions processed by the LCRP and the LCLO exhibit complementary amplitude distributions in the LCT domain. This implies that, by choosing specific parameters, the LCRP and the LCLO may be regarded as mutual inverse operators. 
  
\begin{figure}[H]
\centering
\subfigcapskip=-8pt
\subfigure[]{\includegraphics[width=0.32\linewidth]
	{ 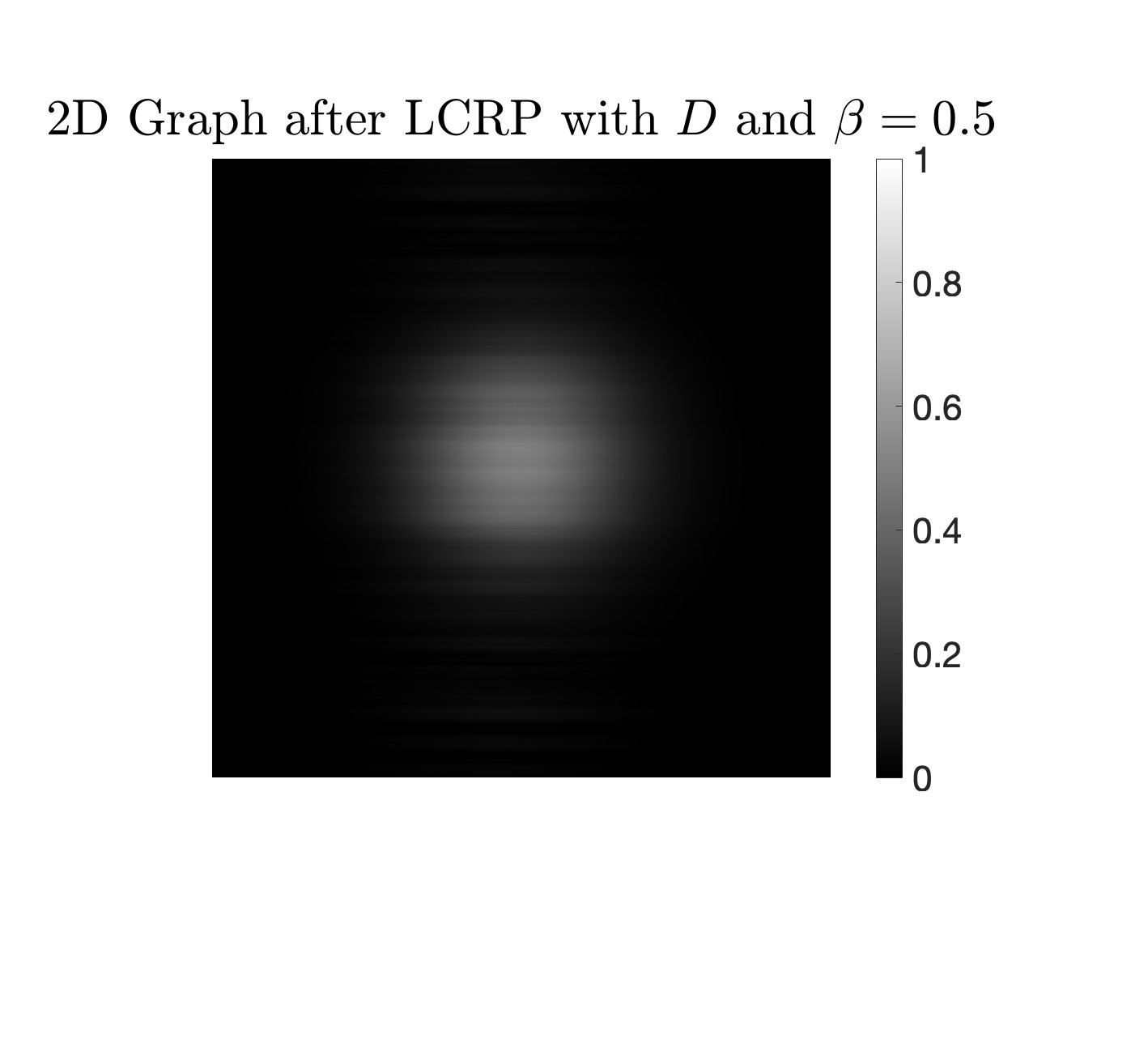}}\ \ \
\subfigure[]{\includegraphics[width=0.32\linewidth]
	{ 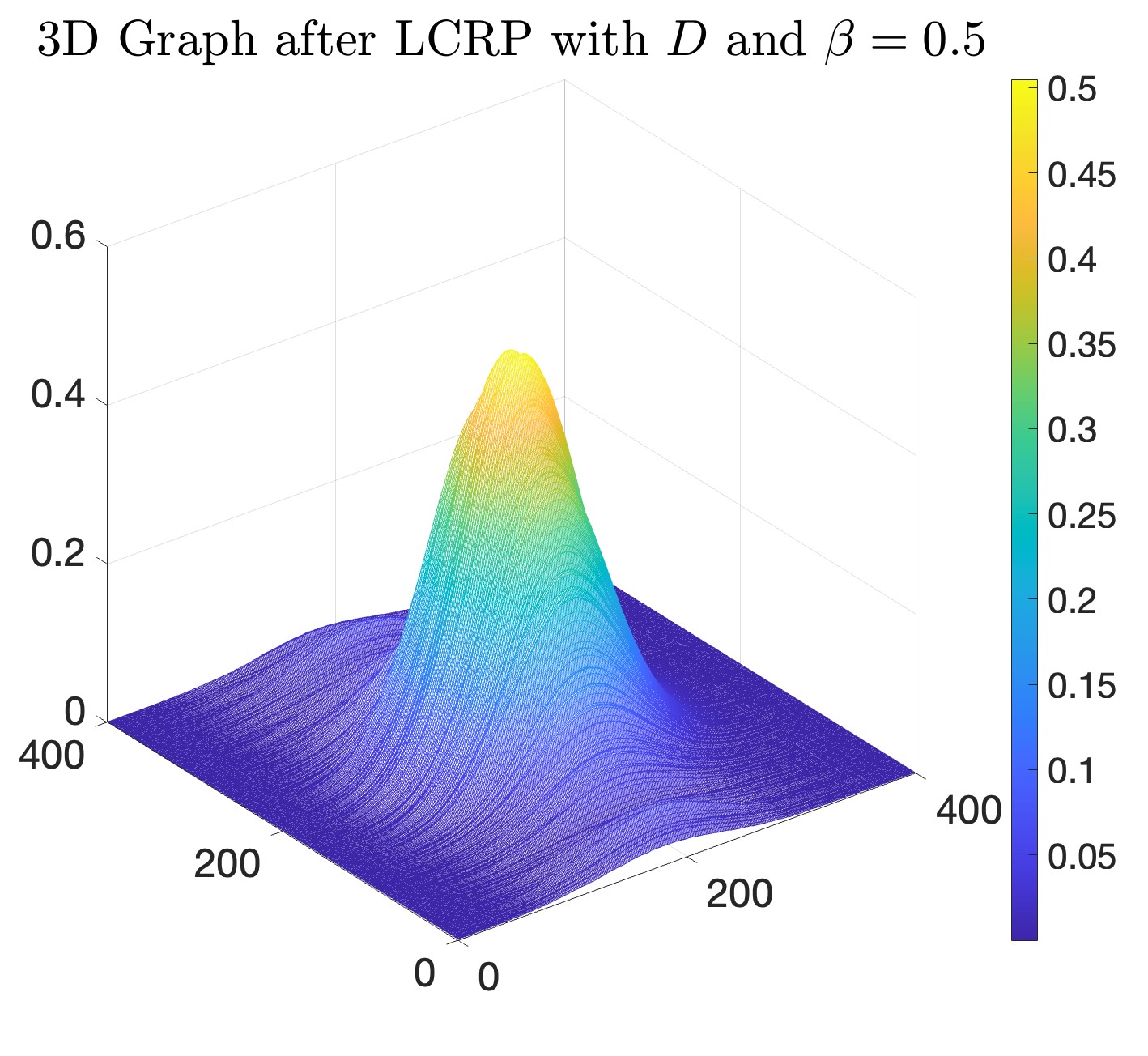}}\ \ \
\subfigure[]{\includegraphics[width=0.32\linewidth]
	{ 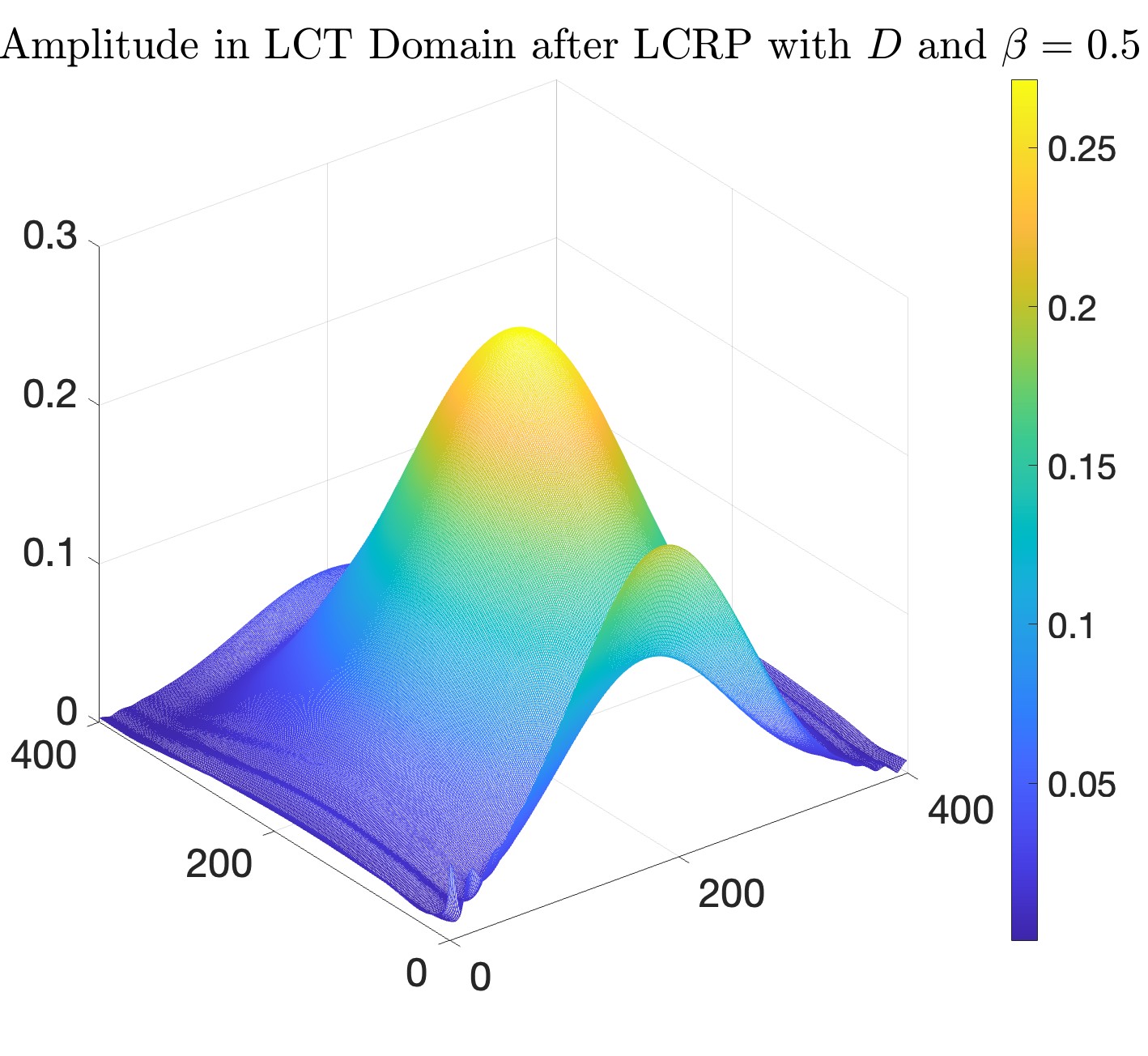}}\\
\vspace{0.5cm}
\subfigure[]{\includegraphics[width=0.32\linewidth]
	{ 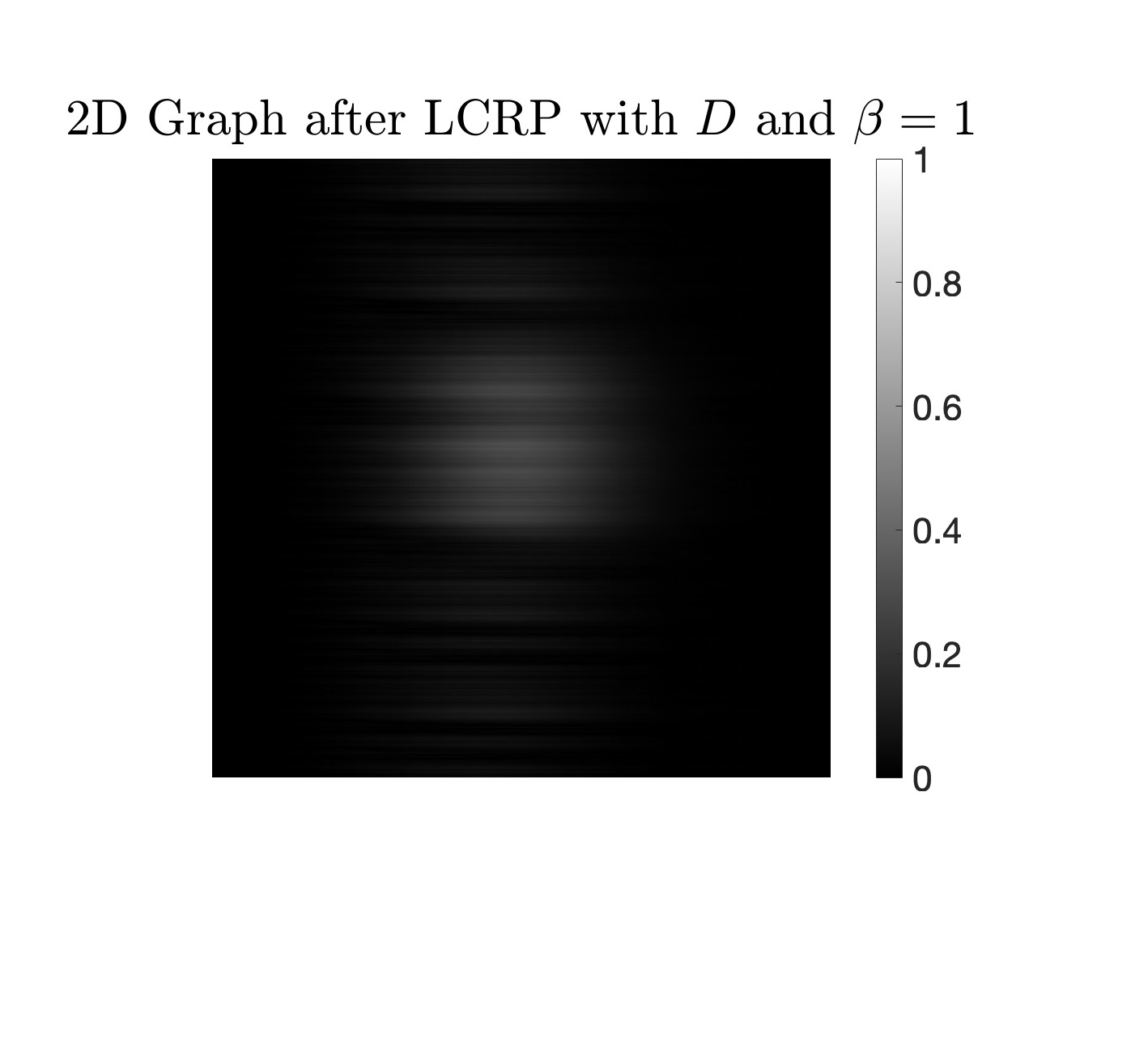}}\ \ \
\subfigure[]{\includegraphics[width=0.32\linewidth]
	{ 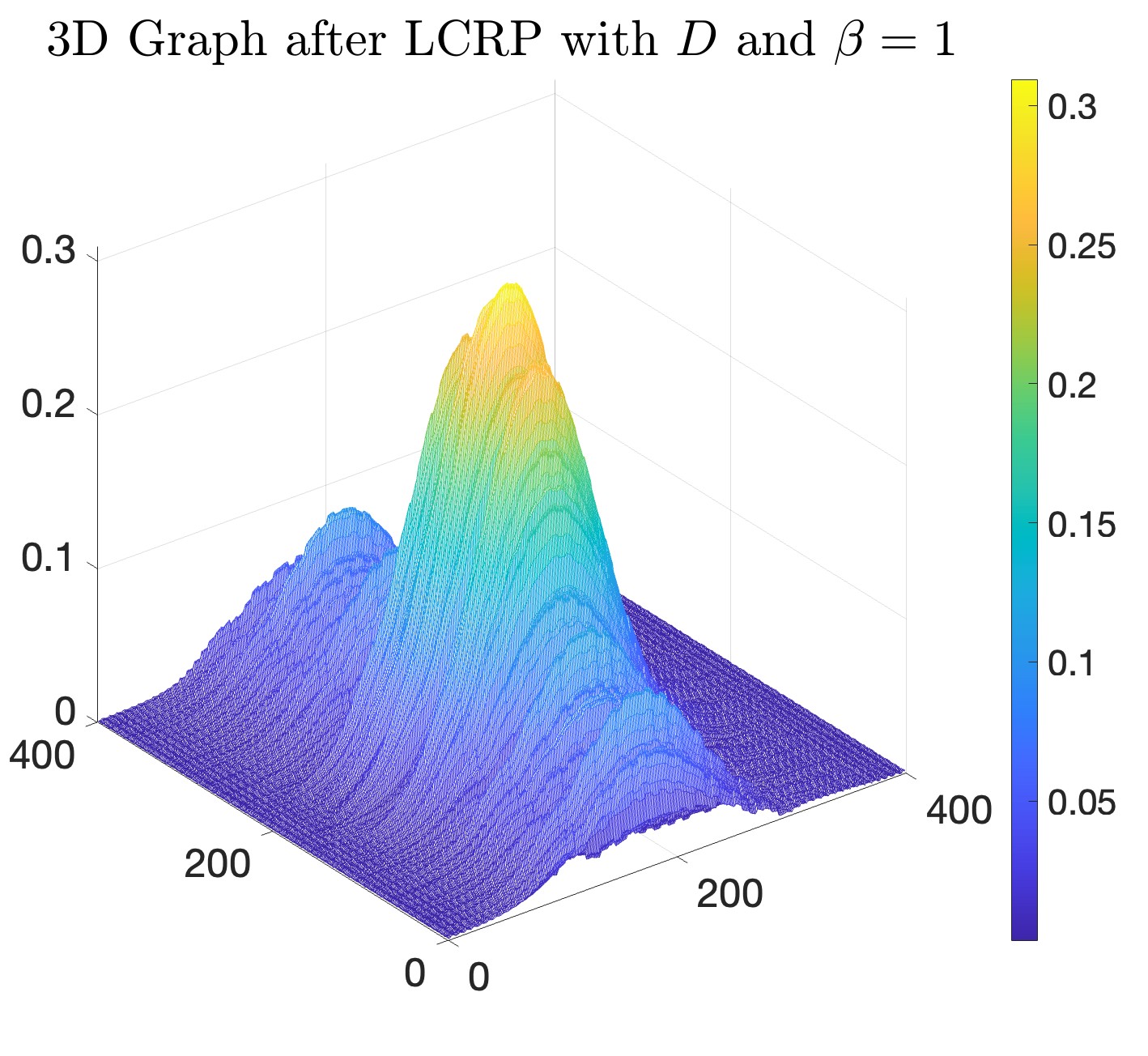}}\ \ \
\subfigure[]{\includegraphics[width=0.32\linewidth]
	{ 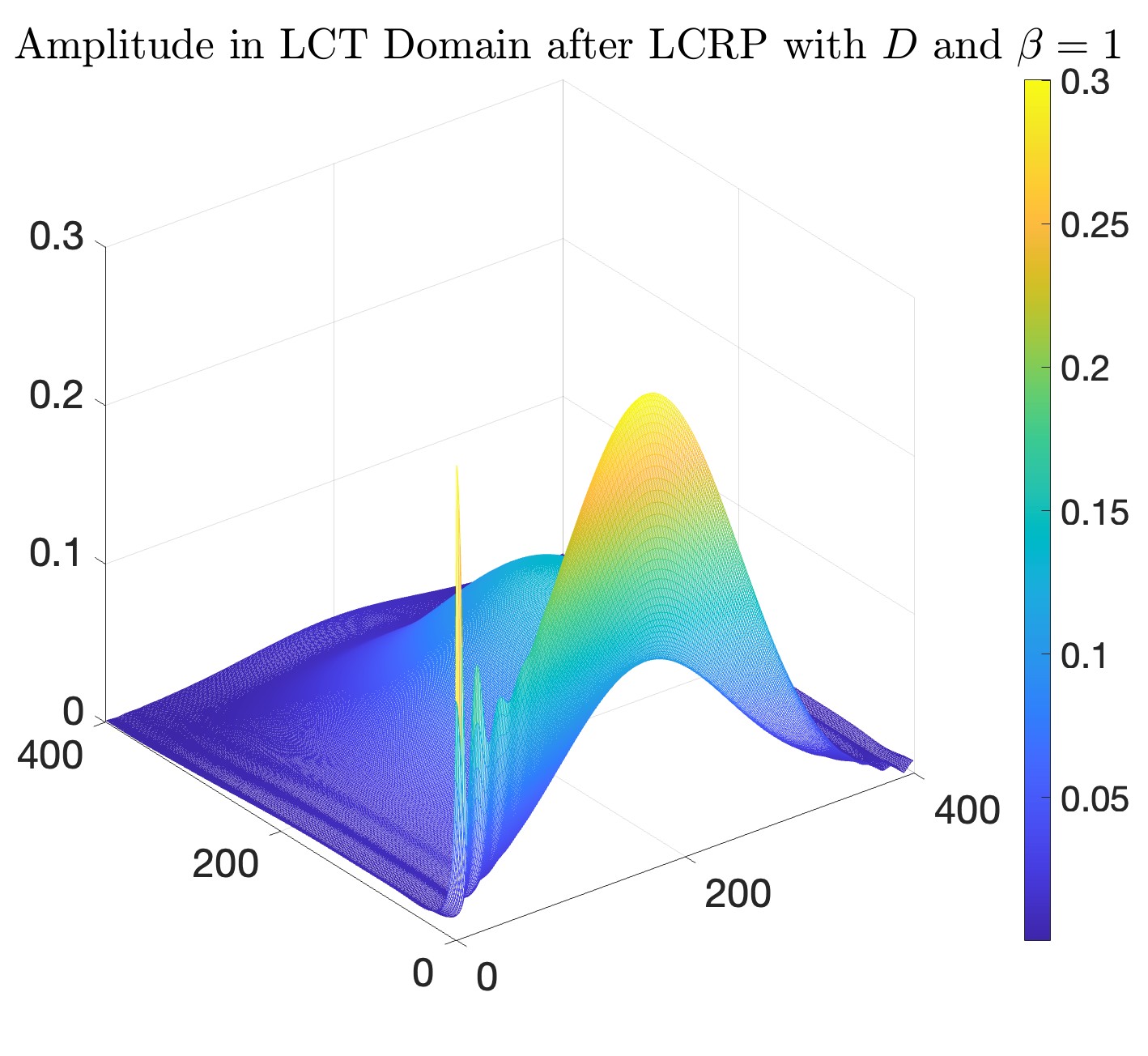}}\\
\vspace{0.5cm}
\subfigure[]{\includegraphics[width=0.32\linewidth]
	{ 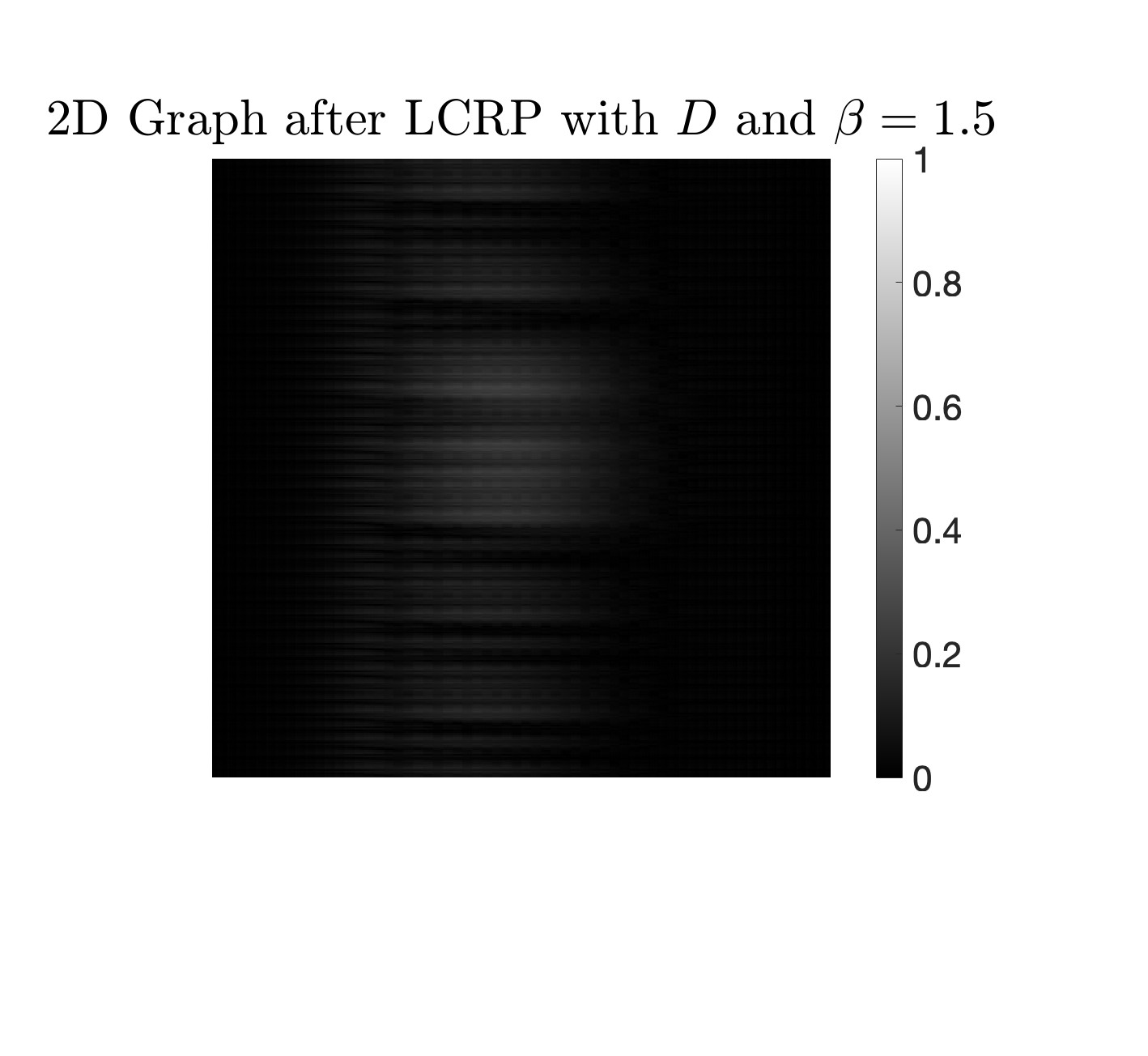}}\ \ \
\subfigure[]{\includegraphics[width=0.32\linewidth]
	{ 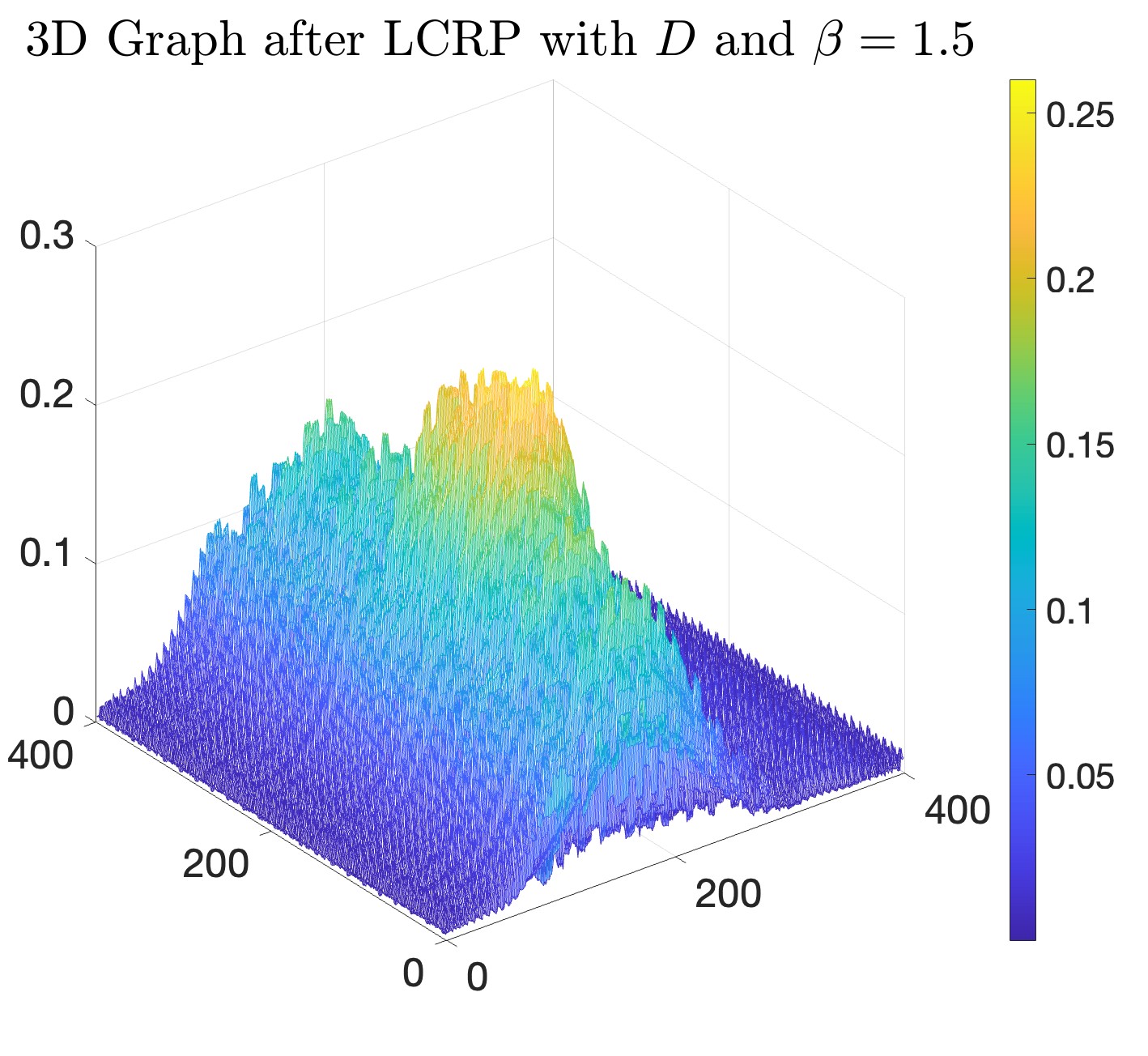}}\ \ \
\subfigure[]{\includegraphics[width=0.32\linewidth]
	{ 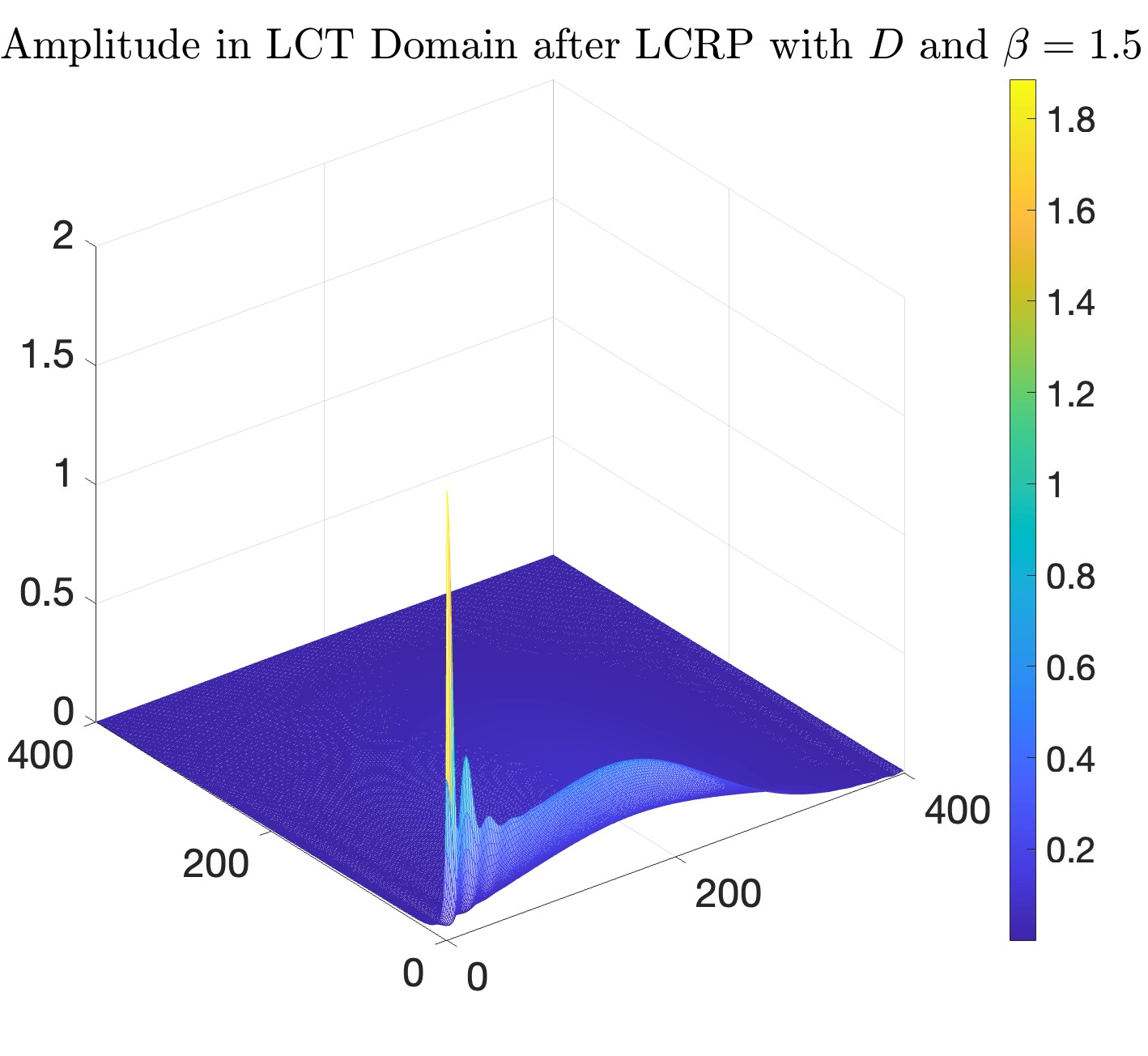}}
\caption{LCRP with $\boldsymbol D$ and, respectively,  
	$\beta=0.5$, $\beta=1$, and $\beta=1.5$.}
\label{FIG3.5}
\end{figure}

Figure \ref{FIG3.4} presents the numerical simulation of the Gaussian function processed by the LCRP  $I^{\boldsymbol{C}}_\beta$ with the parameter matrix $\boldsymbol{C}: = (C_1, C_2)$  and varying values of $\beta$, where  $C_1 = \begin{bmatrix} 20 & 399 \\ 1 & 20 \end{bmatrix}$ and $C_2 = \begin{bmatrix} 40 & 15.99 \\ 100 & 40 \end{bmatrix}$. Subfigures (a), (d), and (g) display the resulting 2D grayscale images, respectively,  for $\beta = 0.5$, $1$, and $1.5$. Subfigures (b), (e), and (h) show the corresponding 3D color representations of Subfigures (a), (d), and (g). The amplitude distributions in the corresponding LCT domains for these outputs are illustrated in Subfigures (c), (f), and (i). From Subfigures (c), (f), and (i), it follows that varying $\beta$ leads to notable changes in amplitude modulation within the LCT domain. Meanwhile, Subfigures (a), (d), and (g) reveal that, under a fixed matrix $\boldsymbol{C}$, different values of $\beta$ also alter the spatial characteristics of the image, making it appear either sparser or denser in the spatial domain. 

Figure \ref{FIG3.5} presents the numerical simulation of the Gaussian function processed by the LCRP  $I^{\boldsymbol{D}}_\beta$ with parameter matrix $\boldsymbol{D}:= (D_1, D_2)$ and varying values of $\beta$, where  
$D_1 = \begin{bmatrix} 20 & 39.9 \\ 10 & 20 \end{bmatrix}$ and  $D_2 = \begin{bmatrix} 3 & 400 \\ 0.02 & 3 \end{bmatrix}$. Subfigures (a), (d), and (g) display the resulting 2D grayscale images, respectively, for $\beta = 0.5$, $1$, and $1.5$. Subfigures (b), (e), and (h) show the corresponding 3D color representations of Subfigures (a), (d), and (g). The amplitude distributions in the corresponding LCT domains for these outputs are illustrated in Subfigures (c), (f), and (i). Analysis of Figure \ref{FIG3.5} leads to the same conclusion as drawn from Figure \ref{FIG3.4}, the parameter $\beta$ significantly influences both the spatial characteristics and the LCT domain amplitude distributions of the processed images.
 
Based on the numerical simulations presented above, we observe that both the parameter matrix $\boldsymbol{A}$ and the scaling factor $\beta$ in the LCRP $I^{\boldsymbol{A}}_\beta$  and the LCLO $\Delta_{\beta}^{\boldsymbol A}$ can significantly modulate the amplitude distribution in the corresponding LCT domains. Furthermore, we have identified the operational mechanism of $\beta$ in the spatial domain. Therefore, the appropriate choices of $\boldsymbol{A}$ and $\beta$ play a crucial role in effectively processing different types of images or signals.
 
\section{Applications to Multi-Image Encryption}\label{sec4}
Image processing has long been a central focus in information sciences. As a vital branch, image encryption serves as one of the essential processes for protecting image content and ensuring secure transmission (see, for example, \cite{bshyh2007,gcn21,gcn21-2,hcg24,rj1995,ss04}). However, most existing cryptosystems in this field are linear in nature, rendering them susceptible to various attacks, including chosen-plaintext attacks. Moreover, a substantial portion of research has been confined to single-image encryption, which limits practical applicability and efficiency. To address evolving encryption requirements, techniques for double-image and multi-image encryption have attracted considerable attention. For instance, Tao et al. \cite{tlw10} proposed embedding two images into the amplitude and the phase of a complex function. Z. Liu and S. Liu \cite{ll2007} introduced an iterative fractional Fourier transform-based method to merge two images into a single ciphertext. Sui et al. \cite{sds2014} developed a double-image encryption approach using the discrete fractional random transform combined with chaotic processes to disrupt pixel correlations effectively. In the context of multi-image encryption, Kong et al. \cite{kzxxg2013} designed a scheme based on the cascaded fractional Fourier transform. Additionally, encryption of color images, as a specific yet important case, has also advanced significantly \cite{cdly2015, lglwwp2015}. However, as highlighted by Peng et al. \cite{pzw2006, pzyb2006}, traditional linear cryptosystems remain vulnerable to specific types of attacks. To mitigate these limitations, recent efforts have shifted toward breaking linearity by introducing asymmetric mechanisms. For example, Qin and Peng \cite{qp2010} proposed an asymmetric cryptosystem using phase-truncated Fourier transforms with distinct encryption and decryption keys. This idea was later extended to double-image encryption by Wang and Zhao \cite{wz2012} and Sui et al. \cite{sdlh2014}. Subsequently, Li et al. \cite{lyzlt2015} devised an asymmetric multiple-image encryption framework based on the cascaded fractional Fourier transform, yielding notable improvements in both security and efficiency.

Although significant advances have been made in multi-image encryption, particularly in asymmetric encryption mechanisms, these methods still face challenges in meeting higher security demands, resisting stronger attack models, and enhancing the efficiency of multi-image fusion and encryption. 
To address these challenges, based on Theorems \ref{thm-frp}, \ref{p-sharp}, \ref{rem-imp-1},   and \ref{thm-polygon}, Proposition \ref{pro-lclo}, and Remarks \ref{zyr-1}, \ref{zyr-2}, and \ref{the-spaces}, combining with extending the idea of asymmetric encryption to multi-image scenarios,  we propose a new method called the  {\bf Multi-Image Encryption Asymmetric Cascaded {\rm \bf LCRP}} (for short, {\bf MIE-AC-LCRP}) \textbf{Method}. To be precise, Subsection \ref{4.1} elaborates on its encryption and decryption algorithms, Subsection \ref{4.2} verifies its feasibility and key sensitivity via numerical experiments, and finally Subsection \ref{4.3} conducts a comprehensive security evaluation on it from four aspects: sensitivity analysis, statistical analysis, noise attack analysis, and occlusion attack analysis.
 
\subsection{Encryption and Decryption Algorithms}\label{4.1}

This subsection details the core components of the aforedescribed MIE-AC-LCRP method.  The encryption process converts multiple plaintext images into a single ciphertext through a series of the phase manipulation and the LCRP operations, while the decryption process reverses these steps to reconstruct the original images. The concrete  algorithms are described as follows.

\textbf{Encryption Algorithm.} As illustrated in Figure \ref{fig.ep}, the encryption procedure consists of the following 5 stages.

\begin{figure}[H]
\centering
\includegraphics[width=1\linewidth]{ 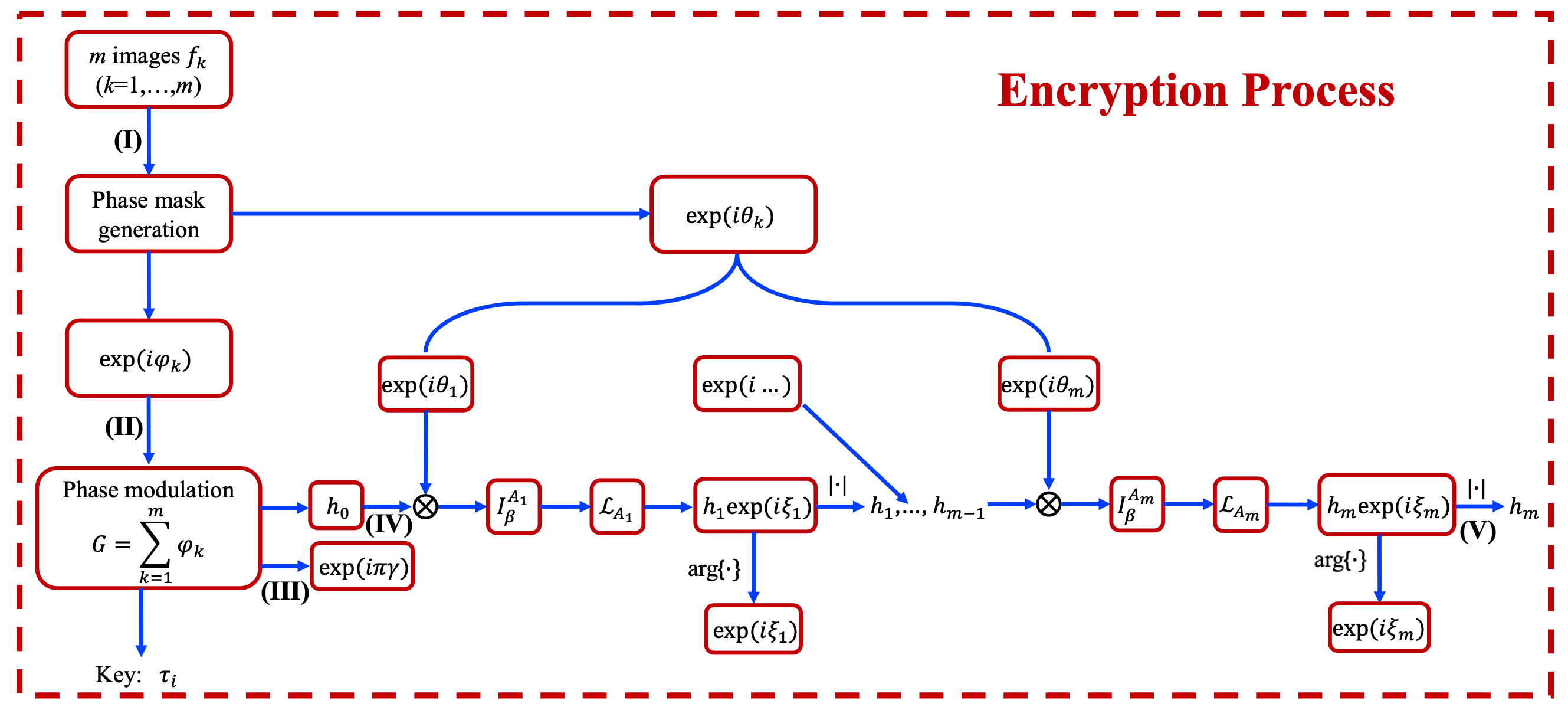} 
\vspace{-10pt}
\caption{Encryption Process of MIE-AC-LCRP Method.}
\label{fig.ep}
\end{figure}
 
\begin{enumerate}[(I)]
\item \textbf{Phase Mask Generation via Image Separation.}
We decompose the normalized input image $\{f_k(x,y)\}_{k=1}^n$ (representing the  pixel value distribution over spatial coordinates $(x,y)$) into two independent phase masks through 
a specific geometric transformation. The detailed procedure 
is outlined as the following four steps.
{\bf (a)} Generate the primary random phase mask $\exp(i\theta_k)$ with $\theta_k$ uniformly distributed in $[0, 2\pi)$.  
{\bf (b)} To achieve the desired decomposition, an image-dependent phase offset $\alpha_k$ must be computed for each pixel. This offset $\alpha_k$ is derived through an arccosine transformation applied to the normalized image intensity $f_k$ as follows
$$\alpha_k(x,y) := \pi - \arccos\left(\frac{2 - f_k^2(x,y)}{2}\right).$$
{\bf (c)} The secondary phase mask $\exp(i\varphi_k)$ is 
synthesized by combining the primary random phase 
$\theta_k$ with the computed separation angle $\alpha_k$ as follows
$$\exp\left[i\varphi_k(x,y)\right] := \exp\left\{i\left[\theta_k(x,y) +
 \alpha_k(x,y)\right]\right\}.$$
{\bf (d)} Validate the following decomposition 
$$f_k(x,y) = \left| \exp\left[i\varphi_k(x,y)\right]+ \exp\left[i\theta_k(x,y)\right] \right|.$$
Therefore, this procedure achieves image separation via random 
phase mask generation.

 \item \textbf{Nonlinear Phase Modulation with Security Enhancement.}
Via two steps, we perform a nonlinear combination of the phase masks obtained in (i), incorporating anti-cryptanalysis measures. {\bf (a)} Construct composite phase function
 $$G(x,y) := \sum_{k=1}^{m} \varphi_k(x,y),$$
where $\varphi_k$ is the phase mask derived from image 
separation generation.
{\bf (b)} Construct a modified equation system with randomized perturbation
\begin{equation*}
\left\{
\begin{aligned}
\tau_1(x,y) &:= \sum_{d=2}^{m} \varphi_d(x,y) + \mathcal{R}(x,y) \\
\tau_k(x,y) &:= \sum_{\substack{d=1, d\neq k}}^{m} \varphi_d(x,y) \quad (k=2,\dots,m)\\
G(x,y) &:= \sum_{k=1}^{m} \varphi_k(x,y),
\end{aligned}
\right.
\end{equation*}
where $\mathcal{R}$ is a random security term.
Here $\{\tau_k\}_{k=1}^m$ is kept as an essential key for decryption.  

\item \textbf{Unidirectional Phase Correction.}
We apply amplitude rectification to ensure detectability in optical domains via the following two steps.
{\bf (a)} Compute binary phase modulator
$$\gamma(x,y): =  \begin{cases} 
 1 & \text{if } \displaystyle\sum_{k=1}^m \varphi_k(x,y) < 0, \\
 0 & \text{otherwise.}
 \end{cases}$$
 {\bf (b)} Generate a real-valued interim mask
$$ h_0(x,y) := 
 \sum_{k=1}^{m} \varphi_k(x,y)\cdot \exp\left[i\pi\gamma(x,y)\right],$$
where $\exp(i\pi\gamma(x,y))$ is kept as an 
essential key for decryption.

 \item \textbf{Multi-stage Cascaded LCRP Encryption.}
Based on Theorems \ref{thm-frp} and \ref{rem-imp-1}, we encrypt the real-valued interim mask $h_0$ obtained in  Stage (iii) by employing iterative LCRPs with phase truncation in the following two steps.
{\bf (a)  Initialization.} 
Choose $h_0$ from the previous stage as the initial amplitude.
{\bf (b) Iterative Processing.}  
For each encryption stage $j \in\{1, \dots, m\}$, apply random phase modulation using the encryption key
$$\Psi_j(x,y) := h_{j-1}(x,y) \cdot \exp\left[i\theta_j(x,y)\right],$$
where $\theta_j$ is the phase mask derived from image 
separation generation in Stage (i). Then apply the symbol $m^{\beta_j}_{\boldsymbol{A}_j}$ of the LCRP and the LCT $\mathscr{L}_{\boldsymbol{A}_j}$ to obtain 
$$\Gamma_j(x,y): = \mathscr{L}_{\boldsymbol{A}_j}
\left[ \Psi_j(\cdot,\cdot) \right](x,y) \cdot m^{\beta_j}_{\boldsymbol{A}_j}(x,y),$$
where $m^{\beta_j}_{\boldsymbol{A}_j}:=(2\pi)^{-\beta_j}|(\widetilde{x},\widetilde{y})|^{-\beta_j}$ with $\widetilde{x}:=\frac{x}{b_1}$ and $\widetilde{y}:=\frac{y}{b_2}$ ($b_1$ and $b_2$ are respectively the values at position (1,2) of the first and the second matrices in ${\boldsymbol{A}_j}$ as in Definition \ref{nDLCT} with $n=2$).
Separate complex output respectively into amplitude and phase components
$$h_j(x,y) := \left| \Gamma_j(x,y) \right|$$
and
$$\xi_j(x,y): = \arg\left[\Gamma_j(x,y) \right],$$
where $\xi_j$ is kept as an  essential key for decryption.  

\item \textbf{Ciphertext Synthesis.}
The amplitude $h_m$ is the final \textit{ciphertext image}
$$C(x,y):= h_m(x,y).$$
\end{enumerate}

\textbf{Decryption Algorithm.} The decryption process, illustrated in Figure \ref{fig.ds}, constitutes the inverse of encryption and requires the correct ciphertext along with all associated keys, including the phase component $\{\xi_j\}_{j=1}^n$, binary phase modulator $\gamma$, LCRP parameter matrices $\{\boldsymbol{A}_j\}_{j=1}^m$ and order $\{\beta_j\}_{j=1}^m$, and auxiliary phase term $\{\tau_k\}_{k=1}^m$. The detailed decryption process consists of the following five stages.

\begin{enumerate}[(I)]
\item \textbf{Initialization.} We start the decryption process by setting the final ciphertext as the initial amplitude 
 $h_m(x,y): = C(x,y)$.
 
\begin{figure}[H]
\centering
\includegraphics[width=1\linewidth]{ 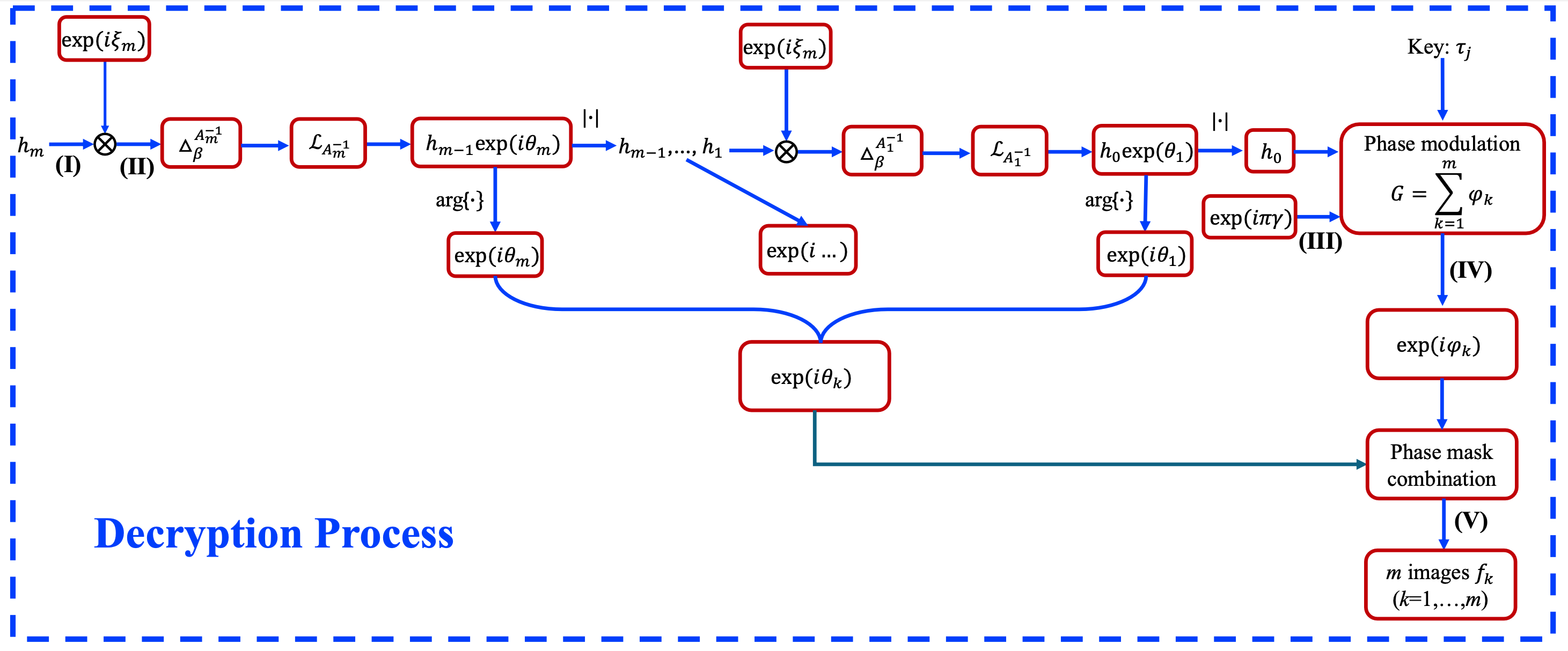}
\vspace{-10pt}
\caption{Decryption Process of MIE-AC-LCRP Method.}
\label{fig.ds}
\end{figure} 

\item \textbf{Iterative Processing.} 
We execute an iterative inverse process for each $j \in \{1, \ldots, m\}$ (in descending order from $m$ down to $1$) to retrieve the interim masks. The detailed procedure is outlined as in the following four steps.
{\bf (a) Reconstruct Complex Signal.}
We reconstruct the complex signal $\Gamma_j$ by using the amplitude $h_j$ and the
phase key $\xi_j$ as follows
$$ \Gamma_j(x,y) := h_j(x,y) \cdot \exp\left[i \xi_j(x,y)\right].$$
{\bf (b)} {\bf Apply the LCLO.} Based on Definition \ref{def-lclo} and Proposition \ref{pro-lclo}, we apply the inverse operator $\Delta_{\beta}^{A_j^{-1}}$ of the LCRP to reverse the amplitude modulation in the LCT domain.
$$V_j(x,y): = \Delta_{\beta}^{A_j^{-1}} \left[  \Gamma_j(x,y) \right],$$    
where $\boldsymbol{A}_j^{-1}$ is as in Definition \ref{def-ia} with $n=2$.            
{\bf (c)} {\bf Inverse the LCT.}
We apply the inverse LCT to obtain
$$\Psi_j(x,y) := \mathscr{L}_{A_j^{-1}} \left( V_j \right)(x,y).$$           
{\bf (d)  Extract Amplitude and Phase.}
We extract the amplitude $h_{j-1}$ and the phase $\theta_j$, respectively, by setting 
 $$h_{j-1}(x,y) := \left| \Psi_j(x,y) \right|\quad\mathrm{and}\quad\theta_j(x,y) := \arg\left[\Psi_j(x,y) \right].$$
   
\item \textbf{Phase Sum Recovery.}
We recover the combined phase sum $G$ by using the binary phase modulator  $\gamma$ as follows
$$G(x,y): = h_0(x,y) \cdot \exp\left(i \pi \gamma(x,y)\right).$$

\item \textbf{Phase Demodulation.}
By solving the following modified equation system, we recover the individual image phases $\varphi_k$ as follows
$$\begin{cases}
\varphi_k(x,y) := G(x,y) - \tau_k(x,y) &    (k\in\{2,\ldots,m\})\\
\displaystyle\varphi_1(x,y): = G(x,y) - \sum_{k=2}^m\varphi_k(x,y). 
\end{cases}$$

\item \textbf{Image Reconstruction.}
Using $\{\varphi_k\}_{k=1}^m$ and $\{\theta_k\}_{k=1}^m$, we reconstruct the original images
$$f_k(x,y) := \left| \exp\left[ i\varphi_k(x,y)\right] + \exp\left[i\theta_k(x,y)\right] \right|, \quad k\in\{1,\ldots,m\}.$$
\end{enumerate}
Through these inverse operations, the original images $\{f_k\}_{k=1}^m$ are successfully reconstructed, demonstrating the reversibility of the proposed MIE-AC-LCRP scheme under the correct key usage.

This method proposed in this subsection is also suitable for color image encryption. To be precise, noting that the color image can be decomposed into red, green, and blue channels, we separately process  each channel by  using the MIE-AC-LCRP method, and then recompose them into a color image. This decomposition approach was also given in \cite{flyy24}.

\subsection{Numerical Simulations}\label{4.2}
  
This subsection presents numerical simulations to validate the effectiveness of the proposed multi-image encryption and decryption scheme based on the MIE-AC-LCRP method. Three original images, namely Baboon, Peppers, and Barbara, are selected as test samples. We examine the influence of the key parameters, namely  the orders $\beta_1$, $\beta_2$, and $\beta_3$ of the LCRT and the parameter matrices $\boldsymbol A_1$, $\boldsymbol A_2$, and $\boldsymbol A_3$ of the LCRT, on the encryption performance.

Figure \ref{fig.421} illustrates the encryption and the decryption results obtained by using the correct keys, displaying the original, the encrypted, and the correctly decrypted images. The key parameters used are 
$\beta_1=1$, $\beta_2=1.5$, $\beta_3=0.7$,
$\boldsymbol A_1=(A_1^1,A_1^2)$
with $A_1^1=\begin{bmatrix}{6}&{7}\\{5}&{6}
\end{bmatrix}$ and  $A_1^2=\begin{bmatrix}{1}&{20}\\{0}&{1}
\end{bmatrix}$,
$\boldsymbol A_2=(A_2^1,A_2^2)$
with $A_2^1=\begin{bmatrix}{5}&{12}\\{2}&{5}\end{bmatrix}$ and  $A_2^2=\begin{bmatrix}{1}&{11}\\{9}&{100}
\end{bmatrix}$, and
$\boldsymbol A_3=(A_3^1,A_3^2)$ 
with $A_3^1=\begin{bmatrix}{7}&{11}\\{5}&{8}
\end{bmatrix}$ and  $A_3^2=\begin{bmatrix}{11}&{21}\\{1}&{2}\end{bmatrix}$. 

Next, we conduct three numerical simulations with some error keys. 

Simulation (i) consists of Figure \ref{fig.422}(I)--(III), which show the decryption outcomes when a single parameter error is introduced in $\boldsymbol A_1 = (A_1^1, A_1^2)$, $\boldsymbol A_2 = (A_2^1, A_2^2)$, and $\boldsymbol A_3 = (A_3^1, A_3^2)$, respectively. Specifically, three parameters from the correct key set in Figure \ref{fig.421}, $A_1^1$, $A_2^2$, and $A_3^1$, are replaced, respectively, by erroneous values $A_1^1=\begin{bmatrix}{6}&{6}\\{5}&{6}\end{bmatrix}$, $A_2^2=\begin{bmatrix}{1}&{11}\\{8}&{100}\end{bmatrix}$, and $A_3^1=\begin{bmatrix}{7}&{11}\\{5}&{9}\end{bmatrix}$. All other key components remain unchanged. The decrypted results, shown in Figure \ref{fig.422}(I)--(III), confirm that altering even a single parameter in any $\boldsymbol A_i$ results in a complete loss of the  original image information. 

\begin{figure}[H]
\centering
\includegraphics[width=1\linewidth]{ 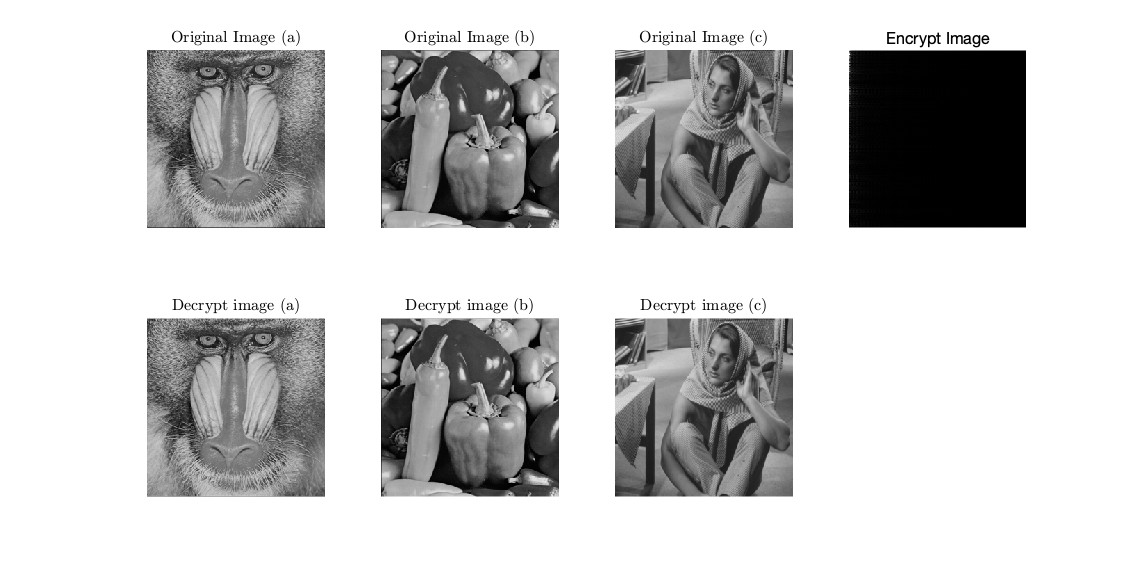}	
\vspace{-30pt}
\caption{Correct Encryption and Decryption Simulation.}
\label{fig.421}
\end{figure}

\begin{figure}[H]
\centering
\includegraphics[width=0.99\linewidth]{ 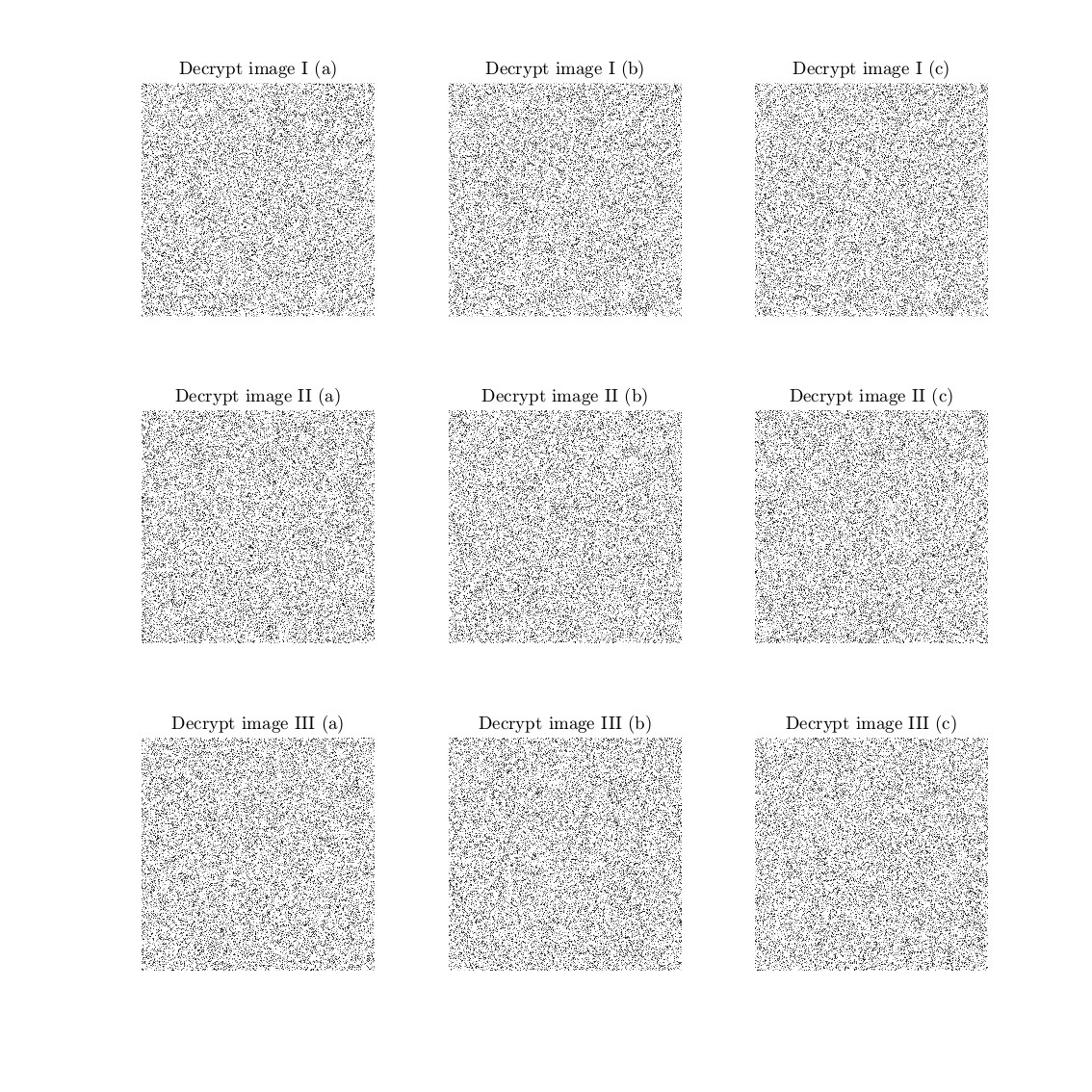} 	
\vspace{-50pt}
\caption{Error Encryption and Decryption Simulation (i).}
\label{fig.422}
\end{figure}

\begin{figure}[H]
\centering
\includegraphics[width=1\linewidth]{ 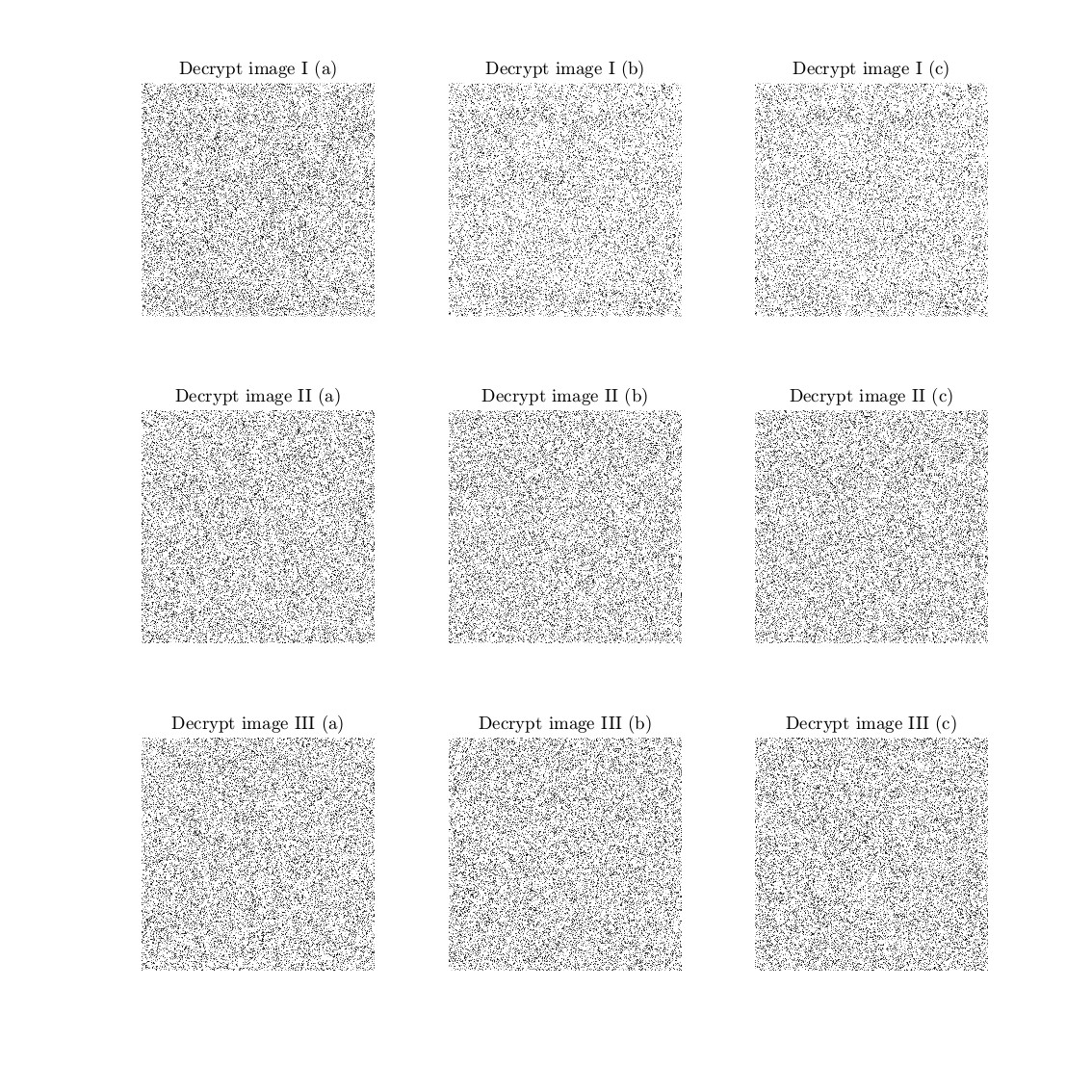}	
\vspace{-45pt}
\caption{Error Encryption and Decryption Simulation (ii).}
\label{fig.423}
\end{figure}

\vspace{-11pt}
\begin{table}[H]
	\caption{Decryption Images Corresponding to Different Keys.}
	\centering
	\begin{tabular}{ccc}
		\toprule[1.5pt]
		\textbf{\ (Correct/{\red Error}) key\quad} & \textbf{\quad(Correct/{\red Error}) Parameters\quad}    & \quad \textbf{Image}\\
		\midrule[1pt]
		Correct Key & 
		${\boldsymbol{A}_1},{\boldsymbol{A}_2},{\boldsymbol{A}_3},\ 
		\beta_1 ,\beta_2,\beta_3$ 
		&  Figure \ref{fig.421} \\
		\toprule[1.5pt]
		{\red Error} Key &
		${\red{\boldsymbol{A}_1}},{\boldsymbol{A}_2},{\boldsymbol{A}_3},\ 
		\beta_1 ,\beta_2,\beta_3$ 
		& Figure \ref{fig.422}(I)\\
		{\red Error} Key & 
		${\boldsymbol{A}_1},{\red{\boldsymbol{A}_2}},{\boldsymbol{A}_3},\ 
		\beta_1 ,\beta_2,\beta_3$ 
		&  Figure \ref{fig.422}(II) \\
		{\red Error} Key & 
		${\boldsymbol{A}_1},{\boldsymbol{A}_2},{\red{\boldsymbol{A}_3}},\ 
		\beta_1 ,\beta_2,\beta_3$  
		&  Figure \ref{fig.422}(III)\\
		\toprule[1.5pt]
		{\red Error} Key & 
		${\boldsymbol{A}_1},{\boldsymbol{A}_2},{\boldsymbol{A}_3},\ 
		{\red \beta_1},\beta_2,\beta_3$ 
		&  Figure \ref{fig.423}(I)\\
		{\red Error} Key & 
		${\boldsymbol{A}_1},{\boldsymbol{A}_2},{\boldsymbol{A}_3},\ 
		\beta_1 ,{\red \beta_2},\beta_3$     
		& Figure \ref{fig.423}(II)\\
		{\red Error} Key & 
		${\boldsymbol{A}_1},{\boldsymbol{A}_2},{\boldsymbol{A}_3},\ 
		\beta_1 ,\beta_2,{\red \beta_3}$  
		&  Figure \ref{fig.423}(III)\\
		\toprule[1.5pt]
		{\red Error} Key & 
		${\red \boldsymbol{A}_1},{\boldsymbol{A}_2},{\boldsymbol{A}_3},\ 
		{\red \beta_1},\beta_2,\beta_3$ 
		& Figure \ref{fig.424}(I)\\
		{\red Error} Key & 
		${\boldsymbol{A}_1},{\red \boldsymbol{A}_2},{\boldsymbol{A}_3},\ 
		\beta_1,{\red \beta_2},\beta_3$    
		&  Figure \ref{fig.424}(II)\\
		{\red Error} Key & 
		${\boldsymbol{A}_1},{\boldsymbol{A}_2},{\red \boldsymbol{A}_3},\ 
		\beta_1 ,\beta_2,{\red \beta_3}$   
		&  Figure \ref{fig.424}(III)\\
		\bottomrule[1.5pt]
	\end{tabular}
	\label{tab:1}
\end{table}

Simulation (ii) consists of  Figure \ref{fig.423}(I)--(III), which reflect the decryption effects when incorrect values of $\beta_1$, $\beta_2$, and $\beta_3$ are used. In this case, the correct values of $\beta_1$, $\beta_2$, and $\beta_3$ from Figure \ref{fig.421} are replaced by
$\beta_1 = 1.1$, $\beta_2 =1.45$, and $\beta_3 =0.72$. All other key elements are kept correct. As shown in Figures \ref{fig.423}(I)--(III), even a slight deviation in any $\beta_i$ will surely lead to a total loss of recoverable image content.

\vspace{-11pt}
\begin{figure}[H]
\centering
\includegraphics[width=0.99\linewidth]{ 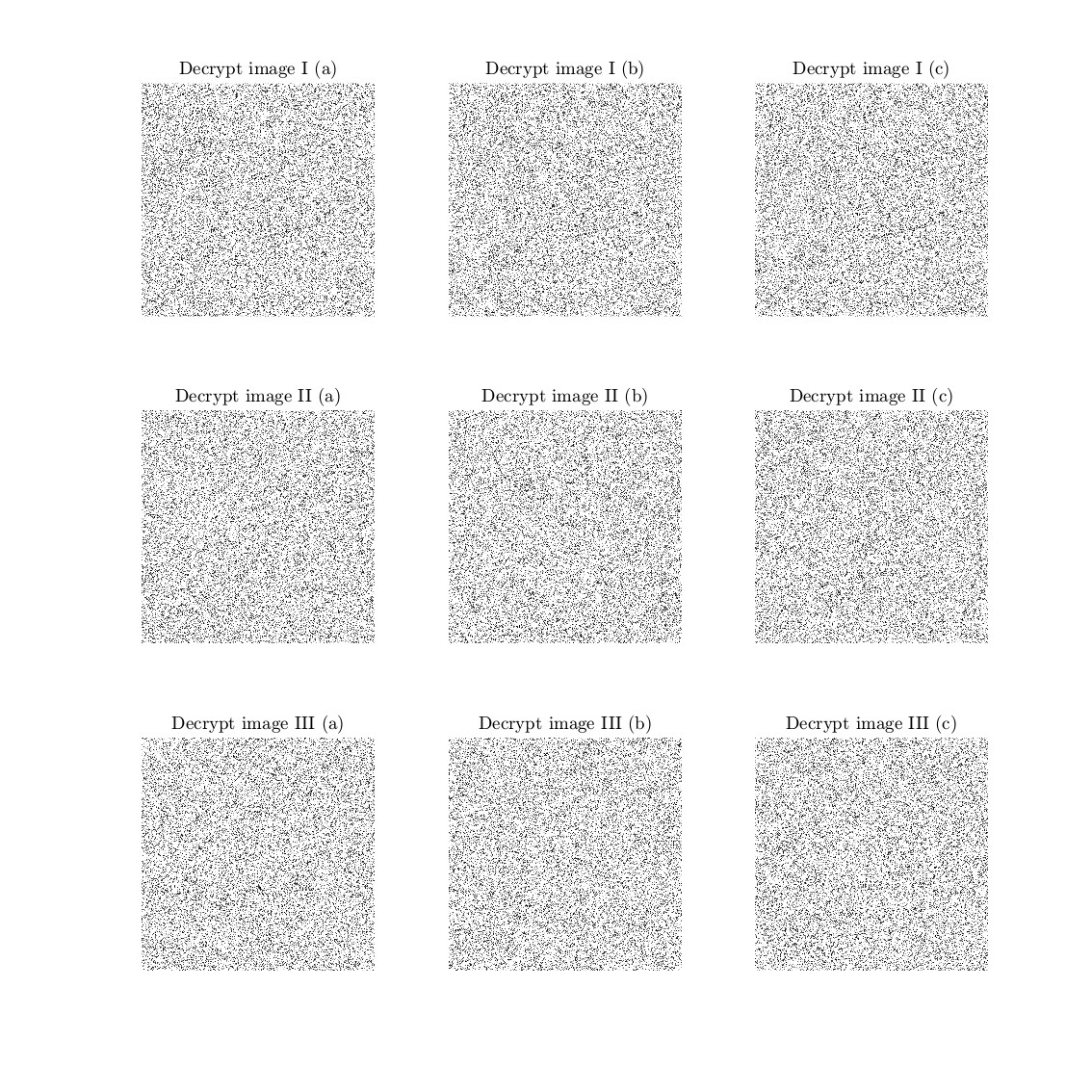} 	
\vspace{-50pt}
\caption{Error Encryption and Decryption Simulation (iii).}
\label{fig.424}
\end{figure}

Simulation (iii) consists of  Figure \ref{fig.424}(I)--(III), which present the decryption results when errors occur simultaneously in both parameters $\boldsymbol A_i$ and $\beta_i$. The erroneous $\boldsymbol A_i$ values are the same as those in Figure \ref{fig.422}, and the erroneous $\beta_i$ values match those in Figure \ref{fig.423}. Experimental results confirm that the original image information is entirely irrecoverable under such compound errors.

Table \ref{tab:1} summarizes the decryption outcomes under various key configurations, highlighting the relationship between the correctness of the encryption keys, namely the LCRP matrices $\boldsymbol A_1$, $\boldsymbol A_2$,  and $\boldsymbol A_3$ and the LCRP orders $\beta_1$, $\beta_2$, and $\beta_3$, and the resulting decrypted images. When all keys are correct, the original images are perfectly reconstructed. However, if any key is erroneous, whether in $\boldsymbol A_i$, $\beta_i$, or both, the decrypted images become entirely unrecognizable, underscoring the high key sensitivity and security of the proposed encryption system.

In summary, the numerical simulations affirm the feasibility and the robustness of the proposed encryption scheme. The correct decryption results (Figure \ref{fig.421}) validate the reversibility of the MIE-AC-LCRP method encryption process, while the sensitivity analyses (Figures \ref{fig.422}--\ref{fig.424} and Table \ref{tab:1}) emphasize its strong reliance on both the LCRT parameter matrices $\{\boldsymbol A_i\}_{i=1}^3$ and the LCRT orders $\{\beta_i\}_{i=1}^3$. Even minor deviations in these parameters will result in complete decryption failure, highlighting the security and the key sensitivity of the proposed MIE-AC-LCRP method cryptosystem.
 
\subsection{Security Analysis}\label{4.3}

Having examining the feasibility of the MIE-AC-LCRP method cryptosystem and its sensitivity to core key parameters (the LCRP matrices $\{\boldsymbol{A}_i\}_{i=1}^{3}$ and the  LCRP orders $\{\beta_i\}_{i=1}^{3}$) in the preceding section, we now proceed to conduct a systematic assessment on its robustness and its security from the following four aspects: a quantitative key sensitivity analysis (Subsection \ref{3.3.1}), an examination of its resistance to statistical attacks (Subsection \ref{3.3.2}), and tests of its performances, respectively, under noise attacks (Subsection \ref{3.3.3}) and occlusion attacks (Subsection \ref{3.3.4}).

\subsubsection{Sensitivity Analysis}\label{3.3.1}
Sensitivity analysis is crucial for evaluating the robustness 
of the encryption system against key estimation attacks. We 
quantitatively assess the sensitivity of two core keys: the 
LCRT matrix parameters $\{\boldsymbol{A}_i\}_{i=1}^3$ and the LCRT 
orders $\{\beta_i\}_{i=1}^3$. The Mean Squared Error
(for short, MSE) $\text{MSE}(I_{\text{orig}},I_{\text{dec}})$  between decrypted and original images 
serves as our primary metric
\begin{equation*}
\text{MSE}(I_{\text{orig}},I_{\text{dec}}) := \frac{1}{MN} \sum_{i=1}^{M} 
\sum_{j=1}^{N} [I_{\text{orig}}(i,j) - I_{\text{dec}}(i,j)]^2,
\end{equation*}
where $I_{\text{orig}}$ and $I_{\text{dec}}$ denote the original and the decrypted images (each is of dimensions $M \times N$) and where $I_{\text{orig}}(i,j)$ and $I_{\text{dec}}(i,j)$ represent the pixel intensity values at position $(i,j)$ in the respective images.

Next, without loss of generality, we may conduct a sensitivity analysis only on the parameter $A_i^2[1,2]$, which denotes the row 1 and the column 2 of the second matrix in the LCRP parameter matrix $\boldsymbol{A}_i:=(A_i^1, A_i^2)$, by using three $512 \times 512$ test images: Baboon, Peppers, and Barbara. For each image $I_i$, we vary its specific parameter $A_i^2[1,2]$ around the true value $A_{i,\text{cor}}^2[1,2]$ in steps of $\Delta=2$, evaluating 15 points on each side (31 points total). To ensure the unimodular condition $\det(A_i^2) = 1$ holds, we dynamically update the element $A_i^2[2,1]$ by setting
\begin{equation*} 
A_i^2[2,1] := \frac{A_i^2[1,1] \cdot A_i^2[2,2] - 1}{A_i^2[1,2]}.
\end{equation*}
For each perturbed parameter set of $A_i$, we decrypt the corresponding ciphertext image of the $i$-th image and compute the MSE between the decrypted image and the original image.

\begin{figure}[H]
	\centering
	\includegraphics[width=0.99\linewidth]{ 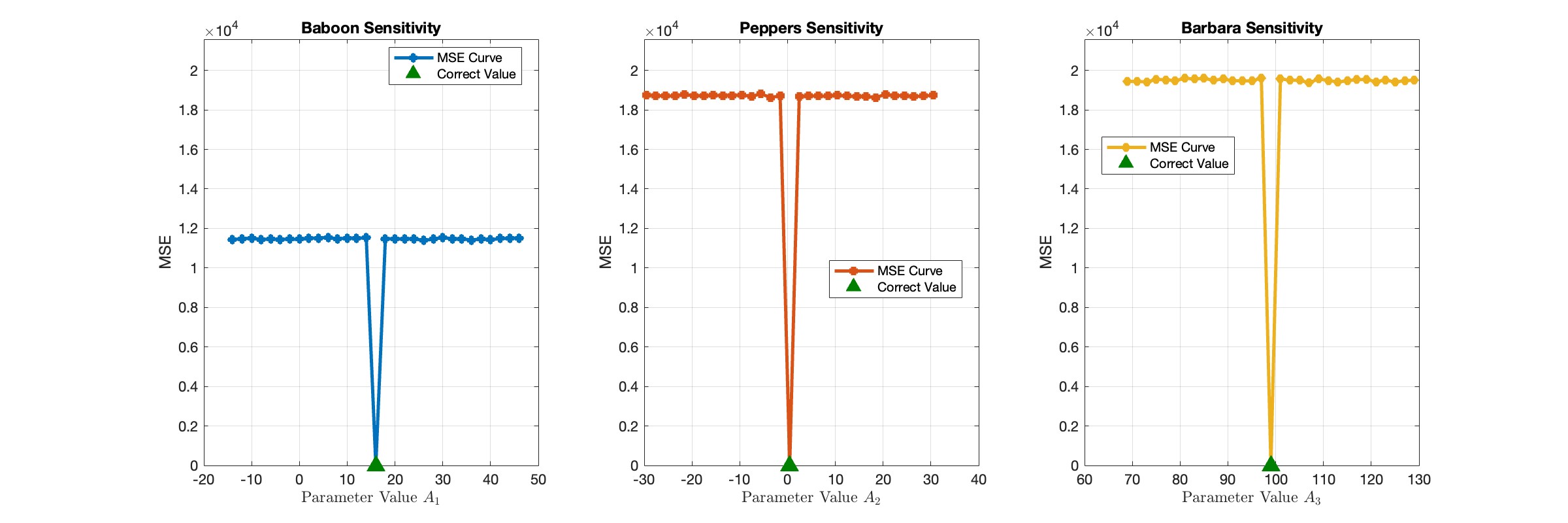} 	
	\caption{MSEs Between Decrypted Image of Error Key $A_i^2[1,2]$  and the Original Image.}
	\label{fig. 11}
\end{figure}

Figure \ref{fig. 11} demonstrates the relationship 
between MSE variation and the corresponding parameter 
$A_i^2[1,2]$ for each image. The correct parameter 
values and the minimal MSE values obtained with 
incorrect parameters are, respectively,
Baboon: $A^2_{1,\text{cor}}[1,2]  = 15.99$,
MSE$_{\text{min}} =1.114836 \times 10^{4}$;
Peppers: $A^2_{2,\text{cor}}[1,2]   =0.4495$, 
MSE$_{\text{min}} = 1.86587 \times 10^{4}$;
Barbara: $A^2_{3,\text{cor}}[1,2]  = 99$, 
MSE$_{\text{min}} = 1.94533 \times 10^{4}$.
Notably, the MSE exhibits exponential growth 
with increasing parameter deviation. For the Baboon 
image, a minor perturbation in $A_1^2[1,2]$ elevates
the MSE by four orders of magnitude ($10^4$). This 
hypersensitivity confirms that slight parameter deviations 
completely disrupt the decryption process. Crucially, we observe 
relatively plateaued MSE curves under incorrect parameters,
indicating attackers cannot deduce valid keys through MSE 
analysis, thereby enhancing resistance against statistical 
cryptanalysis.

Now, we evaluate the sensitivity of the parameters $\{\beta_i\}_{i=1}^3$ by using the same three test images (Baboon, Peppers, and Barbara). For each image $I_i$, we perturb its corresponding parameter $\beta_i$ around the correct value $\beta_{i,\text{cor}}$ with a step size of $\Delta = 0.025$, sampling 15 points in each direction (31 points in total), while all other parameters remain unchanged.

Figure \ref{fig. 12} demonstrates the relationship between the variation in MSE and the corresponding parameter $\beta_i$ for each test image. The correct parameter values and the minimal MSE values achievable with incorrect parameters are, respectively,
Baboon: $\beta_{1,\text{cor}} = 1$, $\text{MSE}_{\text{min}} = 1.02834 \times 10^{4}$;
Peppers: $\beta_{2,\text{cor}} = 1.8$, $\text{MSE}_{\text{min}} = 1.76192 \times 10^{4}$;
Barbara: $\beta_{3,\text{cor}}= 0.3$, $\text{MSE}_{\text{min}} = 1.93988 \times 10^{4}$.
We observe that the impact of $\beta_i$ on the MSE is qualitatively similar to that of $A_i^2[1,2]$ on the MSE.

 \begin{figure}[H]
\centering
\includegraphics[width=0.99\linewidth]{ 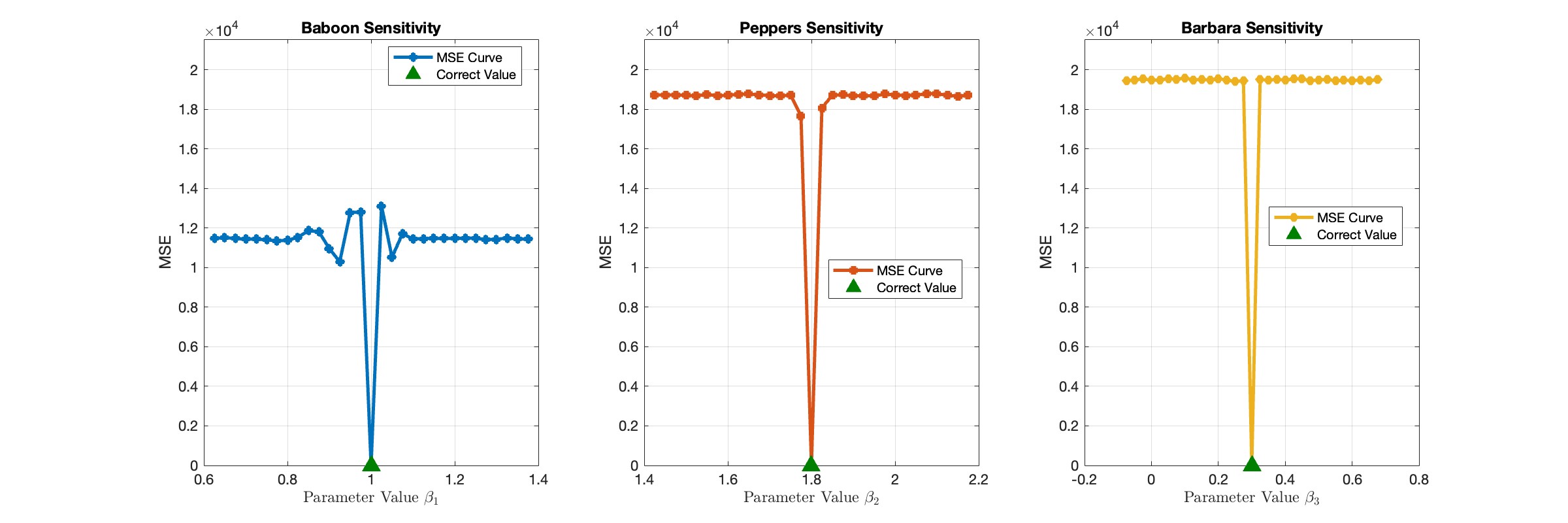} 	
\caption{MSEs Between Decrypted Image of 
	Error Key $\beta_i$  and  Original Image.}
\label{fig. 12}
\end{figure}

This hypersensitivity is not merely an empirical observation but a direct analytical consequence of Theorems \ref{rem-imp-1} and \ref{thm-polygon}.
Specifically, the cryptosystem's nonlinear defense mechanism is rooted in the two completely different properties of the LCRP: the global divergence risk for non-decaying backgrounds proven in Theorem  \ref{rem-imp-1}, and the local analytic instability at image edges established in Theorem  \ref{thm-polygon}.
As further analyzed in Remarks \ref{zyr-1}, \ref{zyr-2}, and \ref{the-spaces} and in Tables \ref{tab:convergence-comparison} and \ref{tab:sensitivity-mechanism}, the structural mismatch and the phase misalignment during near-critical parameter recovery provide a rigorous theoretical guarantee for the system's robustness against key-estimation attacks.

Although the MIE-AC-LCRP method proposed in this article is designed for multi-images, its security advantages remain equally prominent in single-image encryption scenarios. Figure~\ref{fig. bjmse} presents a key sensitivity analysis of the proposed scheme applied to single-image encryption. To ensure a fair comparison, we adopt the same test image (Cameraman) as used in \cite[Figure~12 ]{flyy24}, and compare our method with the encryption approach based on the fractional Riesz potential. The experimental results demonstrate that, when the key parameters, including the LCRP order $\beta$ and the LCRT matrix $\boldsymbol{A}$, are subject to minor perturbations, the MSE between the decrypted and the original images for our method is significantly higher than that reported in \cite{flyy24} by approximately a factor of three. This result demonstrates that, even in single-image scenarios, the proposed MIE-AC-LCRP method exhibits higher key sensitivity and stronger security performance. Slight key perturbations lead to severe distortion in the decrypted image, effectively countering key estimation attacks and further affirming the superiority of the present approach.

\begin{figure}[H]
\centering
\includegraphics[width=1\linewidth]{ 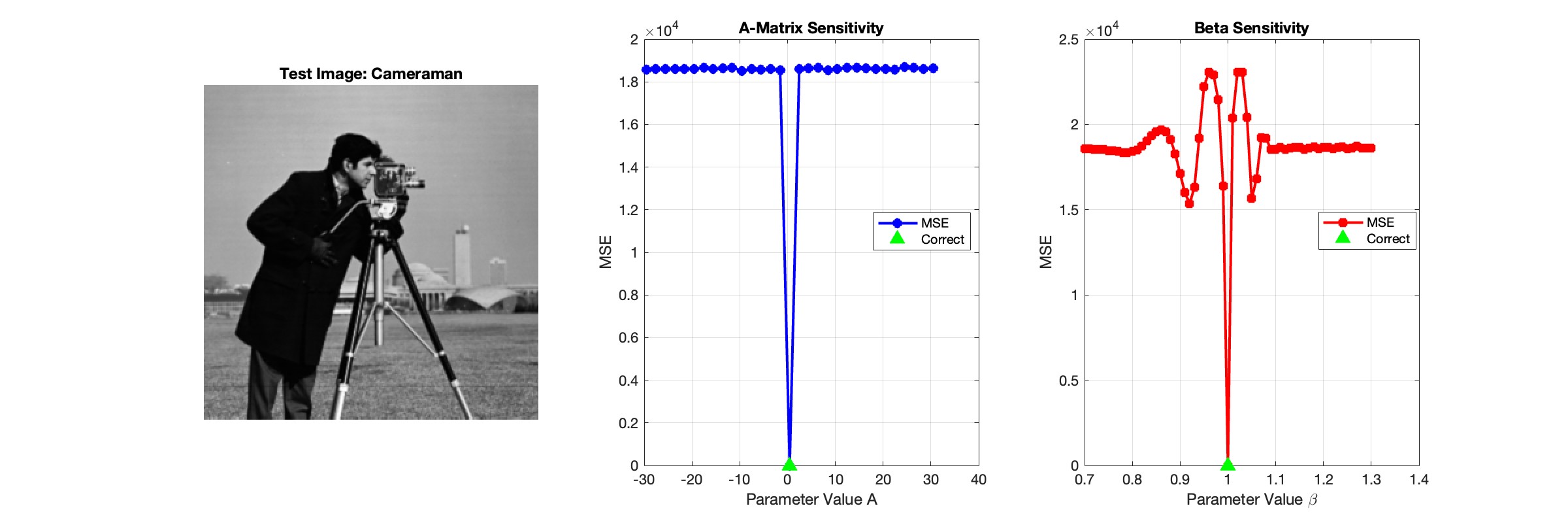} 	
\caption{MSEs of Erroneously Decrypted Images in Single-Image Encryption.}
\label{fig. bjmse}
\end{figure}

The above experiments further verify that both the parameters $\boldsymbol{A}$ and $\beta$ in the  MIE-AC-LCRP method exhibit nonlinear hypersensitivity, with minimum mean square error values reaching orders of magnitude of $10^4$. Moreover, the MSE curves in regions of incorrect parameters display a plateau-like behavior. These characteristics together ensure the algorithm's robustness against quantitative key perturbation estimation attacks. The underlying sensitivity mechanism stems from the structure of the LCRP symbol system and the matrix parameters, where minor deviations cause catastrophic mismatches during decryption, constituting a fundamental defense mechanism.

\subsubsection{Statistical Analysis}\label{3.3.2}
The resilience of an image encryption scheme against 
statistical attacks is critically evaluated by two key 
metrics: internal correlation and global correlation. 
The former quantifies the adjacency dependency on
pixels, while the latter gauges the similarity between 
the statistical profiles of the cipher and the plain images.

The internal correlation is assessed by computing 
the Pearson correlation coefficient for adjacent 
pixels along the horizontal, the vertical, and the diagonal 
directions. The \emph{Pearson correlation coefficient} 
$r$ is formally defined by setting
\begin{equation*}
r:= \frac{\sum\limits_{i=1}^{n}(x_i - \widetilde{x})
(y_i - \widetilde{y})}{\sqrt{\sum\limits_{i=1}^{n}(x_i - 
\widetilde{x})^2}  \sqrt{\sum\limits_{i=1}^{n}
(y_i - \widetilde{y})^2}},
\end{equation*}
where $\{(x_i,y_i)\}\}_{i=1}^{n}$ denotes a sequence of adjacent 
pixel pairs and where $\widetilde{x}$ and $\widetilde{y}$ are their 
corresponding means of  $\{x_i\}_{i=1}^{n}$ and  $\{y_i\}_{i=1}^{n}$.

As shown in Table \ref{tab.421}, all original 
images demonstrate strong spatial correlations, 
with coefficients nearing unity (for instance, 0.97361 in 
the vertical direction). In stark contrast, the ciphertext 
images exhibit near-zero correlations ($|r| < 0.015$), 
indicating that the encryption effectively scrambles the
local pixel correlations.

The results in Table \ref{tab.422} demonstrate that 
the decryption process is highly sensitive to both 
the matrix parameters $A_j$ and the orders $\beta_j$. 
When either is incorrect, the correlation coefficients 
of the decrypted outputs are in a near-zero 
range $[-0.012166, 0.010027]$ for erroneous $A_j$ and
$[-0.0061311, 0.0054178]$ for erroneous $\beta_j$. 
This indicates that any deviation from the correct key 
results in a decrypted image that preserves the 
statistical profile of random noise.

To quantify the global correlation, we need to compute the \textit{Pearson correlation coefficient} $\rho$ between the original image $I_{\text{orig}}$ and its encrypted version $I_{\text{enc}}$, which is defined by setting
\begin{equation*}
	\rho := \frac{\sum\limits_{i=1}^{M}\sum\limits_{j=1}^{N}[I_{\text{orig}}(i,j) - \mu_{\text{orig}}][I_{\text{enc}}(i,j) - \mu_{\text{enc}}]}{\sqrt{\left\{\sum\limits_{i=1}^{M}\sum\limits_{j=1}^{N}[I_{\text{orig}}(i,j) - \mu_{\text{orig}}]^2\right\} \left\{\sum\limits_{i=1}^{M}\sum\limits_{j=1}^{N}[I_{\text{enc}}(i,j) - \mu_{\text{enc}}]^2\right\}}},
\end{equation*}
where $I_{\text{orig}}(i,j)$ and $I_{\text{enc}}(i,j)$ denote,
respectively, the pixel intensities at coordinate $(i,j)$ 
of the original and the encrypted image matrices, $M$ 
and $N$ are the image dimensions, and $\mu_{\text{orig}}$ and $\mu_{\text{enc}}$ represent the global mean pixel intensities of $I_{\text{orig}}$ and $I_{\text{enc}}$, respectively.

\vspace{-10pt}
\begin{table}[H]
\caption{Internal Correlations Between  Original 
	Image and  Encrypted Image.}
\centering
\includegraphics[width=1\linewidth]{ 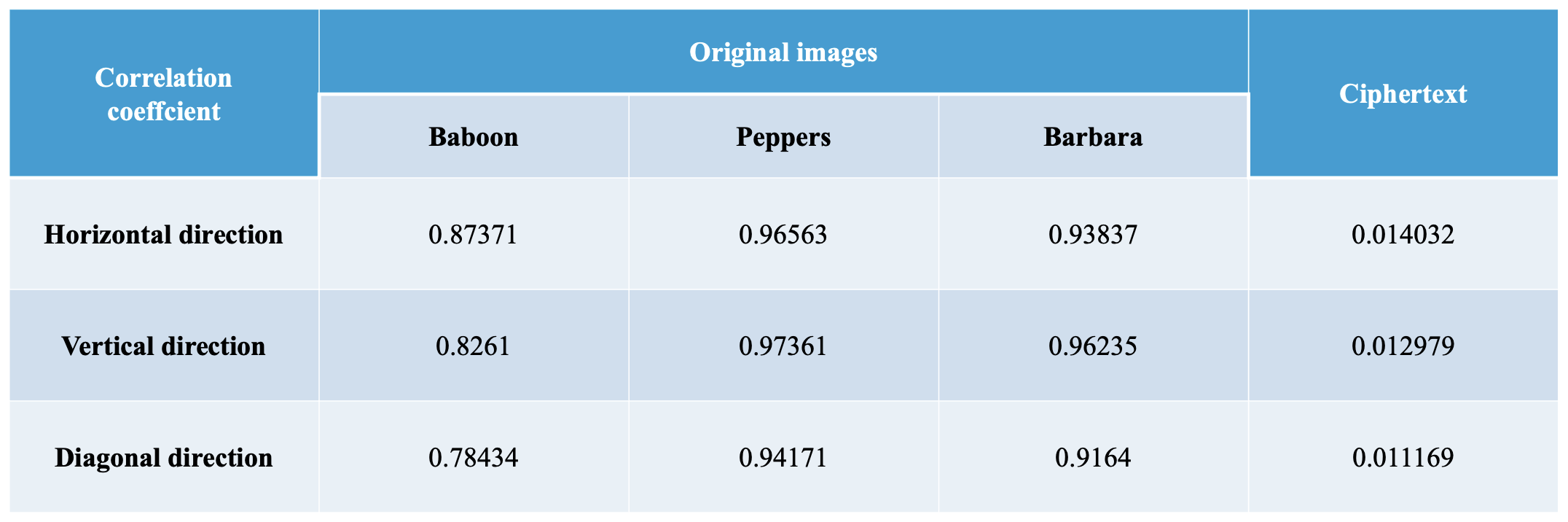}	
\vspace{-20pt}
\label{tab.421}
\end{table}

 \vspace{-10pt}
\begin{table}[H]
\caption{Internal Correlations in Erroneously 
Decrypted Image.}
	\centering
\subfigure[Internal Correlations: Decryption with 
Incorrect $A_j$.]
{\includegraphics[width=1\linewidth]
	{ 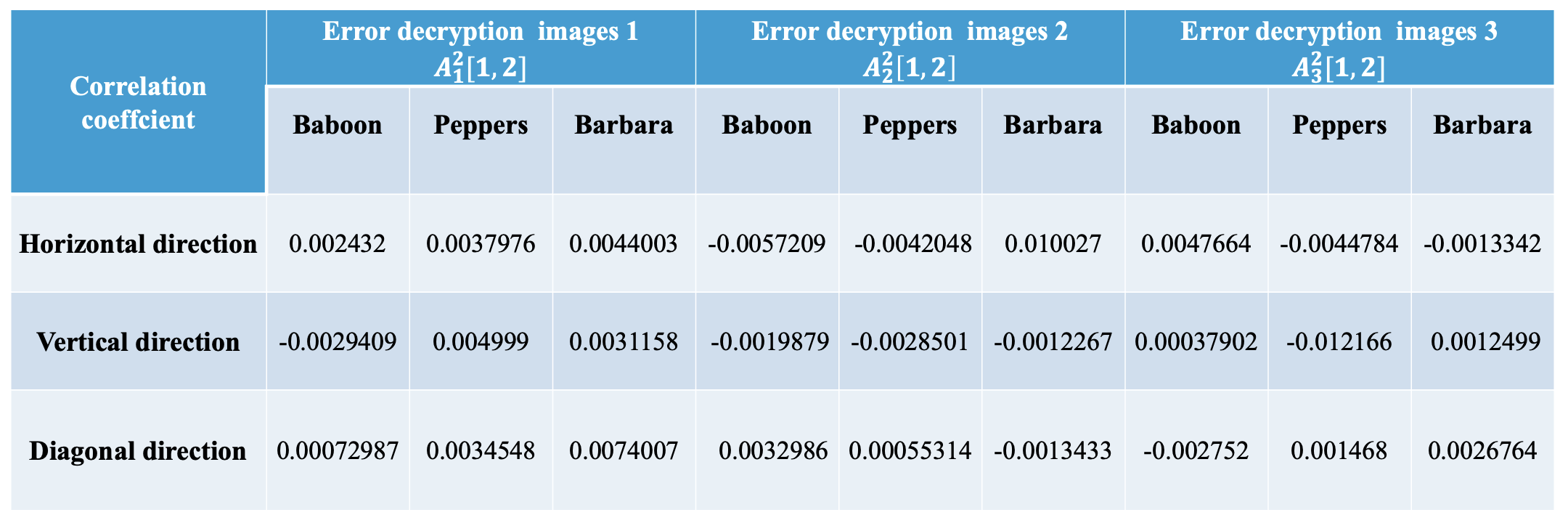}}	
\vspace{-20pt}
\subfigure[Internal Correlations: Decryption with 
Incorrect  $\beta_j$.]
{\includegraphics[width=1\linewidth]
	{ 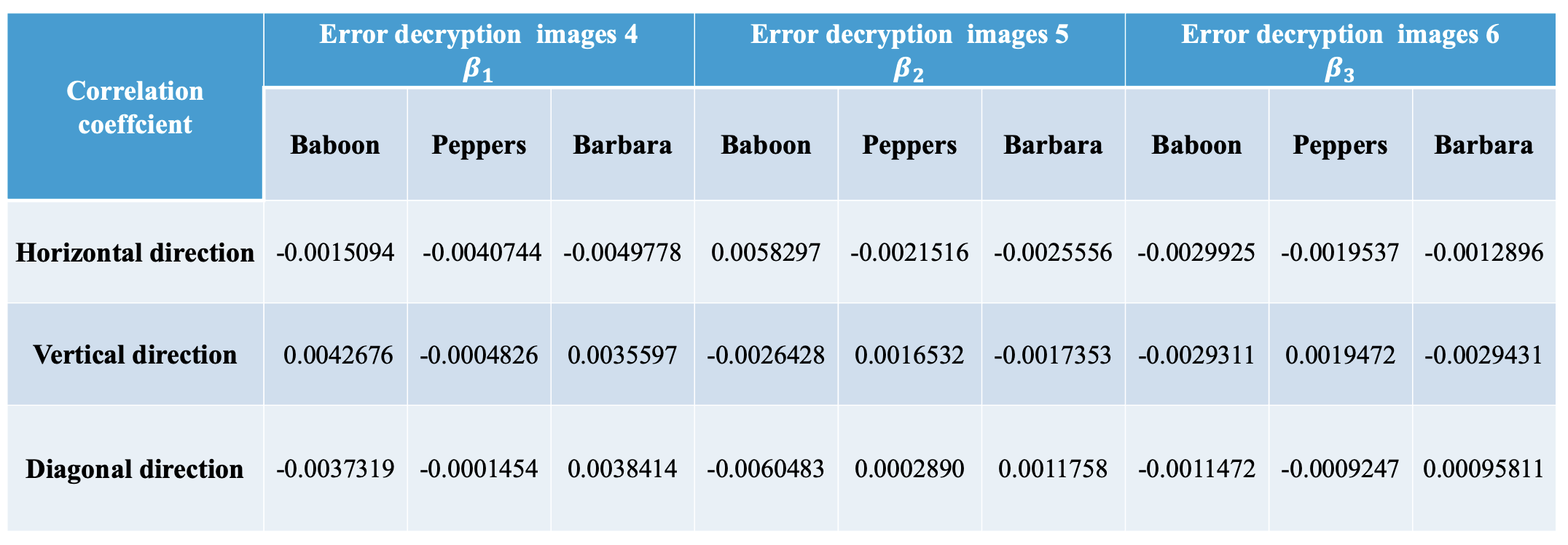}}	
\vspace{5pt}
\label{tab.422}
\end{table}

\begin{table}[H]
\caption{Global Correlations Between Original 
Image and  Encrypted Image.} 
\centering
\includegraphics[width=1\linewidth]{ 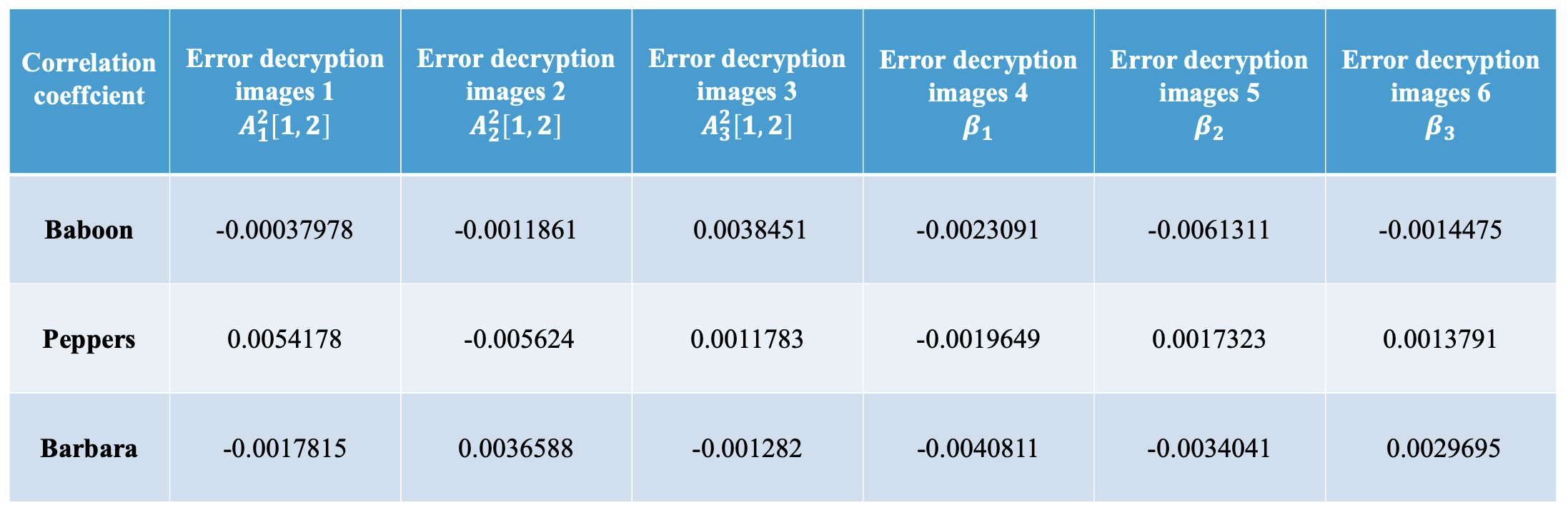} 
\vspace{-20pt} 
\label{tab.424}
\end{table}
The results in Table \ref{tab.424} confirm that any 
deviation in the matrix parameters $A_j$ or the orders 
$\beta_j$ during decryption results in a decrypted image
that is statistically indistinguishable from random noise, 
with correlation coefficients consistently approaching 
zero. This indicates a complete breakdown of structural
information, regardless of which specific key 
parameter is altered.

\begin{figure}[H]
\centering	
\subfigure[]{
\begin{minipage}{0.28\linewidth}
\centering
{\footnotesize Cameraman} \\[0.1cm]
\includegraphics[width=\linewidth]{ 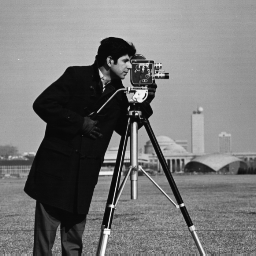}\vspace{0.2cm}
\end{minipage}}
\subfigure[]{
\begin{minipage}{0.28\linewidth}
\centering
{\footnotesize Goldhill} \\[0.1cm]
\includegraphics[width=\linewidth]{ 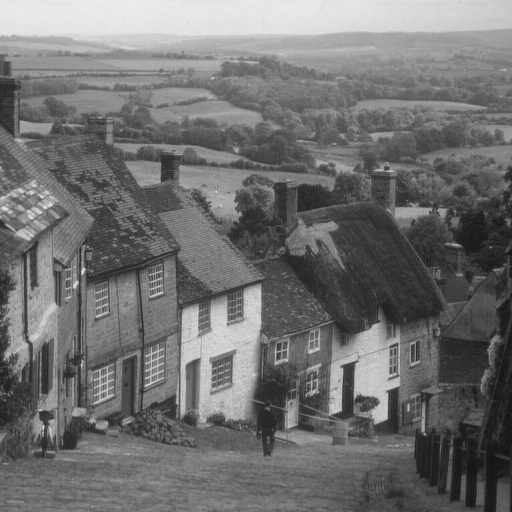}\vspace{0.2cm}
\end{minipage}
}
\subfigure[]{
\begin{minipage}{0.28\linewidth}
\centering
{\footnotesize Couple} \\[0.1cm]
\includegraphics[width=\linewidth]{ 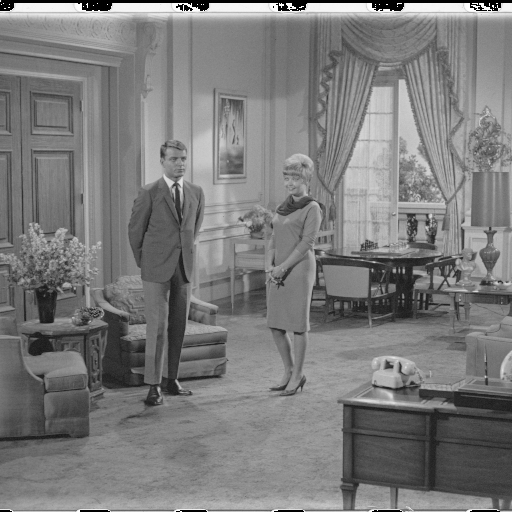}\vspace{0.2cm}
\end{minipage}
}
\vspace{-0.2cm}
\caption{New Original Images to be Encrypted.}
\label{fig.3ys}
\end{figure}

\begin{figure}[H]
\centering
\includegraphics[width=0.9\linewidth]{ 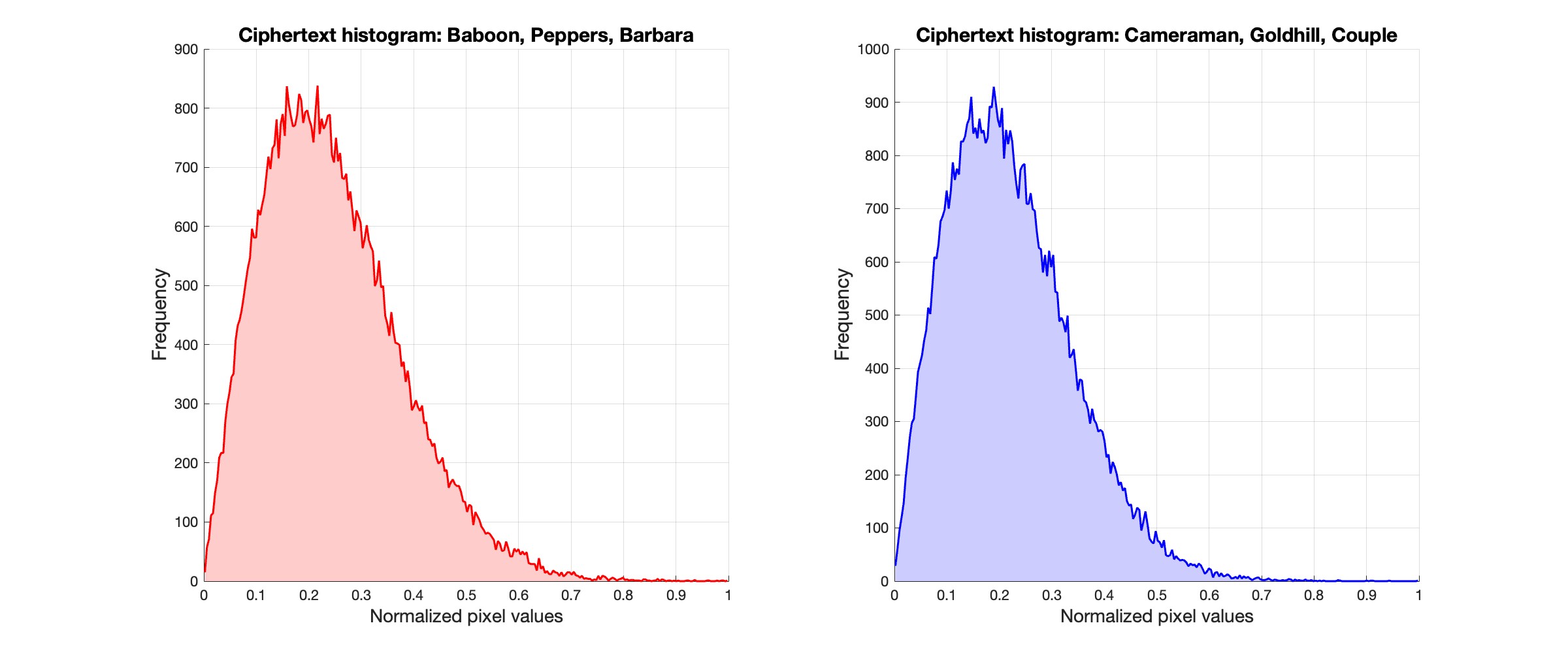}
\caption{Histograms of Encrypted Images.}
\label{fig. zhifangtu}
\end{figure}
 \vspace{-0.5cm} 	
 
 Figure \ref{fig.3ys} presents a set of new original images used for encryption.  Figure \ref{fig. zhifangtu} are the histograms of ciphertext images produced by applying the MIE-AC-LCRP method to two distinct sets of original images (namely Figures \ref{fig.421} and \ref{fig.3ys}). The horizontal and the vertical axes respectively correspond to normalized pixel values  and frequencies of their occurrences. Both subfigures reveal highly uniform distribution patterns, and these two different ciphertext histograms display nearly identical profiles. These results confirm the efficacy of the encryption algorithm in obscuring statistical information from the original images. Furthermore, the highly uniform distribution patterns of histograms observed in Figure \ref{fig. zhifangtu} can be attributed, as revealed by Theorem \ref{rem-imp-1} and Remarks \ref{zyr-1} and \ref{the-spaces}, to the intrinsic oscillatory factors of the two‑parameter chirp functions of the LCRP, which induce strong phase cancellation. When further cascaded with random phase masks, this oscillatory regularization completely obliterates the original statistical intensity distribution, yielding a highly uniform and noise‑like histogram.
 
All the above experimental results demonstrate that the proposed MIE-AC-LCRP method possesses excellent statistical security. The encryption process, leveraging the combined effect of LCRPs and random phase masks, thoroughly disrupts the statistical correlations of the image at both internal and global levels. This is supported by two key evidences: the highly uniform patterns and the nearly identical profiles of different ciphertext histograms, and the effective disruption of pixel correlations. The transformation of diverse input images into ciphertexts with these noise-like characteristics ensures that the original image can be perfectly reconstructed only with the correct keys, while any minor deviation in the key parameters during decryption will yield the results statistically indistinguishable from random noise. This fulfillment of Kerckhoffs's principle and the confirmation of strong diffusion characteristics collectively validate the encryption scheme's effectiveness and its robust resistance against statistical cryptanalysis.
 
\subsubsection{Noise Attack Analysis}\label{3.3.3}
To assess the robustness of the proposed  MIE-AC-LCRP method  against noise attacks, we perturb the ciphertext image $M$ by adding the Gaussian random noise, which is denoted by $M'$, namely $M' := M + \lambda G,$ 
where $\lambda$ is the noise strength coefficient and $G$ 
represents the standard Gaussian random noise. Figure
\ref{fig. jxzs}(a)--(f) present the original Barbara 
image together with its decrypted counterparts, respectively, under 
increasing noise intensities $\lambda = 0.2, 0.4, 
0.6, 0.8, 1.0$. The results show that the decrypted 
images remain visually recognizable even at high 
noise levels, demonstrating the scheme's considerable 
resilience. This consistent performance is also  observed 
across other test images, further validating the 
robustness of the MIE-AC-LCRP method.

\vspace{-0.3cm} 	
\begin{figure}[H]
	\includegraphics[width=0.9\linewidth]{ 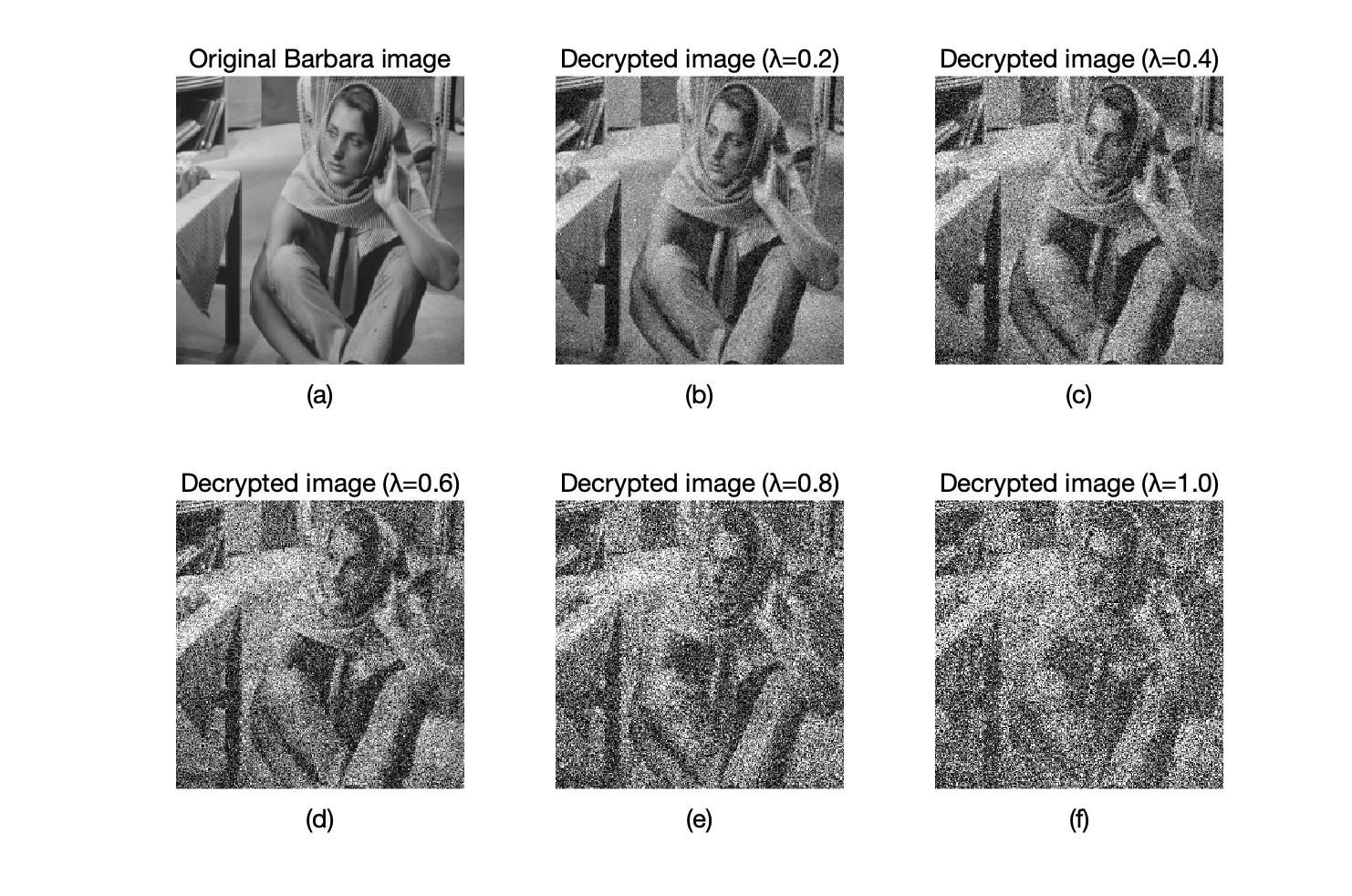}
	\vspace{-1cm} 	
	\caption{Noise Attack Analysis.}
	\label{fig. jxzs}
\end{figure}

The remarkable robustness of the proposed MIE-AC-LCRP method against noise attacks analytically stems from the oscillatory regularization mechanism established in Theorem \ref{rem-imp-1} and Remarks \ref{zyr-1} and \ref{the-spaces}. During decryption, the encrypted image reconstructs coherently only via precise phase matching relying on the correct keys. Conversely, uncorrelated random noise undergoes severe phase cancellations when encountering the LCLO, which contains highly oscillatory factors. This regularization mechanism evenly dissipates the noise energy across the spatial domain, thereby preserving the macroscopic structural features of the image even under noise interference.

\subsubsection{Occlusion Attack Analysis}\label{3.3.4}
Beyond noise resilience, the system's robustness against occlusion attacks which simulate partial data loss in ciphertext, is also systematically examined. This evaluation involves applying simulated occlusions to specific regions of the ciphertext images (top-left corner, bottom-left corner, and left side, as shown in Figures \ref{fig. oa}(d)--(f)), where the lost data are replaced by zeros. Decryption using the complete correct key set yields reconstructed ``Barbara" images, Figures \ref{fig. oa}(h)--(j), that retain sufficient structural information to clearly identify the primary features of the original image. These results confirm the cryptosystem's strong robustness against occlusion attacks, and a consistency holds across all other test images.

\vspace{-15pt}
\begin{figure}[H]
	\includegraphics[width=0.9\linewidth]{ 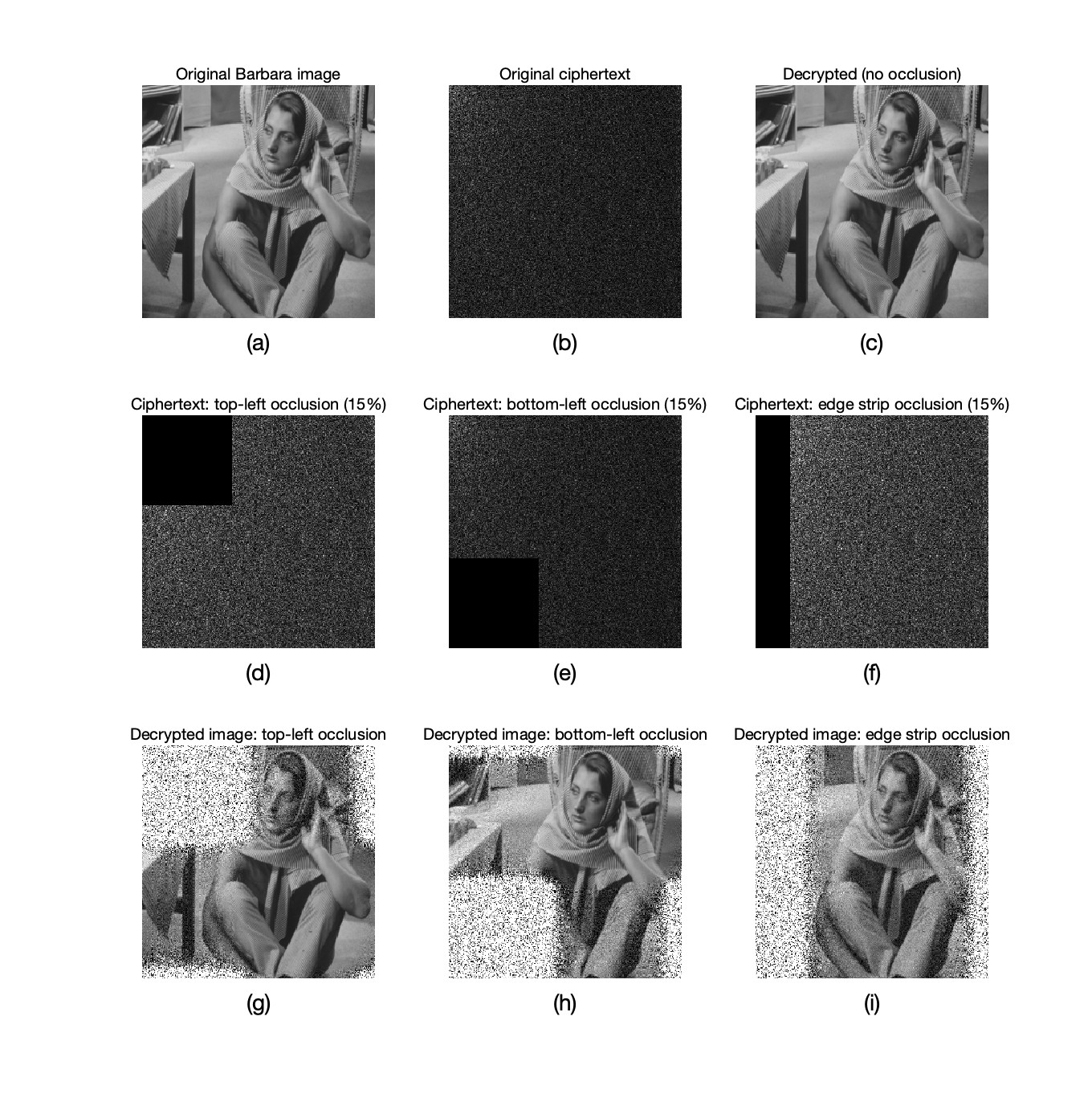} 	
	\vspace{-1cm} 	
	\caption{Occlusion Attack Analysis.}
	\label{fig. oa}
\end{figure}
\vspace{-12pt} 

The robustness of the MIE-AC-LCRP method against occlusion attacks mathematically stems from the non-local nature of the LCRP and the destructive interference mechanism of its oscillatory factors (see Theorem \ref{rem-imp-1} and Remarks \ref{zyr-1} and \ref{the-spaces}). During the encryption phase, the diffusion effect of the LCRP ensures that local occlusions do not destroy the core features of the image. During decryption, the occluded region in the ciphertext, acting as a discontinuous bounded perturbation (see Theorem \ref{thm-polygon}), undergoes severe destructive interference under the action of the LCLO due to phase mismatch. Consequently, the local perturbation energy is diffused and confined to specific spatial regions, manifesting as noise, which thereby preserves the macroscopic structure of the original image.

\section{Conclusions}\label{sec5}
This article introduces the LCRP related to the chirp Poisson equation, formulates its symbol in terms of the LCT, and establishes the convergence and the divergence properties of these LCRPs on different function classes. 
Notably, our analysis uncovers the intrinsic analytic mechanisms acting on fundamental image spaces.  For bounded non-decaying spaces (representing global image backgrounds), the LCRP utilizes the oscillatory factors of chirp functions to guarantee its convergence, overcoming the divergence issues of the classical Riesz potential. This transition from the ``incapability'' of the classical operator to the ``capability'' of the LCRP enables the successful processing of full-frame images with non-zero backgrounds. 
Furthermore, for piecewise smooth spaces (representing image edges), the limiting behavior of the classical Riesz potential exhibits discontinuities near critical values due to a mismatch between global normalization and local sectorial truncation geometry. 
This analytic instability ensures that any slight deviation in key parameters disrupts phase alignment and destroys edge reconstruction, providing a rigorous theoretical foundation for the cryptosystem's extreme key sensitivity. 
Building upon this theoretical framework, the inverse operator of the LCRP, namely the LCLO, is also defined.

As applications, we propose the MIE-AC-LCRP method for multi-image encryption, constructing an efficient and secure asymmetric cryptosystem. 
This system achieves the effective fusion and encryption of multi-images through cascaded LCRP processing and nonlinear phase modulation. Systematic security evaluations, including key sensitivity, statistical, noise attack, and occlusion attack analyses, validate the method's remarkable security and strong robustness. Even for single-image scenarios, the proposed method demonstrates higher efficiency and greater security, compared to the known encryption approach based on the fractional Riesz potential.

\bigskip

\noindent Zunwei Fu

\smallskip

\noindent School of Mathematics and Statistics, Linyi University,
Linyi 276000, The People's Republic of China; College of Information
Technology, The University of Suwon, Hwaseong-si 18323,
South Korea

\smallskip

\noindent{\it E-mail:} \texttt{zwfu@suwon.ac.kr}

\bigskip

\noindent  Dachun Yang (Corresponding author) and  Shuhui Yang

\smallskip

\noindent  Laboratory of Mathematics and Complex
Systems (Ministry of Education of China),
School of Mathematical Sciences, Institute for Advanced Study, Beijing Normal
University, Beijing 100875, The People's Republic of China

\smallskip

\noindent{\it E-mails:} \texttt{dcyang@bnu.edu.cn} (D. Yang)

\noindent\phantom{{\it E-mails:}} \texttt{shuhuiyang@bnu.edu.cn} (S. Yang) 

 \end{document}